\documentclass[envcountsect,orivec]{llncs} 
\usepackage{etex} 
\usepackage[]{graphicx}
\newif\ifignore 
\ignorefalse
\newcommand{\auxproof}[1]{
  \ifignore\mbox{}\newline
  \textbf{BEGIN: AUX-PROOF} \dotfill\newline
  {#1}\mbox{}\newline
  \textbf{END: AUX-PROOF}\dotfill\newline
  \fi}

\renewcommand{\marginpar}[1]{}    

\usepackage{times}

\usepackage{theorem}
\usepackage{fancybox,amssymb,amstext,amsmath,stmaryrd,wasysym,cite,proof,mathtools,multirow,stackrel}
\SetSymbolFont{stmry}{bold}{U}{stmry}{m}{n}
\SetSymbolFont{wasy}{bold}{U}{wasy}{m}{n}

\usepackage{algorithm}
\usepackage[noend]{algpseudocode} 
\usepackage{etoolbox}

\makeatletter
\newcommand*{\algrule}[1][\algorithmicindent]{%
   \makebox[#1][l]{%
       \hspace*{.2em}
       \vrule height .75\baselineskip depth .25\baselineskip
   }
}

\newcount\ALG@printindent@tempcnta
\def\ALG@printindent{%
    \ifnum \theALG@nested>0
    \ifx\ALG@text\ALG@x@notext
    \else
    \unskip
    \ALG@printindent@tempcnta=1
    \loop
    \algrule[\csname ALG@ind@\the\ALG@printindent@tempcnta\endcsname]%
    \advance \ALG@printindent@tempcnta 1
    \ifnum \ALG@printindent@tempcnta<\numexpr\theALG@nested+1\relax
    \repeat
    \fi
    \fi
}
\patchcmd{\ALG@doentity}{\noindent\hskip\ALG@tlm}{\ALG@printindent}{}{\errmessage{failed to patch}}
\patchcmd{\ALG@doentity}{\item[]\nointerlineskip}{}{}{} 
\makeatother
\newcommand{\Break}{\textbf{break}\; }

\usepackage{wrapfig}
\theorembodyfont{\itshape}
\newtheorem{mytheorem}{Theorem}[section]

\theorembodyfont{\rmfamily}

\newtheorem{mydefinition}[mytheorem]{Definition}
\newtheorem{myassumption}[mytheorem]{Assumption}

\spnewtheorem*{myproof}{Proof}{\itshape}{\rmfamily}

\usepackage{xcolor}
\definecolor{dgreen}{rgb}{0, .6, 0}

\usepackage{pifont}
\newcommand{\cmark}{\ding{51}}%
\newcommand{\xmark}{\ding{55}}%

\newcommand{\Z}{{\mathbb{Z}}}
\newcommand{\Zp}{{\mathbb{Z}_{>0}}}
\newcommand{\R}{{\mathbb{R}}}
\newcommand{\Rp}{{\mathbb{R}_{>0}}}

\newcommand{\ttrue}{\mathrm{t{\kern-1.5pt}t}}
\newcommand{\ffalse}{\mathrm{f{\kern-1.5pt}f}}

\newcommand{\Rnn}{\R_{\ge 0}}

\newcommand{\SG}{\mathit{SG}^a}

\makeatletter
\newcommand{\figcaption}[1]{\def\@captype{figure}\caption{#1}}
\newcommand{\tblcaption}[1]{\def\@captype{table}\caption{#1}}
\makeatother

\newif\iftikzgnuplot
\tikzgnuplottrue 
\usepackage[hyphens]{url}
\usepackage{gnuplot-lua-tikz}
\usepackage{pgfplotstable} 
\pgfplotsset{compat=1.12}
\usepackage{booktabs}
\usepackage{tabularx}

\usepackage{arydshln}

\usepackage{enumitem}
\setlistdepth{20}
\renewlist{itemize}{itemize}{20}
\setlist[itemize]{label=\textbullet}

\newcommand{\Opt}{\mathit{Opt}}

\newcommand{\str}{w}
\newcommand{\pat}{\mathit{pat}}
\newcommand{\reset}{\mathrm{reset}}
\newcommand{\eval}{\mathrm{eval}}

\newcommand{\solConstr}{\mathrm{solConstr}}
\newcommand{\CurrConf}{\mathit{CurrConf}}

\newcommand{\NextConf}{\mathit{NextConf}}

\newcommand{\Conf}{\mathit{Conf}}

\newcommand{\rhoEmpty}{\rho_{\emptyset}}

\newcommand{\regionstate}[4][]{
 \node[state,region] (#2) [#1] {
 \begin{tabular}{c}
  #3\\
  #4\\
 \end{tabular}
 }}

\usepackage{tikz}
\usetikzlibrary{automata,positioning,matrix,shapes.callouts}
\tikzset{
region/.style={
rectangle,
rounded corners,
draw=black,very thick
},
accepting/.style={double distance=2pt}
}

\makeatletter
\newcommand{\Biggg}{\bBigg@{4}}
\newcommand{\Bigggg}{\bBigg@{5}}
\newcommand{\Biggggg}{\bBigg@{6}}
\makeatother

 \title{Efficient Online Timed Pattern Matching by Automata-Based Skipping}

 \author{
Masaki Waga
\inst{1}
  \and
Ichiro Hasuo
\inst{2}
  \and
Kohei Suenaga
\inst{3}
 }
 \institute{
     University of Tokyo, Tokyo, Japan
     \and
     National Institute of Informatics, Tokyo, Japan
     \and
    Kyoto University and JST PRESTO, Kyoto, Japan
}

\begin{document}
\maketitle

\begin{abstract}
The \emph{timed pattern matching} problem is an actively studied topic
 because of its relevance in \emph{monitoring} of real-time systems. There
 one is given a log $w$ and a specification $\mathcal{A}$ (given by a
 \emph{timed word} and a \emph{timed automaton} in this paper), and one
 wishes to return the set of intervals for which the log $w$, when
 restricted to the interval, satisfies the specification
 $\mathcal{A}$. In our previous work we presented an efficient timed pattern matching
 algorithm: it adopts a skipping mechanism inspired by the classic Boyer--Moore
 (BM) string matching algorithm. In this work we tackle the problem of \emph{online} timed pattern matching, towards embedded applications where it
 is vital to process a vast amount of incoming data in a timely
 manner. Specifically, we start with the Franek-Jennings-Smyth (FJS)
 string matching algorithm---a recent  variant of the BM
 algorithm---and extend it to timed pattern matching. Our
 experiments indicate the  efficiency of our FJS-type algorithm
 in online and offline timed pattern matching.
\end{abstract}

\section{Introduction}
\label{sec:intro}
\emph{Monitoring} of real-time properties is an actively studied topic with
numerous applications such as
automotive systems~\cite{DBLP:conf/rv/KaneCDK15},
 medical systems~\cite{DBLP:journals/jcse/ChenSWL16},
data classification~\cite{DBLP:conf/hybrid/BombaraVPYB16},
web service~\cite{DBLP:conf/sigsoft/RaimondiSE08},
and quantitative performance
measuring~\cite{DBLP:conf/cav/FerrereMNU15}. 
\marginpar{what is ``quantitative performance
measuring''?}
Given a specification $\mathcal{A}$ and a log $w$ of activities, monitoring would ask questions like: 
\emph{if $w$ has a segment that matches $\mathcal{A}$}; \emph{all the segments of $w$ that match $\mathcal{A}$}; and so on. 

For a monitoring algorithm \emph{efficiency} is a critical matter. Since we
often need to monitor a large number of logs, each of which tends to be
very long, one monitoring task can take hours. Therefore even \emph{constant} speed up can make  significant practical differences. Another important issue is an algorithm's performance in \emph{online usage scenarios}. Monitoring algorithms are often deployed in \emph{embedded} applications~\cite{kane2015runtime}, and this incurs the following online requirements: 
\begin{itemize}
 \item \emph{Real-time properties}, such as: on prefixes of the log $w$, we want to know their monitoring result soon, possibly before the whole log $w$ arrives.
 \item  \emph{Memory consumption}, such as: early prefixes of $w$ should not affect the monitoring task of later segments of $w$, so that we can throw the prefixes away and free memory (that tends to be quite limited in embedded applications).
 \item \emph{Speed} of the algorithm. In an online setting this means: if the log $w$  arrives at a speed faster than the algorithm processes it, then the data that waits to be processed will fill up the memory. 
\end{itemize}
Constant improvement in aspects like speed and memory consumption will be appreciated in online settings, too: if an algorithm is twice as fast, then this means the same monitoring task can be conducted with cheaper hardware that is twice slower. 

The goal of the current paper is thus monitoring algorithms that perform well both in offline and online settings. We take a framework where \emph{timed words}---they are
essentially sequences of time-stamped events---stand for logs, and  \emph{timed automata} express a specification. Both constructs are well-known in the community of real-time systems. The problem we solve is that of \emph{timed pattern matching}: see~\S{}\ref{subsec:timedPatternMatching} for its definition; Fig.~\ref{fig:input_timed_pattern_matching} for an example; and Table~\ref{table:matching} for comparison with other matching problems.

\begin{figure}[t]
\begin{minipage}{0.6\textwidth} 
 \centering
 \scalebox{0.7}{
  \begin{tikzpicture}[shorten >=1pt,node distance=2cm,on grid,auto]
 \node[state,initial] (s_0) {$s_0$};
 \node[state] (s_1) [right of=s_0] {$s_1$};
 \node[state] (s_2) [right of=s_1] {$s_2$};
 \node[state] (s_3) [right of=s_2]{$s_3$};
 \node[state,accepting] (s_4) [right of=s_3]{$s_4$};

 \path[->] 
  (s_0) edge [above] node {\begin{tabular}{c}
                            $\text{a},x > 1$\\
                            $/ x := 0$
                           \end{tabular}} (s_1)
  (s_1) edge [above] node {\begin{tabular}{c}
                            $\text{a},x < 1$\\
                            $/x := 0$
                           \end{tabular}} (s_2)
  (s_2) edge [above] node {$\text{a}, x < 1$} (s_3)
  (s_3) edge [above] node {$\$,\mathbf{true}$} (s_4);
 \end{tikzpicture}}
 \end{minipage}
\begin{minipage}{0.4\textwidth}
 \centering
 \scalebox{0.7}{
 \begin{tikzpicture} 
  \draw [thick, -stealth](-0.5,0)--(6.5,0) node [anchor=north]{$t$};
  \draw (0,0.1) -- (0,-0.1) node [anchor=north]{$0$};

 \draw (0.5,0.1) node[anchor=south]{$\text{a}$} -- (0.5,-0.1) node[anchor=north]{$0.5$};
 \draw (0.9,0.1) node[anchor=south]{$\text{a}$} -- (0.9,-0.1) node[anchor=north]{$0.9$};
 \draw (1.3,0.1) node[anchor=south]{$\text{b}$} -- (1.3,-0.1) node[anchor=north]{$1.3$};
 \draw (1.7,0.1) node[anchor=south]{$\text{b}$} -- (1.7,-0.1) node[anchor=north]{$1.7$};
 \draw (2.8,0.1) node[anchor=south]{$\text{a}$} -- (2.8,-0.1) node[anchor=north]{$2.8$};
 \draw (3.7,0.1) node[anchor=south]{$\text{a}$} -- (3.7,-0.1) node[anchor=north]{$3.7$};
 \draw (5.3,0.1) node[anchor=south]{$\text{a}$} -- (5.3,-0.1) node[anchor=north]{$5.3$};
 \draw (4.9,0.1) node[anchor=south]{$\text{a}$} -- (4.9,-0.1) node[anchor=north]{$4.9$};
 \draw (6.0,0.1) node[anchor=south]{$\text{a}$} -- (6.0,-0.1) node[anchor=north]{$6.0$};
 \end{tikzpicture}}
 \end{minipage}
 \caption{An example of timed pattern matching. For the pattern timed automaton $\mathcal{A}$ and the target timed word $w$, as shown, the output is the set of matching intervals $\{(t,t')\mid w|_{(t,t')}\in L(\mathcal{A})\}=\{(t,t') \mid t \in [3.7,3.9), t' \in (6.0,\infty)\}$. Here \$ is a special terminal character.}
 \label{fig:input_timed_pattern_matching}
\end{figure}
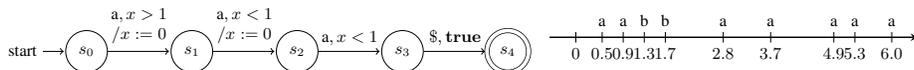

\begin{table}[t]
\centering
\caption{Matching problems}\label{table:matching}
\scalebox{.8}{\begin{tabular}{c||c|c|c}
&log, target&specification, pattern &output\\\hline\hline
string matching&
a word $\str\in \Sigma^{*}$
&
a word $\pat\in\Sigma^{*}$
&
$\{(i,j)\in(\Zp)^{2}\mid \str(i,j)=\pat\}$
\\\hline
pattern matching&
a word $\str\in \Sigma^{*}$
&
an NFA $\mathcal{A}$
&
$\{(i,j)\in(\Zp)^{2}\mid \str(i,j)\in L(\mathcal{A})\}$
\\\hline
timed pattern matching&
a timed word $\str\in(\Sigma \times \Rp)^{*}$
&
a timed automaton $\mathcal{A}$
&
$\{(t,t') \in (\Rp)^{2}\mid w|_{(t,t')} \in L (\mathcal{A})\}$
\\\hline
\end{tabular}
}
\end{table}

Towards the goal our strategy is to exploit the idea of \emph{skip values} in  efficient string matching algorithms (such as Boyer--Moore (BM)~\cite{Boyer1977}), together with their \emph{automata-based extension} for pattern matching by Watson \& Watson~\cite{Watson2003}, to skip unnecessary matching trials. In our previous work~\cite{DBLP:conf/formats/WagaAH16} we took the strategy and introduced a timed pattern matching algorithm with BM-type skipping. The current work improves on this previous BM algorithm: it is based on the more recent \emph{Franek--Jennings--Smyth (FJS) algorithm}~\cite{DBLP:journals/jda/FranekJS07} for string matching (instead of BM); and our new algorithm is faster than 
 our previous BM-type one. Moreover, in online usage, our FJS-type algorithm better addresses the online requirements that we listed in the above. This is in contrast with our previous BM-type algorithm that works necessarily in an offline manner (it must wait for the whole log $w$ before it starts).

%


\vspace{0em}
\noindent\textbf{Contributions}\quad
Our main contribution is an efficient algorithm for timed pattern matching that employs (an automata-theoretic extension of) skip values from the Franek--Jennings--Smyth (FJS)  algorithm for string matching~\cite{DBLP:journals/jda/FranekJS07}. By experiments we show that the algorithm generally outperforms a brute-force one and our previous BM algorithm~\cite{DBLP:conf/formats/WagaAH16}: it is twice as fast for some realistic automotive examples. Through our theoretical analysis as well as experiments on memory consumption, we claim that our algorithm is suited for online usage scenarios, too. We also compare its performance with a recent tool \emph{Montre} for timed pattern matching~\cite{DBLP:journals/corr/Ulus16}, and observe that ours is faster, at least in terms of the  implementations currently available. 

In its course we have obtained an FJS-type algorithm for \emph{untimed} pattern matching, which is one of the  main contributions too. The algorithm  is explained rather in detail, so that it paves the way to our FJS-type \emph{timed} pattern matching that is more complex.

A central theme of the paper is benefits of the formalism of \emph{automata}, a mathematical tool whose use is nowadays widespread in fields like temporal logic, model checking, and so on. We follow Watson \& Watson's idea of extending skipping from string matching to pattern matching~\cite{Watson2003}, where the key is overapproximation of words and languages by states of automata. Our main contribution on the conceptual side is that the same idea applies to \emph{timed} automata as well, where we rely on \emph{zone}-based abstraction (see e.g.~\cite{DBLP:conf/tacas/BehrmannBFL03,DBLP:journals/sttt/BehrmannBLP06,DBLP:conf/cav/HerbreteauSW10}) for computing reachability.

\vspace{0em}
\noindent\textbf{Related Works}\quad
Several algorithms have been proposed for online monitoring of real-time temporal logic specifications.
 An online monitoring
algorithm for ptMTL (a past time fragment of MTL) is 
in~\cite{DBLP:journals/fmsd/ReinbacherFB14} and
an algorithm for MTL[U,S] (a variant of MTL with both forward and
backward temporal modalities) is  in~\cite{DBLP:conf/rv/HoOW14}.
In addition, a case study on an autonomous research vehicle
monitoring~\cite{DBLP:conf/rv/KaneCDK15} shows such procedures can be
performed in an actual vehicle---this is where our motivation comes from, too. 

We have chosen timed automata as a specification formalism. This is because of their expressivity as well as various techniques that operate on them.
Some other formalisms can be translated to timed automata, and via translation,
our algorithm offers to these formalisms 
 an online monitoring algorithm.
In~\cite{Asarin2002}, a variant of \emph{timed regular expressions (TREs)} are
proved to have the same expressive power as timed automata.
For MTL and MITL, transformations into automata are introduced for many 
different settings; see
e.g.~\cite{DBLP:conf/focs/AlurH92,DBLP:conf/formats/MalerNP06,DBLP:conf/formats/NickovicP10,DBLP:conf/formats/KiniKP11,d2013clock}.

The work with closest interests to ours is by Ulus, Ferr\`ere, Asarin,  Maler  and their colleagues~\cite{DBLP:conf/formats/UlusFAM14,DBLP:conf/tacas/UlusFAM16,DBLP:journals/corr/Ulus16}. In their series of work, logs are presented by \emph{signals}, i.e.\ values that vary over time. Their logs are thus \emph{state-based} rather than \emph{event-based} like timed words. Their specification formalism is timed regular expressions (TREs).
An offline monitoring algorithm is presented in~\cite{DBLP:conf/formats/UlusFAM14} and an online
one
is  in~\cite{DBLP:conf/tacas/UlusFAM16}.
These algorithms are implemented in the tool \emph{Montre}~\cite{DBLP:journals/corr/Ulus16}, with which we conduct performance comparison. The difference between different specification formalisms (TREs, timed automata, temporal logics, etc.) are subtle, but for many realistic examples the difference does not matter. In the current paper we exploit various operations on automata, most notably zone-based abstraction.

%

\vspace{0em}
\noindent
\textbf{Notations}\quad
\label{subsec:notations}
Let $\Sigma$ be an alphabet and $w=a_{1}a_{2}\dotsc a_{n}\in \Sigma^{*}$
be a string over $\Sigma$, where $a_{i}\in\Sigma$ for each $i\in
[1,n]$. We let $w(i)$ denote the $i$-th character $a_{i}$ of
$w$. Furthermore, for $i,j\in [1,n]$, when $i\le j$ we let $w(i,j)$
denote the substring $a_{i}a_{i+1}\dotsc a_{j}$, otherwise  we let
$w(i,j)$ denote the empty string $\varepsilon$.
 The length $n$ of the string $w$ is denoted by $|w|$. 


\vspace{0em}
\noindent\textbf{Organization of the Paper}\quad
In \S\ref{sec:preliminaries}  are preliminaries on: our formulation of the problem of timed pattern matching; and the FJS algorithm for string matching. 
 The FJS-type skipping is extended to (untimed) pattern matching
in \S\ref{sec:fjs_pattern_matching}, where we describe the algorithm in detail. This paves the way to our FJS-type timed pattern matching algorithm in \S{}\ref{sec:timedFJS}. In \S{}\ref{sec:timedFJS} we also sketch zone-based abstraction of timed automata, a key technical ingredient in the algorithm. In \S{}\ref{sec:experiments} we present our experiment results. They indicate our algorithm's performance advantage in both offline and online usage scenarios. 

\section{Preliminaries}
\label{sec:preliminaries}

\subsection{Timed Pattern Matching}
\label{subsec:timedPatternMatching}
Here we formulate our problem. Our target strings are \emph{timed words}~\cite{Alur1994}, that are time-stamped words over an alphabet $\Sigma$.
 Our patterns are given by 
\emph{timed automata}~\cite{Alur1994}. 


\begin{mydefinition}[timed word, timed word segment]
 For an alphabet $\Sigma$, a \emph{timed word} is a
 sequence $w$ of pairs $(a_i,\tau_i) \in (\Sigma \times \Rp)$
 satisfying $\tau_i < \tau_{i + 1}$ for any $i \in [1,|w|-1]$.
 Let $w = (\overline{a},\overline{\tau})$ be a timed word.
 We denote the subsequence $(a_i, \tau_i),(a_{i+1},
 \tau_{i+1}),\cdots,(a_j,\tau_j)$ by $w (i,j)$.
 For $t \in \Rnn$, the \emph{$t$-shift} of $w$ is
 $(\overline{a}, \overline{\tau}) + t = (\overline{a}, \overline{\tau} +
 t)$ where
 $\overline{\tau} + t = \tau_1 + t,\tau_2 + t,\cdots, \tau_{|\tau|} + t$.
 For timed words $w = (\overline{a},\overline{\tau})$ and
 $w' = (\overline{a'},\overline{\tau'})$,
 their \emph{absorbing concatenation} is
 $w \circ w' = (\overline{a} \circ \overline{a'}, \overline{\tau} \circ
 \overline{\tau'})$ where $\overline{a} \circ \overline{a'}$ and 
 $\overline{\tau} \circ \overline{\tau'}$ are usual concatenations, and 
 their \emph{non-absorbing concatenation} is
 $w \cdot w' = w \circ (w' + \tau_{|w|})$.
 We note that the absorbing concatenation $w \circ w'$ is defined only
 when $\tau_{|w|} < \tau'_{1}$.

 For a timed word $w = (\overline{a}, \overline{\tau})$ on $\Sigma$
 and 
 $t,t' \in \R_{>0}$ satisfying $t < t'$, a \emph{timed word segment}
 $w|_{(t,t')}$ is defined by the timed word $(w (i,j) - t) \circ (\$,t')$
 on the augmented alphabet $\Sigma \sqcup \{\$\}$,
 where $i,j$ are chosen so that
 $\tau_{i-1} \leq t < \tau_i$ and
 $\tau_{j} < t' \leq \tau_{j+1}$.
 Here the fresh symbol ${\$}$ is called the \emph{terminal character}. 
\end{mydefinition}


\auxproof{
A timed automaton is a NFA equipped with clock variables.
A transition of a timed automaton has a \emph{guard} and
\emph{reset variables}:
the former is a proposition defined by a conjunction of inequalities
over clock variables that is an extra requirement of the transition;
the latter is a set of variables that are reset to zero after the
transition.
The \emph{language} of a timed automaton is defined similarly to the
one of a standard NFA.
Formally, the syntax and the semantics of timed automata are defined as
follows.
}
\begin{mydefinition}[timed automaton]
\label{def:semantics_ta}
Let $C$ be  a finite set of \emph{clock variables}, and $\Phi (C)$ denote the set of
 conjunctions of inequalities $x \bowtie c$ where $x \in C$, $c \in
 \Z_{\geq 0}$, and ${\bowtie} \in \{>,\geq,<,\leq\}$.
 A \emph{timed automaton}
 $\mathcal{A} = (\Sigma,S,S_0,C,E,F)$
is a tuple where: $\Sigma$ is an alphabet;
 $S$ is a finite set of states; $S_0 \subseteq S$ is a set of initial states; 
 $E \subseteq S \times S \times \Sigma \times \mathcal{P}(C) \times \Phi(C)$
 is a set of transitions; and $F \subseteq S$ is a set of accepting
 states. The components of a transition $(s,s',a,\lambda,\delta)\in E$ represent:  the source,  target,  action, reset variables and guard of the transition, respectively.

We define a \emph{clock valuation} 
 $\nu$ as a function $\nu: C \to \R_{\geq 0}$.
We define the \emph{$t$-shift} $\nu + t$  of a clock
 valuation $\nu$, where $t \in \R_{\geq 0}$, by  $(\nu + t) (x) = \nu (x) + t$ for any $x \in C$.
 For a timed automaton $\mathcal{A} = (\Sigma,S,S_0,E,C,F)$ and a timed
 word $w = (\overline{a},\overline{\tau})$, a \emph{run} of $\mathcal{A}$
 over $w$ is a sequence $r$ of pairs 
 $(s_i, \nu_i) \in S \times (\R_{\geq 0})^C$ satisfying the following:
 (initiation) $s_0 \in S_0$ and
 $\nu_0 (x) = 0$ for any $x \in C$; and
 (consecution) for any $i \in [1,|w|]$, there
 exists a transition $(s_{i-1}, s_i, a_i, \lambda, \delta) \in E$
 such that $\nu_{i-1} + \tau_i - \tau_{i-1} \models \delta$ and
 $\nu_i (x) = 0$  (for $x \in \lambda$) and 
 $\nu_i (x) = \nu_{i-1} (x) + \tau_i - \tau_{i-1}$ (for $x \not\in \lambda$).
 A run only satisfying the consecution condition is a \emph{path}.
 A run $r = (\overline{s},\overline{\nu})$ is \emph{accepting} if the last element $s_{|s|-1}$ of $s$ belongs to $F$. 
 The \emph{language} $L (\mathcal{A})$
 is defined to be the set
 $\{w \mid \text{ there is an accepting run of $\mathcal{A}$ over $w$}\}$ of timed words.
\end{mydefinition}



\begin{mydefinition}[timed pattern matching]\label{def:TimedPatternMatching}
Let $\mathcal{A}$
be a timed automaton, and  $\str$ be a timed word,  over a common  alphabet $\Sigma$. The \emph{timed pattern matching} problem requires
all the intervals $(t,t')$ for which the segment  $w|_{(t,t')}$ is 
 accepted by $\mathcal{A}$. That is, it requires
 the \emph{match set} $\mathcal{M}
 (w,\mathcal{A}) = \{(t,t') \mid w|_{(t,t')} \in L (\mathcal{A})\}$.
\end{mydefinition}
The match set $\mathcal{M}
 (w,\mathcal{A})
$ is in general uncountable; however it allows finitary representation, as a finite union of special polyhedra called \emph{zones}. See~\cite{DBLP:conf/formats/WagaAH16}.

\subsection{String Matching and the FJS Algorithm}
\label{subsec:stringMatching}

\emph{String matching} is a fundamental problem in computer science.
Given a \emph{pattern string} $\pat$ and a \emph{target string} $\str$, it requires
the set 
\begin{math}
 \bigl\{\,(i,j)\in(\Zp)^{2} \,\big|\, \str(i,j)=\pat\,\bigr\}
\end{math}
of all the occurrences of $\pat$ in $\str$. 
A brute-force algorithm, by trying to match $|\pat|$ characters for all the possible $|\str|-|\pat|$ positions of the pattern string, solves the string matching problem in 
$O(| \pat || \str |)$. Efficient algorithms for this classic problem  have been sought  for a long time, with significant progress made as recently as in the last decade~\cite{DBLP:journals/csur/FaroL13}. Among them the \emph{Knuth--Morris--Pratt (KMP) algorithm}~\cite{Knuth1977} and
the \emph{Boyer--Moore (BM) algorithm}~\cite{Boyer1977} are well-known, where unnecessary matching trials are \emph{skipped} utilizing  \emph{skip value functions}.
Empirical studies have shown speed advantage of BM---and its variants 
like  \emph{Quick Search}~\cite{DBLP:journals/cacm/Sunday90}---over KMP, while
theoretically  KMP
 exhibits better worst-case complexity $O(|\pat| + |\str|)$.
By combining  KMP and  Quick Search,
the \emph{Franek--Jennings--Smyth (FJS) algorithm}~\cite{DBLP:journals/jda/FranekJS07}, proposed in 2007, achieves both linear worst-case complexity and good practical performance.

The current paper's goal is to introduce FJS-like optimization to timed pattern matching. We therefore take the FJS algorithm as an example and demonstrate how skip values are utilized in the string matching algorithms we have mentioned.\footnote{The FJS-type algorithm we present here is a simplified version of the original FJS algorithm.  Our simplification is equipped with all the features that we will exploit later for pattern matching and timed pattern matching; the original algorithm further  omits some other trivially unnecessary matching trials.
We note that, because of the difference (that is conceptually inessential), our simplified algorithm (for string matching) no longer enjoys   linear worst-case complexity. }

\begin{figure}[t]
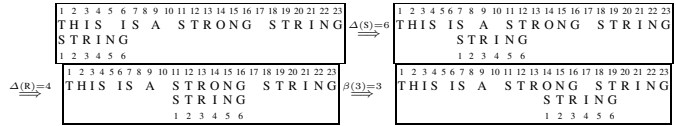

 \centering
 \setlength{\tabcolsep}{1pt}
\scriptsize
\scalebox{0.7}{
\begin{tabular}{ccc}
&
  \begin{tabular}{|ccccccccccccccccccccccc|}
  \hline
  \tiny 1&\tiny 2&\tiny 3&\tiny 4&\tiny 5&\tiny 6&\tiny 7&\tiny 8&\tiny
  9&\tiny 10&\tiny 11&\tiny 12&\tiny 13&\tiny 14&\tiny 15&\tiny 16&\tiny
  17&\tiny 18&\tiny 19&\tiny 20&\tiny 21&\tiny 22&\tiny 23\\
  T&H&I&S&\ &I&S&\ &A&\ &S&T&R&O&N&G&\ &S&T&R&I&N&G\\
  S&T&R&I&N&G &&&&& &&&&& &&&&& &&\\
  \tiny 1&\tiny 2&\tiny 3&\tiny 4&\tiny 5&\tiny 6
   &&&&& &&&&& &&&&& &&		       
 \\\hline
 \end{tabular}
$\stackrel{\Delta(\textrm{S})=6}{\Longrightarrow}$
&
\scriptsize
  \begin{tabular}{|ccccccccccccccccccccccc|}
  \hline
  \tiny 1&\tiny 2&\tiny 3&\tiny 4&\tiny 5&\tiny 6&\tiny 7&\tiny 8&\tiny
  9&\tiny 10&\tiny 11&\tiny 12&\tiny 13&\tiny 14&\tiny 15&\tiny 16&\tiny
  17&\tiny 18&\tiny 19&\tiny 20&\tiny 21&\tiny 22&\tiny 23\\
  T&H&I&S&\ &I&S&\ &A&\ &S&T&R&O&N&G&\ &S&T&R&I&N&G\\
  &&&&&& S&T&R&I&N&G &&&&& &&&&& &\\
  &&&&&&\tiny 1&\tiny 2&\tiny 3&\tiny 4&\tiny 5&\tiny 6 &&&&& &&&&& &
  \\\hline
 \end{tabular}
\\
 $\stackrel{\Delta(\textrm{R})=4}{\Longrightarrow}$
\scriptsize
     &
  \begin{tabular}{|ccccccccccccccccccccccc|}
   \hline
  \tiny 1&\tiny 2&\tiny 3&\tiny 4&\tiny 5&\tiny 6&\tiny 7&\tiny 8&\tiny
  9&\tiny 10&\tiny 11&\tiny 12&\tiny 13&\tiny 14&\tiny 15&\tiny 16&\tiny
  17&\tiny 18&\tiny 19&\tiny 20&\tiny 21&\tiny 22&\tiny 23\\
  T&H&I&S&\ &I&S&\ &A&\ &S&T&R&O&N&G&\ &S&T&R&I&N&G\\
  &&&&&&&&&& S&T&R&I&N&G  &&& &&&&\\
  &&&&&&&&&& \tiny 1&\tiny 2&\tiny 3&\tiny 4&\tiny 5&\tiny 6 &&& &&&&
   \\\hline
 \end{tabular}
 $\stackrel{\beta(3)=3}{\Longrightarrow}$
&
\scriptsize
  \begin{tabular}{|ccccccccccccccccccccccc|}
   \hline
  \tiny 1&\tiny 2&\tiny 3&\tiny 4&\tiny 5&\tiny 6&\tiny 7&\tiny 8&\tiny
  9&\tiny 10&\tiny 11&\tiny 12&\tiny 13&\tiny 14&\tiny 15&\tiny 16&\tiny
  17&\tiny 18&\tiny 19&\tiny 20&\tiny 21&\tiny 22&\tiny 23\\
  T&H&I&S&\ &I&S&\ &A&\ &S&T&R&O&N&G&\ &S&T&R&I&N&G\\
  &&&&&&&&&&&&& S&T&R&I&N&G &&&& \\
  &&&&&&&&&&&&& \tiny 1&\tiny 2&\tiny 3&\tiny 4&\tiny 5&\tiny 6 &&&&
   \\\hline
 \end{tabular}
\end{tabular}}
  \caption{The Franek--Jennings--Smyth (FJS) algorithm for string matching: an example}
  \label{fig:stringMatchConfig}
\end{figure}

 The FJS algorithm  combines two skip value functions: $\Delta\colon \Sigma\to [1,|\pat|+1]$
and $\beta\colon [0,|\pat|]\to [1,|\pat|]$; the former $\Delta$ comes from Quick Search and the latter $\beta$ comes from KMP 
(the choice of symbols follows~\cite{DBLP:journals/jda/FranekJS07}). See Fig.~\ref{fig:stringMatchConfig} where the pattern string $\pat=\textrm{STRING}$ is shifted by 6, 4 and 3 (instead of one-by-one). 

In the first shift we use the Quick Search skip value $\Delta(\textrm{S})=6$: we try matching the tail of $\pat$; it fails ($\pat(6)
\neq
\str(6)$);  then we find that the next character $\str(7)=\textrm{S}$ of the target only occurs in the first position of the pattern. 
Formally we define $\Delta$ by
\begin{equation}\label{eq:deltaForFJSStringMatch}
\small
\begin{aligned}
  \Delta(a)
 =
 \min\bigl(\,\bigl\{i\in[1,|\pat|]\;\big|\; a=\pat({|\pat|-i+1})\bigr\}\cup\bigl\{|\pat|+1\bigr\}\,\bigr)
\quad\text{for  $a\in\Sigma$.}
\end{aligned}
\end{equation}
In the example of Fig.~\ref{fig:stringMatchConfig} we have $\Delta(\text{I})=3$ and $\Delta(\text{Q})=7$.

Now we are in the second configuration of Fig.~\ref{fig:stringMatchConfig} and we try matching the tail $\pat(6)=\textrm{G}$ with $\str(12)$. It fails and we invoke the Quick Search skip value function $\Delta$; this results in a shift by $\Delta(\text{R})=4$ positions. 

\begin{wrapfigure}{r}{0pt}
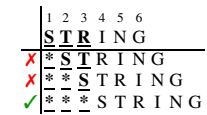

 \setlength{\tabcolsep}{1pt}
 \scriptsize
 \begin{tabular}{c|cccccccccccccccccccccccccccccccccccccccccccccc}
 & \tiny 1&\tiny 2&\tiny 3&\tiny 4&\tiny 5&\tiny 6\\[-.0em]
 &\bf\underline S&\bf\underline T&\bf\underline R&I&N&G\\[-.0em]\hline
 \color{red}\xmark&\bf\underline *&\bf\underline S&\bf\underline T&R&I&N&G\\[-.0em]
 \color{red}\xmark&\bf\underline *&\bf\underline *&\bf\underline S&T&R&I&N&G\\[-.0em]
 \color{dgreen}\cmark&\bf\underline *&\bf\underline *&\bf\underline *&S&T&R&I&N&G
 \end{tabular}
  \caption{$\beta(3)=3$, where the argument $3$ is the length of the successful partial match.}
  \label{fig:tableForBeta}
\end{wrapfigure}
For the shift from the third configuration to the fourth in Fig.~\ref{fig:stringMatchConfig} we employ the KMP skip value function $\beta$. It is defined as follows. Observe first that, in the third configuration of Fig.~\ref{fig:stringMatchConfig}, matching trials from the head succeed for three positions and then fail ($\str(11,13)=\pat(1,3), \str(14)\neq\pat(4)$). From this information alone we can see that, for a potential string match, the pattern string must be shifted at least by $\beta(3)=3$. See Fig.~\ref{fig:tableForBeta} where shifting the pattern string $\pat$  by one or two positions necessarily leads to a mismatch with $\pat(1,3)$. It is important here that we know $\pat(1,3)$ coincides with $\str(11,13)$ from the previous successful matching trials. Formally:
\begin{equation}\label{eq:betaForFJSStringMatch}
\small
\begin{aligned}
  \beta(p)= \min\{n \in [1,\pat] \mid \pat(1,p-n) = \pat(1+n,p)
  \}
\quad\text{for $p\in[0,|\pat|]$.}
\end{aligned}
\end{equation}

In the FJS algorithm we combine
 the two skip value function
 $\Delta$ and $\beta$. 
Specifically:
let us be in a configuration where $\pat(1)$ is in the position of $\str({1+n})$. We first try matching the pattern's tail $\pat({|\pat|})$ with its counterpart $\str({|\pat|+n})$; if it fails  we invoke the Quick Search skipping $\Delta$; otherwise we turn to the pattern's head $\pat(1)$  try matching from left to right. After its success or failure  we invoke the KMP skipping $\beta$. Note that preference is given to the Quick Search skipping. See Algorithm~\ref{alg:FJSStringMatch}. 
\begin{algorithm}[t]
 \caption{The FJS string matching algorithm (simplified)}
 \label{alg:FJSStringMatch}
 \scalebox{0.8}{
 \parbox{1.25\textwidth}{
 \begin{algorithmic}[1]
  \Require A target string $w$ and a pattern string $\pat$.
  \Ensure $Z$ is the set of matching intervals.
  
  \State $n \gets 1;$
   \Comment{$n$ is the position in $\str$ of the head of $\pat$}
  \While{$n \leq |\str|-|\pat|+1$}     
  \While{$\str({n+|\pat|-1})\neq \pat({|\pat|})$}
    \Comment{Try matching the tail of $\pat$}
  \State $n \gets n + \Delta(\str({n + |\pat|}))$
  \Comment{Quick Search-type skipping}
   \If{$n > |w| - m + 1$}
   \Return
   \EndIf
  \EndWhile
  \If{$\pat=\str(n,n+|\pat|-1)$}
   \Comment{We try  matching from left to right}
    \State $p\gets |\pat|+1;\qquad Z\gets Z\cup \{[n,n+|\pat|-1]\}$
  \Else
    \State $p\gets \min \{p'\mid \pat(p')\neq\str(n+p'-1)\}$
     \Comment{Matching trials fail at position $p$ for the first time}
  \EndIf
  \State $n \gets n + \beta(p)$
  \Comment{KMP-type skipping}
  \EndWhile
 \end{algorithmic}
}}
\end{algorithm}

It is important that the skip value functions
$\Delta\colon \Sigma\to [1,|\pat|+1]$
and $\beta\colon [0,|\pat|]\to [1,|\pat|]$ rely only on the pattern string $\pat$. Therefore it is possible to pre-compute the function values in advance (i.e.\ before a target string $\str$ arrives); moreover since $|\pat|$ is usually not large those values can be stored effectively in look-up tables. Skipping by these skip values does not improve the worst-case complexity, but practically it brings pleasing constant speed up, as demonstrated in Fig.~\ref{fig:stringMatchConfig}. 

Finally we note the following alternative presentation of
 $\Delta$ and $\beta$.
\begin{equation}
\label{eq:alternativeReprOfFJSSkipping}
\small
  \begin{aligned} 
  \Delta (a) &=
  \min
  \{n \in \mathbb{Z}_{>0} \mid \Sigma^{n}  \pat \cap
  \Sigma^{|\pat|}  a  \Sigma^* \neq \emptyset\}
  \qquad \text{for each $a\in \Sigma$,}
  \\
  \beta (p) &=
  \min
  \{n \in \mathbb{Z}_{>0} \mid \Sigma^{n} \pat(1,p) \cap \pat(1,p)
  \Sigma^* \neq \emptyset\}
  \qquad \text{for each $p\in [0,|\pat|]$.}
 \end{aligned}
\end{equation}

\auxproof{
\begin{algorithm}[t]
 \caption{A common outline of FJS-type algorithms}
 \label{alg:FJS}
 \begin{algorithmic}[1]
  \Require A target $w$ and a pattern $\pat$.
  \Ensure $Z$ is the set of matching intervals.
  
  \State $i \gets 1; m \gets \mathrm{QuickSearchLookAheadDistance} (\pat)$
  \While{$i \leq |w|$}
  \While{$\neg \mathrm{QuickSearchCheck}(i,w,\pat)$}\label{alg_line:common_quick_check}
  \State $i \gets i + \Delta(w_{i + m})$ \label{alg_line:common_quick_skip}
  \Comment{Quick Search-type skipping}
   \If{$n > |w| - m + 1$}
   \Return
   \EndIf
  \EndWhile
  \State $(\mathsf{ans}, \Conf) \gets \mathrm{MatchingTrial}(i,w,\pat)$\label{alg_line:common_matching_trial}
  \State $Z \gets Z \cup \mathit{ans}$
  \State $i \gets i + \beta(\Conf)$   \label{alg_line:common_KMP}
  \Comment{KMP-type skipping}
  \EndWhile
 \end{algorithmic}
\end{algorithm}

}

\auxproof{The FJS algorithm has two phases. See Alg.~\ref{alg:FJS}, where we present an outline of the algorithm.\footnote{The outline in  Alg.~\ref{alg:FJS} is simplified from the original FJS algorithm, for our technical developments later, so that the same outline applies to other problems like \emph{pattern} matching.  We note that the simplification results in loss of linear worst-case complexity in string matching. }
Firstly, the FJS algorithm tries to shift the pattern string $\pat$, by the skip value
$\Delta (a)$ that comes from   Quick Search, to skip some matching trials.
An intuition behind this Quick Search-type shifting is that $\pat(1)$ is easier to
match than $\pat(|\pat|)$. The initial configuration is in
(\ref{eq:config1}) of Fig.~\ref{fig:stringMatchConfig} and
attempts to match the last character ``\texttt{G}'' of $pat$ --- that is
line~\ref{alg_line:common_quick_check} in Alg.~\ref{alg:FJS}.
Since $\pat(6) \neq \str(6)$, we move the pattern where the ``next''
character $\str(7) = \mathtt{S}$ matches a character in $pat$
(line~\ref{alg_line:common_quick_skip} in Alg.~\ref{alg:FJS}).
In the next configuration (\ref{eq:config2}) in
Fig.~\ref{fig:stringMatchConfig}, we have $\pat(1) = \str(7) =
\mathtt{S}$.
Again we try matching $\pat(6)$ and $\str(12)$ and similarly, we can
move to the configuration (\ref{eq:config3}) in
Fig.~\ref{fig:stringMatchConfig} immediately.
Now $\pat(6) = \str(16) = \mathtt{G}$ and we start matching trials  --- that is
we finished the while loop from line~\ref{alg_line:common_quick_check} and
go to line~\ref{alg_line:common_matching_trial} in Alg.~\ref{alg:FJS}.
Trying matching from left to right, we
realize that $\pat(1,3) = \str(11,13) = \mathtt{STR}$ but
$\pat(4) \neq \str(14)$.
Then we try KMP like skipping --- that is line~\ref{alg_line:common_KMP}
in Alg.~\ref{alg:FJS}.
The KMP like skip value $\beta (3)$ is calculated by shifting the
pattern string as in Fig.~\ref{fig:tableForBeta} beforehand.
Since any prefix of $\pat(1,3)$ is not a proper suffix of $\pat(1,3)$,
we can move the pattern by $3$ and our configuration is
(\ref{eq:config4}) in Fig.~\ref{fig:stringMatchConfig}.
Then we go back to line~\ref{alg_line:common_quick_skip} in
Alg.~\ref{alg:FJS} again.
By repeating this procedure, we can solve the string matching problem.
The skip value functions $\Delta$ and $\beta$ are formally as follows.

\begin{mydefinition}[Skip values of the FJS algorithm for string matching]
 For a character $a\in\Sigma$ and a positive integer
 $p\in\mathbb{Z}_{>0}$, the Quick Search-type skip value $\Delta (a)$ and
 the KMP-type skip value $\beta (p)$ of the FJS algorithm for string
 matching is as follows.
 \begin{align*} 
  \Delta (a) &=
  \min
  \{n \in \mathbb{Z}_{>0} \mid \Sigma^{n}  \pat \cap
  \Sigma^{|\pat|}  a  \Sigma^* \neq \emptyset\}
  \\
  \beta (p) &=
  \min
  \{n \in \mathbb{Z}_{>0} \mid \Sigma^{n} \pat(1,p-1) \cap \pat(1,p-1)
  \Sigma^* \neq \emptyset,\\
  &\qquad  \qquad \pat (p) \neq \pat (p-n)\}
 \end{align*}
\end{mydefinition}
Note that the skip value functions $\Delta$ (over $a\in \Sigma$) and $\beta$ (over $p\in [1,|\pat|]$) can be organized as \emph{finite} tables. These tables are computed in the preprocessing stage, and are used in actual matching as in Fig.~\ref{fig:stringMatchConfig}. 
}

\auxproof{\subsection{Zone Automata}\label{subsec:zoneAutomata}

The reachability of a timed automaton can be checked with
a \emph{zone}~\cite{DBLP:conf/tacas/BehrmannBFL03} (or
\emph{region}~\cite{Alur1994}) automaton.
The main idea of the zone- (or region-) based reachability checking is
to take a finite equivalent class of the infinitely many clock configurations
and reduce to the reachability checking of an NFA.
A \emph{zone} is a convex polyhedron on the clock value space.
Since there exist infinitely many zones, some abstractions of zones are
proposed in the
literature~\cite{DBLP:conf/tacas/BehrmannBFL03,DBLP:journals/sttt/BehrmannBLP06}
to make the number of zones finite.
In this paper, we use a common abstraction
$\mathit{SG^a}$~\cite{DBLP:conf/cav/HerbreteauSW10}, that equate any
sufficiently large clock valuations.
A \emph{zone automaton} is incrementally constructed by tracing from the
initial conditions, reserving the soundness and the completeness of the
reachability.
Formally, they are defined as follows.

\begin{mydefinition} [zone]
 Let $\mathcal{A}$ be a timed automaton, $M$ be the maximum constant
 appearing in the guards of $\mathcal{A}$, $C$ be the set of clock
 variables of $\mathcal{A}$, and $\nu$ be a clock
 valuation for $C$.
A \emph{zone} $[\nu]$ abstracted by $\mathit{SG^a}$ is an equivalence
 class of clock
 valuations defined as a conjunction of the constraints in the
 form $\nu (x_j) - \nu (x_i) \prec c$ or
 $\pm \nu (x_i) \prec c$ where ${\prec} \in \{<,\leq\}$ and 
 $c \in [-M, M]$ is an integer.
\end{mydefinition}

\begin{mydefinition} [zone automaton]
For a timed automaton $\mathcal{A}$, the \emph{zone automaton}
 $\mathit{SG^a} (\mathcal{A})$ is a NFA whose state is the pair
 $(s,[\nu])$ of a state $s$ of the timed automaton $\mathcal{A}$ and a
zone $[\nu]$ abstracted by $\mathit{SG^a}$ and there exists a transition
$(s,[\nu]) \xrightarrow{a} (s',[\nu'])$ if for any
$\nu \in [\nu]$, there is a $\nu' \in [\nu']$ such that 
there is a run of $\mathcal{A}$ contain
$(s,\nu) \xrightarrow{\sigma} (s', \nu')$ as a subsequence and
 $[\nu']$ is the smallest such zone.
\end{mydefinition}

In addition to its practically small state space,
another advantage of the clock abstraction by zones is its efficient
representation by Difference Bound Matrices
(DBM)~\cite{DBLP:conf/avmfss/Dill89}.
A DBM $D$ is a $(|C| + 1) \times (|C| + 1)$ matrix that each element
$d_{i,j}$ represent a constraint $\nu (x_i) - \nu (x_j) \leq d_{i,j}$
for each clock variables $x_i,x_j \in C$
where we let $x_0$ be a special variable satisfying $\nu (x_0) = 0$ for
any clock valuation $\nu$.
A detail of DBM is found in~\cite{DBLP:conf/avmfss/Dill89} and a
detail of the operations on zones 
is found in~\cite{DBLP:journals/sttt/BehrmannBLP06}.
}

\section{An FJS-Type Algorithm for Pattern Matching}
\label{sec:fjs_pattern_matching}
In this section we present our first main contribution, namely an adaptation of the FJS algorithm (\S{}\ref{subsec:stringMatching}) from string matching to \emph{pattern matching}. 
\begin{mydefinition}[pattern matching]\label{def:patternMatching}
Let $\mathcal{A}$
be a nondeterministic finite automaton over an alphabet $\Sigma$ (a \emph{pattern} NFA), and  $\str\in\Sigma^{*}$ be a \emph{target} string. The \emph{pattern matching} problem requires
all the intervals $(i,j)$ for which the substring 
$\str(i,j)$ is accepted by $\mathcal{A}$. That is, it requires the set
\begin{math}
 \bigl\{\,(i,j)\,\big|\,1\le i\le j\le |\str| \text{ and }\str(i,j)\in L(\mathcal{A})\,\bigr\}
\end{math}.
\end{mydefinition}
For an example see Fig.~\ref{fig:patternMatchingExample}, where the automaton $\mathcal{A}$ satisfies $L(\mathcal{A}) = L(\mathtt{\{ab,cd\}cc^*d})$.

\begin{figure}[tbp]
 \begin{minipage}{0.43\textwidth}
\scalebox{.7}{ \begin{tikzpicture}[shorten >=1pt,node distance=1.3cm,on grid,auto]
  \node[initial,state] (S0){$s_0$};
  \node[state] (S1) at (1.2,0.5) {$s_1$};
  \node[state] (S2) at (1.2,-0.5) {$s_2$};
  \node[state] (S3) [right =2.4cm of S0] {$s_3$};
  \node[state] (S4) [right =of S3] {$s_4$};
  \node[accepting,state] (S5) [right =of S4]{$s_5$};

  \path[->]
  (S0) edge node {a} (S1)
  (S0) edge[below] node {c} (S2)
  (S1) edge node {b} (S3)
  (S2) edge[below] node {d} (S3)
  (S3) edge node {c} (S4)
  (S4) edge[loop above] node {c} (S4)
  (S4) edge node {d} (S5);
 \end{tikzpicture} 
}
\end{minipage}
\hfill
 \begin{minipage}{0.43\textwidth}
 \centering
 \setlength{\tabcolsep}{1pt}
 \scriptsize
 \vspace*{-1em}
  \begin{tabular}{ccccccccccccccccccccccccccccccccccccccccccccccc}
  &\tiny 1&\tiny 2&\tiny 3&\tiny 4&\tiny 5&\tiny 6&\tiny 7&\tiny 8&\tiny
  9&\tiny 10&\tiny 11&\tiny 12\\
  $\str =$&a&b&d&a&b&c&c&b&a&b&c&d\\
  $L(\mathcal{A}) \ni$&&&&&&&&&a&b&c&d\\
 \end{tabular}
\end{minipage}
 \caption{Pattern matching. For the pattern NFA $\mathcal{A}$ on the left, for which it is easy to see that $L(\mathcal{A}) = L(\mathtt{\{ab,cd\}cc^*d})$, the output is
$\{(9,12)\}$ as shown on the right.}
 \label{fig:patternMatchingExample}
 \centering
  \scalebox{0.8}{
 \begin{minipage}{0.7\textwidth}
 \centering
 \[
  L (\mathcal{A}) = \left\{
 \begin{array}{ll}
  \mathtt{abcd},& \mathtt{cdcd},\\
  \mathtt{abcc}\mathit{d},& \mathtt{cdcc}\mathit{d},\\
  \mathtt{abcc}\mathit{cd},& \mathtt{cdcc}\mathit{cd},\\
  \multicolumn{2}{c}{\vdots}
 \end{array}
 \right\}
\quad \leadsto\quad
 L'' =
 L'\cdot\Sigma^{*}=
  \left\{
 \begin{array}{c}
  \mathtt{abcd}, \mathtt{cdcd},\\
  \mathtt{abcc}, \mathtt{cdcc}\\
 \end{array}
 \right\} \Sigma^*
 \]
 \end{minipage}}
 \caption{Overapproximation of the language $L(\mathcal{A})$}
 \label{fig:approx}
\end{figure}

A brute-force algorithm solves pattern matching in $O(|S| |\str|^2)$, 
where $S$ is the state space of the pattern $\mathcal{A}$ (the factor $|S|$ is there due to nondeterminism). Some optimizations are known, among which is the adaptation of the Boyer--Moore algorithm by Watson \& Watson~\cite{Watson2003}. In their algorithm they adapt the BM-type skip values to pattern matching: the core idea in doing so is to \emph{overapproximate} languages and substrings, so that the skip value function can be organized as a finite table and hence can be computed in advance. Our adaptation of the FJS algorithm  employs similar overapproximation.

In the original FJS algorithm (for string matching) one uses skip value functions
\begin{equation}
 \Delta\colon \Sigma\to [1,|\pat|+1]
\quad\text{and}\quad
\beta\colon [0,|\pat|]\to [1,|\pat|]\enspace.
\end{equation} 
One may wonder what we can use in place of $|\pat|$, now that the pattern $\mathcal{A}$ can accept infinitely many words that are unboundedly long.

It turns out that our adaptations have the types
\begin{equation}
 \Delta: \Sigma \to [1 ,m + 1] \quad\text{and}\quad
 \beta: S \to [1, m]\enspace, 
\end{equation}
where 
 $m$ is the length of the shortest words accepted by $\mathcal{A}$ and $S$ is the state space of $\mathcal{A}$.
Intuitively, the original $\Delta$ does a comparison of the pattern
 $\pat$ with a
 character $a \in \Sigma$ and the original $\beta$ does a comparison of
 $\pat$ with the substring $\str(i,j)$ of the target string we actually
 read in the last matching trial.
Thus the adaptation can be done by a finite presentation of the
 overapproximation of $L(\mathcal{A})$ and $\str(i,j)$.

More specifically, for the approximation of $L(\mathcal{A})$: 1)
we focus on the length $m$ of the shortest accepted strings (four in the example of Fig.~\ref{fig:patternMatchingExample}); 2) 
we collect all the prefixes of  length $m$ that appear in $L(\mathcal{A})$ 
(\begin{math}
    \textrm{abcd}, \textrm{cdcd},
  \textrm{abcc}, \textrm{cdcc}
 \end{math} in the same example); and
3) we let an overapproximation $L''$ consist of any word that starts with those prefixes. See Fig.~\ref{fig:approx} for illustration; precise definitions are as follows.
\begin{displaymath}\small
\begin{array}{c}
  m=\min\{|w|\mid w\in L(\mathcal{A}) \}
 \quad
 L'
 =\bigl\{w'\in\Sigma^{m}\,\big|\,
 \exists w''\in\Sigma^{*}.\, w'w''\in L(\mathcal{A})
 \bigr\}
 \quad
 L''=L'\cdot \Sigma^{*}
\end{array}
\end{displaymath}
Here $L'\subseteq\Sigma^{m}$ is necessarily a finite set; thus $L''=L'\cdot \Sigma^{*}$ is an overapproximation of $L(\mathcal{A})$ with a finite representation $L'$.

For the overapproximation of the substring $\str(i,j)$ that we actually
read at the last matching trial, we exploit the set 
\begin{math}
 \mathcal{S}(\str(i,j)) = \{s\in S\mid s_{0}\xrightarrow{\str(i,j)} s \text{ in $\mathcal{A}$}\}
\end{math}
of states of $\mathcal{A}$.
 We have 
\begin{math}
 \str(i,j) \in \{w' \mid \forall s \in \mathcal{S}(\str(i,j)), \exists s_0 \in S_0.\,  s_{0}\xrightarrow{w'} s \text{ in $\mathcal{A}$}\}
\end{math}
, when $\mathcal{S}(\str(i,j)) \ne \emptyset$.
Using the overapproximation same as the one for $L'$, we obtain an
overapproximation of such $\str(i,j)$ 
represented by at most $2^{|S|}$ sets.



\begin{figure}[t]
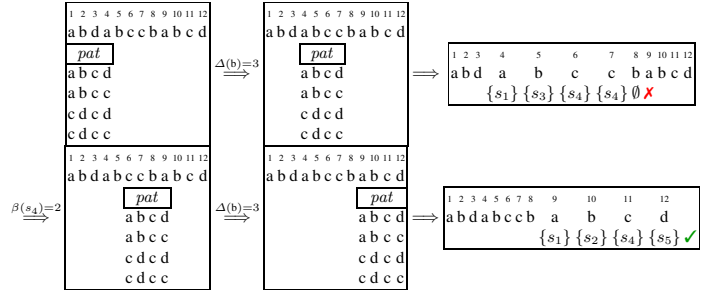

 \setlength{\tabcolsep}{1pt}
 \centering
\scalebox{.7}{  \begin{tabular}{c c c c}
   &
 \begin{tabular}{|cccccccccccc|}
  \hline
 \tiny 1&\tiny 2&\tiny 3&\tiny 4&\tiny 5&\tiny 6&\tiny 7&\tiny 8&\tiny
 9&\tiny 10&\tiny 11&\tiny 12\\
  a&b&d&a&b&c&c&b&a&b&c&d\\ \cline{1-4}
  \multicolumn{4}{|c|}{$\pat$}&&&&&&&&\\ \cline{1-4}
  a&b&c&d &&&&&&&&\\
  a&b&c&c &&&&&&&&\\
  c&d&c&d &&&&&&&&\\
  c&d&c&c &&&&&&&&\\  \hline
 \end{tabular}
 $\stackrel{\Delta(\textrm{b})=3}{\Longrightarrow}$&
 \begin{tabular}{|cccccccccccc|}
  \hline
 \tiny 1&\tiny 2&\tiny 3&\tiny 4&\tiny 5&\tiny 6&\tiny 7&\tiny 8&\tiny
 9&\tiny 10&\tiny 11&\tiny 12\\
  a&b&d&a&b&c&c&b&a&b&c&d\\ \cline{4-7}
  &&& \multicolumn{4}{|c|}{$\pat$} &&&&&\\ \cline{4-7}
  &&& a&b&c&d &&&&& \\
  &&& a&b&c&c &&&&&\\
  &&& c&d&c&d &&&&&\\
  &&& c&d&c&c &&&&&\\\hline
 \end{tabular}
 $\Longrightarrow$&
 \begin{tabular}{|cccccccccccc|}
  \hline
 \tiny 1&\tiny 2&\tiny 3&\tiny 4&\tiny 5&\tiny 6&\tiny 7&\tiny 8&\tiny
 9&\tiny 10&\tiny 11&\tiny 12\\
  a&b&d&a&b&c&c&b&a&b&c&d\\ 
  &&&$\{s_1\}$&$\{s_3\}$&$\{s_4\}$&$\{s_4\}$&$\emptyset$&\color{red}{\xmark}&&&\\\hline
 \end{tabular}
 \\
$\stackrel{\beta(s_4)=2}{\Longrightarrow}$&
 \begin{tabular}{|cccccccccccc|}
  \hline
 \tiny 1&\tiny 2&\tiny 3&\tiny 4&\tiny 5&\tiny 6&\tiny 7&\tiny 8&\tiny
 9&\tiny 10&\tiny 11&\tiny 12\\
  a&b&d&a&b&c&c&b&a&b&c&d\\ \cline{6-9}
  && &&& \multicolumn{4}{|c|}{$\pat$}&&&\\ \cline{6-9}
  && &&& a&b&c&d &&&\\
  && &&& a&b&c&c &&&\\
  && &&& c&d&c&d &&&\\
  && &&& c&d&c&c &&&\\\hline
 \end{tabular}
 $\stackrel{\Delta(\textrm{b})=3}{\Longrightarrow}$&
 \begin{tabular}{|cccccccccccc|}
  \hline
 \tiny 1&\tiny 2&\tiny 3&\tiny 4&\tiny 5&\tiny 6&\tiny 7&\tiny 8&\tiny
 9&\tiny 10&\tiny 11&\tiny 12\\
  a&b&d&a&b&c&c&b&a&b&c&d\\ \cline{9-12}
  &&& && &&& \multicolumn{4}{|c|}{$\pat$}\\ \cline{9-12}
  &&& && &&& a&b&c&d\\
  &&& && &&& a&b&c&c\\
  &&& && &&& c&d&c&d\\
  &&& && &&& c&d&c&c\\\hline
 \end{tabular}
 $\Longrightarrow$&       
 \begin{tabular}{|ccccccccccccc|}
  \hline
 \tiny 1&\tiny 2&\tiny 3&\tiny 4&\tiny 5&\tiny 6&\tiny 7&\tiny 8&\tiny
 9&\tiny 10&\tiny 11&\tiny 12&\\
  a&b&d&a&b&c&c&b&a&b&c&d&\\
  &&& && &&&$\{s_1\}$&$\{s_2\}$&$\{s_4\}$&$\{s_5\}$&\color{dgreen}{\cmark}\\
  \hline
 \end{tabular}
 \end{tabular}
}  \caption{Our FJS-type algorithm for pattern matching, for the example in Fig.~\ref{fig:patternMatchingExample}}
 \label{fig:patternMatchConfig}
\end{figure}

Let us demonstrate our two skip value functions $\Delta$ and $\beta$ using the example in Fig.~\ref{fig:patternMatchingExample}; the execution trace of our algorithm is in Fig.~\ref{fig:patternMatchConfig}. In the first configuration we try to match the tail of all the possible length-4 prefixes of $L(\mathcal{A})$ with $\str(4)=\text{a}$, which fails. Then we invoke the Quick Search-type skipping $\Delta(\str(5))=\Delta(\text{b})$; since $\text{b}$ occurs no later than in the second position in $L'=\{\text{abcd},\text{abcc},\text{cdcd},\text{cdcc}\}$, we can skip by three positions and reach the second configuration. 

We again try matching from the tail; this time we succeed since $\str(7)=\text{c}$ appears as a tail in $L'$. We subsequently move to the phase where we match from left to right, much like in the original FJS algorithm (\S{}\ref{subsec:stringMatching}). Concretely this means we feed the automaton $\mathcal{A}$ (see Fig.~\ref{fig:patternMatchConfig}) the remaining segment  $\str(4)\str(5)\dotsc$ from left to right; we obtain $\{s_1\}\{s_3\}\{s_4\}\{s_4\}\emptyset$ as the sequence of reachable sets. Since no accepting states occur therein and we have reached the emptyset, we conclude that the matching trial starting at the position $\str(4)$ is unsuccessful.

\begin{wrapfigure}{r}{0pt}
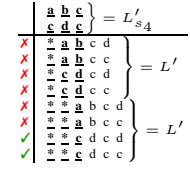

 \setlength{\tabcolsep}{1pt}
 \tiny
 \begin{tabular}{l|l}
  &
      \multirow{2}{*}{
      $\left.
      \begin{array}{ccc}
       \textbf{\underline a}&\textbf{\underline b}&\textbf{\underline c}\\
        \textbf{\underline c}&\textbf{\underline d}&\textbf{\underline c}
      \end{array}
          \right\} = L'_{s_4}
          $}\\
  \\
  \hline 
 \color{red}\xmark&
      \multirow{4}{*}{
      $\left.
\begin{array}{cccccc}
\textbf{\underline *}&\textbf{\underline a}&\textbf{\underline b}&\textrm{c}&\textrm{d}&\\
\textbf{\underline *}&\textbf{\underline a}&\textbf{\underline b}&\textrm{c}&\textrm{c}&\\
\textbf{\underline *}&\textbf{\underline c}&\textbf{\underline d}&\textrm{c}&\textrm{d}&\\
\textbf{\underline *}&\textbf{\underline c}&\textbf{\underline d}&\textrm{c}&\textrm{c}&\\
       \end{array}
  \right \} = L'
$}\\
\color{red}\xmark\\
\color{red}\xmark\\
\color{red}\xmark\\
\color{red}\xmark&
      \multirow{4}{*}{
$\left.
\begin{array}{cccccc}
\textbf{\underline *}&\textbf{\underline *}&\textbf{\underline a}&\textrm{b}&\textrm{c}&\textrm{d}\\
\textbf{\underline *}&\textbf{\underline *}&\textbf{\underline a}&\textrm{b}&\textrm{c}&\textrm{c}\\
\textbf{\underline *}&\textbf{\underline *}&\textbf{\underline c}&\textrm{d}&\textrm{c}&\textrm{d}\\
\textbf{\underline *}&\textbf{\underline *}&\textbf{\underline c}&\textrm{d}&\textrm{c}&\textrm{c}
\end{array}
  \right\} = L'
$}\\
 \color{red}\xmark\\
 \color{dgreen}\cmark\\
 \color{dgreen}\cmark\\
 \end{tabular}
  \caption{$\beta(s_4)$}
  \label{fig:tableForPatternBeta}
\end{wrapfigure}
Now we invoke the KMP-type skipping $\beta$. In the original FJS
algorithm we used the data of successful partial matching
($\str(4,7)=\text{abcc}$ in the current case) for computing $\beta$;
this is not possible, however, since it is infeasible to prepare skip
values for all possible $\str(i,j)$. Instead we use the data
$\mathcal{S}(\str(4,7)) = \{s_4\}$ and the set 
$L'_{s_{4}}=\{\text{abc},\text{cdc}\}$ as an overapproximation of the
partial match $\str(4,7) = \text{abcc}$. The intuition of the set $L'_{s_{4}}$ is that: for a word $w'$ to drive $\mathcal{A}$ from an initial state to $s_{4}$, $w'$ must have either $\text{abc}$ or $\text{cdc}$ as its prefix. In Fig.~\ref{fig:tableForPatternBeta} is how we compute the skip value $\beta(s_{4})$, using the approximant $L'_{s_{4}}$ of the partial match and the approximant $L'$ of the pattern. Note also that it follows the same pattern as Fig.~\ref{fig:tableForBeta}. 

We are now in the fourth configuration in Fig.~\ref{fig:patternMatchConfig}. The matching trial at the position $9$ fails and we invoke the Quick Search-type skipping, much like before. In the fifth configuration, the matching trial at the position $12$ succeeds, which makes us try matching from the left, feeding $\mathcal{A}$ with $\str(9,12)$. We reach $s_{5}$ and thus succeed.

\begin{algorithm}[t]
 \caption{The FJS algorithm for pattern matching, for a target $\str$ and a pattern $\mathcal{A}$}
 \label{alg:FJSPatternMatching}
 \scalebox{0.8}{
 \parbox{1.25\textwidth}{
 \begin{algorithmic}[1]
  \Ensure $Z$ is the set of matching intervals.  
  \State $n \gets 1;$
   \Comment{$n$ is the position in $\str$ of the head of $\pat$}
  \While{$n \leq |\str|-m+1$}     
  \While{$\forall w' \in L'.\, \str({n+m-1})\neq w'(m)$}
    \Comment{Try  matching the tail of $L'$}
  \State $n \gets n + \Delta(\str({n + m}))$
  \Comment{Quick Search-type skipping}
   \If{$n > |w| - m + 1$}
   \Return
   \EndIf
  \EndWhile
  \State $Z\gets Z\cup \{(n,n') \mid w (n,n') \in L (\mathcal{A})\}$
  \Comment{We try matching by feeding $w(n,|w|)$ to $\mathcal{A}$}
  \State $n' \gets \max\{n' \in [1, |\str|] \mid \exists s_0 \in S_0,s \in S.\, s_0
  \xrightarrow{w(n,n')} s\}$
  \Comment{$n'$ is the position of the last successful match}
  \State $S' \gets \{s \in S \mid \exists s_0 \in S_0.\, s_0 \xrightarrow{w(n,n')} s\}$
  \Comment{Matching trials stack at the states $S'$}
  \State $n \gets n + \max_{s \in S'}{\beta(s)}$
  \Comment{KMP-type skipping}
  \EndWhile
 \end{algorithmic}
}}
\end{algorithm}

Overall our FJS-type algorithm for pattern matching is as in Algorithm~\ref{alg:FJSPatternMatching}. The skip value functions therein are defined as follows. They are similar to the ones in~(\ref{eq:alternativeReprOfFJSSkipping}). Since $L'$ and $L'_{s}$ are all finite, computing $\Delta$ and $\beta$ is straightforward.

\vspace{.3em}
\noindent
\begin{minipage}{\textwidth}
\begin{mydefinition}[Skip values in our FJS-type  pattern
 matching algorithm]\label{def:patternMatchingSkipValue}
 Let $\mathcal{A}=(\Sigma,S,S_{0},E,F)$ be a pattern NFA, $a\in\Sigma$
 be a character, $s$ be a state of $\mathcal{A}$, and
$\mathcal{A}_{s}=(\Sigma,S,S_{0},E,\{s\})$ be the automaton where $s$ is
the only accepting state. 
Let 
$m_{s}
= \min\{|w|\mid w\in L(\mathcal{A}_{s})\}$ 
(the length of a shortest word that leads to $s$)
and 
$m
 =
 \min_{s\in F}m_{s}
 $ (the length of a shortest accepted word). 
The skip value functions
 $\Delta: \Sigma \to [1 ,m + 1]$ and
 $\beta: S \to [1, m]$ are defined as follows. 
\begin{equation*}\small
 \begin{aligned}
 L'
 &=
 \{w(1,m) \mid w\in L(\mathcal{A})\}
 \qquad
 L'_{s}
 =
 \{w(1,\min\{m_{s},m\}) \mid w\in L(\mathcal{A}_{s})\}
 \\
 \Delta (a) &=
 \min
 \{n \in \mathbb{Z}_{>0} \mid \Sigma^{n}  L' \cap
 \Sigma^{m}  a  \Sigma^* \neq \emptyset\}
 \\
\beta (s) &= 
 \min
 \{n \in \mathbb{Z}_{>0} \mid \Sigma^{n} L' \cap L'_s \Sigma^* \neq \emptyset
 \}
\end{aligned}\end{equation*}
\end{mydefinition}
\end{minipage}

\section{An FJS-Type Algorithm for Timed Pattern Matching}
\label{sec:timedFJS}
Here we present our second main contribution: an FJS-type algorithm for timed pattern matching. It is superior to our previous Boyer--Moore-type algorithm~\cite{DBLP:conf/formats/WagaAH16}, in its performance both in offline and online scenarios.
We fix a target timed word $w = (\overline{a},
\overline{\tau})$ and a pattern timed automaton
$\mathcal{A} = (\Sigma \sqcup \{\$\},S,S_0,C,E,F)$.
We further assume the following that means $\mathcal{A}$ is a suitable
pattern for timed pattern matching.

\vspace{.3em}
\noindent
\begin{minipage}{\textwidth}
\begin{myassumption}
 $\mathcal{A}$ satisfies the following:
 any transition to an accepting state is labelled with the terminal character \$;
  no other transition is labelled with \$;
 and there is no transition from an accepting state.
\end{myassumption}
\end{minipage}

The basic idea of our FJS-type  algorithm here is the same as
in~\S{}\ref{sec:fjs_pattern_matching}: we use two skip value functions
$\Delta$ and $\beta$; and
and for their finitary representation we let states of automata
 overapproximate various infinitary data, as we 
 explain later. 
 In the current timed setting, however, we cannot use a pattern timed automaton $\mathcal{A}$ itself to play the same role---in a run of $\mathcal{A}$ a state is always accompanied with a clock valuation that takes continuous values. We overcome this difficulty relying on the \emph{zone abstraction} of timed automata, a construction that turns a timed automaton into an NFA maintaining reachability (see e.g.~\cite{DBLP:conf/cav/HerbreteauSW10}).\footnote{In our previous work~\cite{DBLP:conf/formats/WagaAH16} we used \emph{regions}~\cite{Alur1994} in place of zones. Though equivalent in terms of finiteness, zones give more efficient abstraction than regions.}

\vspace{.3em}
\noindent
\begin{minipage}{\textwidth}
\begin{mydefinition} [zone]
\label{def:zone} 
Let $\mathcal{A}$ be a timed automaton over the set $C$ of clock 
 variables, and $M$ be the maximum constant occurring in the guards of
 $\mathcal{A}$. A \emph{zone} is a $|C|$-dimensional polyhedron specified with a conjunction of the constraints of the form 
 $\nu (x_j) - \nu (x_i) \prec c$,
  $\nu (x_i) \prec c$ 
 or $-\nu (x_i) \prec c$,
where ${\prec} \in \{<,\leq\}$ and 
  $c \in [-M, M]$.
\end{mydefinition}
\end{minipage}

A \emph{zone automaton}  $\mathcal{Z}$ for a timed automaton $\mathcal{A}$ is an NFA whose states are pairs $(s,\alpha)$ of a state $s$ of $\mathcal{A}$ and a zone $\alpha$; it is meant to be a finite abstraction of the timed automaton $\mathcal{A}$ via which we study properties of $\mathcal{A}$.
There are many different known constructions of zone automata (see e.g.~\cite{DBLP:conf/cav/HerbreteauSW10,DBLP:conf/tacas/BehrmannBFL03}): they come with different efficiency (i.e.\ the size of the resulting NFA), and with different preservation properties (bisimilarity to $\mathcal{A}$, similarity, etc.). For our current purpose it does not matter which precise construction we use;
we chose a common construction $\mathit{SG^a}$ from~\cite{DBLP:conf/cav/HerbreteauSW10}, mainly for its ease of implementation.

A \emph{path} of a zone automaton  $\mathcal{Z}$ is much like a run, but it is allowed to start at a possibly non-initial state.
A  path $r = (\overline{s},\overline{\nu})$ of a timed automaton $\mathcal{A}$ is called an \emph{instance} of a  path $\overline{r} = (\overline{s},\overline{\alpha})$ of a zone automaton $\mathcal{Z}$ for $\mathcal{A}$ if, for any $n \in [0,|s|-1]$, we have $\overline{\nu}_n \in \overline{\alpha}_n$.
Conversely, such $\overline{r}$ is called an \emph{abstraction} of $r$.
In this paper we rely on the following  preservation property of the specific construction $\mathcal{Z}=\SG(\mathcal{A})$ of zone automata: every run in $\mathit{SG^a}(\mathcal{A})$ is an abstraction of some run of $\mathcal{A}$; conversely every run of $\mathcal{A}$ is an instance of some run in $\mathit{SG^a}(\mathcal{A})$.
See~\cite{DBLP:conf/cav/HerbreteauSW10} for details.

\begin{algorithm}[t]
 \caption{\small Our FJS-type algorithm for timed pattern matching, for a target $\str$ and a pattern $\mathcal{A}$}
 \label{alg:FJSTimedPatternMatching}
 \scalebox{0.8}{
 \parbox{1.25\textwidth}{
 \begin{algorithmic}[1]
   \Ensure $Z$ is the match set $\mathcal{M} (w,\mathcal{A})$ in
   Def.~\ref{def:TimedPatternMatching}.
  \State $n \gets 1;$
   \Comment{$n$ is the  position in $\str$ of the beginning of the current matching trial}
  \State $\nu_0 \gets (\text{the clock valuation that returns $0$ for any
  clock variable})$
  \While{$n \leq |\str|-m+2$}     
  \While{$\forall\, \overline{r} \in L'.\, \overline{a}_{n+m-2}\neq a' \,\text{(where $a'$ is such that $\overline{r}_{m-2} \xrightarrow{a'} \overline{r}_{m-1}$)}$}
    \Comment{Try matching the tail of $L'$}
  \State $n \gets n + \Delta(\overline{a}_{n + m - 1})$
  \Comment{Quick Search-type skipping}
   \If{$n > |w| - m + 2$}
   \Return
   \EndIf
  \EndWhile
  \State $Z\gets Z\cup \{(t,t') \in [\tau_{n-1},\tau_t) \times (\tau_{n-1},\infty) \mid w|_{(t,t')} \in L (\mathcal{A})\}$
  \label{line:leftToRightMatchingInTimedFJS}
  \Comment{Try matching from left to right
 }
  \State $n' \gets \max\{n' \in [1, |\str|] \mid \exists s_0 \in S_0,s \in S,\nu \in (\Rnn)^C.\, (s_0,\nu_0) \xrightarrow{w(n,n')} (s,\nu)\}$
  \State $S' \gets \{s \in S \mid \exists s_0 \in S_0,\nu \in (\Rnn)^C.\, (s_0,\nu_0) \xrightarrow{w(n,n')} (s,\nu)\}$
  \Comment{Matching trials stack at the states $S'$}
  \State $n \gets n + \max_{s \in S'}{\beta(s)}$
  \Comment{KMP-type skipping}
  \EndWhile
 \end{algorithmic}
}}
\end{algorithm}
Our algorithm is in Algorithm~\ref{alg:FJSTimedPatternMatching}. The constructs therein are defined as follows.

\vspace{.3em}
\noindent
\begin{minipage}{\textwidth}
\begin{mydefinition}[FJS-type skip values for timed pattern matching]
 \label{def:timed_fjs_skip_value}
Let $\overline{r}$ be a path of the zone automaton $\SG(\mathcal{A})$. 
The 
 set  $\mathcal{W} (\overline{r})$ of timed words represented by $\overline{r}$ is:
\begin{displaymath}
 \mathcal{W} (\overline{r}) = \{ w\in(\Sigma \times \Rp)^{*} \mid
 \text{there is a path $r$ of $\mathcal{A}$ over $w$ that is an instance of $\overline{r}$}
\}\enspace.
\end{displaymath}
For a set $K$ of paths of  $\SG(\mathcal{A})$, the definition naturally extends by $\mathcal{W} (K)=\bigcup_{\overline{r}\in K}\mathcal{W}(\overline{r})$. 
Let
$\mathcal{A}_s= (\Sigma,S,S_0,E,C,\{s\})$ be the modification of $\mathcal{A}$ in which $s$ is the only accepting state.
Let 
$m_{s}
= \min\{|w|\mid w\in L(\mathcal{A}_{s})\}$ 
and 
$m
 =
 \min_{s\in F}m_{s}
 $.
Following the discussion in \S{}\ref{sec:fjs_pattern_matching},
we define the overapproximations $L''$ of $L(\mathcal{A})$ and $L'_{s}$. as follows.
Note that $L'$ and $L'_{s}$ are in fact sets of runs of $\SG(\mathcal{A})$; $L''$ is a set of timed words.
\begin{equation*}
 \small
\begin{aligned}
 L' &= \{\overline{r}(0,m-1) \mid \text{$\overline{r}$ is a  run of $\SG(\mathcal{A})$, and } \mathcal{W} (\overline{r}) \cap L(\mathcal{A}) \ne \emptyset\} 
\\
 L''&= \mathcal{W}(L') \cdot (\Sigma \times \Rp)^*  
\\
 L'_{s}
 &=
 \{\overline{r}(0,\min\{m_{s},m - 1\}) \mid
 \text{$\overline{r}$ is a  run of $\SG(\mathcal{A})$, and }\mathcal{W} (\overline{r}) \cap L(\mathcal{A}_{s}) \ne \emptyset\}
\end{aligned}
\end{equation*}
These are used in the following definition of skip values. Here $a\in \Sigma$ and $s\in S$. 
\begin{equation}\label{eq:skipValuesForTimedFJS}\small
  \begin{aligned}
 \Delta (a) &=
 \min
 \{n \in \mathbb{Z}_{>0} \mid 
\\
  &
  \exists t \in \mathbb{R}_{>0}.\, (\Sigma \times \mathbb{R}_{>0})^{n} \cdot \mathcal{W}(L') \cap
(\Sigma \times \mathbb{R}_{>0})^{m-1}  \cdot (a,t) \cdot (\Sigma \times \mathbb{R}_{>0})^* \neq \emptyset\}\\
 \beta (s) &=
 \min
 \{n \in \mathbb{Z}_{>0} \mid (\Sigma \times \mathbb{R}_{>0})^{n} \cdot
 \mathcal{W}(L') \cap \mathcal{W}(L'_s) \cdot (\Sigma \times
 \mathbb{R}_{>0})^* \neq \emptyset \}
\end{aligned}
\end{equation}
\end{mydefinition}
Note the similarity between the last definition and~(\ref{eq:alternativeReprOfFJSSkipping}).
\end{minipage}

 Explanation is in order how some operations in Algorithm~\ref{alg:FJSTimedPatternMatching} (and in Def.~\ref{def:timed_fjs_skip_value}) can be implemented. First note that $\mathcal{W}(\overline{r})$ is an infinite set. The set $L'$ is finite and computable nevertheless: due to the preservation property of the zone automaton $\SG(\mathcal{A})$, the condition $\mathcal{W} (\overline{r}) \cap L(\mathcal{A}) \ne \emptyset$ simply means $\overline{r}$ is accepting. The same goes for $L'_{s}$. For $\Delta$, we realize that the second argument
$(\Sigma \times \mathbb{R}_{>0})^{m-1}  \cdot (a,t) \cdot (\Sigma \times \mathbb{R}_{>0})^*$ of the intersection does not pose any timing constraint. Therefore the timed nonemptiness problem reduces to an untimed one that is readily solved. Solving the timed nonemptiness problem for $\beta$ in~(\ref{eq:skipValuesForTimedFJS}) is nontrivial. Here we use emptiness check in
$\SG(\mathcal{A} \times \mathcal{A})$---the zone automaton of the product of $\mathcal{A}$ with itself, changing its initial state suitably in order to address shift of words---to check whether the intersection of the two relevant languages is empty. Finally, the left-to-right matching on Line~\ref{line:leftToRightMatchingInTimedFJS} is done by accumulating constraints on $t$ in the course of necessary transitions. Further details are in Appendices~\ref{appendix:detailFJS}--\ref{appendix:TimedFJSIllustrated}.

A correctness proof (i.e.\ our skipping does not affect the output) is in Appendix~\ref{appendix:correctness}.

\auxproof{
In our FJS-type algorithm for timed pattern matching, we use the following skip value functions: $\Delta: \Sigma \to [1,m]$ and $\beta: S \to [1,m-1]$ where $m$ is the length of a shortest word accepted by $\mathcal{A}$ and $S$ is the state space of $\mathcal{A}$. 
Though the functions have similar type to the ones for pattern matching, we cannot compute them in the way in \S{}\ref{sec:fjs_pattern_matching}.
That is because there exist infinitely many length-$m$ timed words due to the infinity of the time domain.
Similarly to~\cite{DBLP:conf/formats/WagaAH16}, we use an abstraction of clock values maintaining soundness and completeness (w.r.t reachability).
In~\cite{DBLP:conf/formats/WagaAH16}, \emph{region}~\cite{Alur1994} was used.
Here instead, we employ \emph{zone}, more precisely $\SG$ in~\cite{DBLP:conf/cav/HerbreteauSW10}.
Since a state $(s,[\nu])$ of the \emph{zone automaton} $\SG(\mathcal{A})$ is isomorphic to a set of a pair $(s,\nu)$ of a state $s$  of $\mathcal{A}$ and a clock valuation $\nu$ of $\mathcal{A}$, a (partial) run $\overline{r}$ of $\SG(\mathcal{A})$ is naturally identified with a set of (partial) runs $r$ of $\mathcal{A}$.
By a partial run $\overline{r}$ of the zone automaton $\SG(\mathcal{A})$, we represent the set of timed words along $\overline{r}$, precisely 
$\mathcal{W} (\overline{r}) = \{ w | \text{$\exists r \in \overline{r}.\,r$ is a partial run of $\mathcal{A}$ over $w$}\}$.
The definition of $\mathcal{W}(\overline{r})$ naturally extends to a set of partial runs of $\SG(\mathcal{A})$.
The abstraction $L''$ of $L(\mathcal{A})$ is
\[
 L' = \{\overline{r}(0,m-1) \mid \mathcal{W} (\overline{r}) \cap L(\mathcal{A}) \ne \emptyset\} \qquad
 L''= \mathcal{W}(L') \cdot (\Sigma \times \Rp)^*  
\]
Since zone-abstraction makes the state space smaller than region-abstraction practically, zone-based implementation may handle some patterns that are infeasible for region-based implementation e.g. Case~4 in~\cite{DBLP:conf/formats/WagaAH16}.
}

One important idea in our algorithm is that we use timing constraints---in addition to character constraints like in Fig.~\ref{fig:tableForBeta} \&~\ref{fig:tableForPatternBeta}---in calculating skip values. By this we achieve greater skip values, while keeping the computational overhead minimal by the use of the zone automaton $\SG(\mathcal{A} \times \mathcal{A})$. 

The way our algorithm (Algorithm~\ref{alg:FJSTimedPatternMatching}) operates is very similar to  the one in~\S{}\ref{sec:fjs_pattern_matching} for (untimed) pattern matching, as we already described earlier. There the zone automaton $\SG(\mathcal{A})$ plays important roles in the calculation of skip values. For the record we include in Appendix~\ref{appendix:TimedFJSIllustrated} the illustration of our algorithm using the example in Fig.~\ref{fig:input_timed_pattern_matching}.

\vspace{.3em}
\noindent
\begin{minipage}{\textwidth}
\paragraph{Online Properties}
We claim that the current FJS-type algorithm is much better suited to online usage scenarios than our previous BM-type one~\cite{DBLP:conf/formats/WagaAH16}. 
See Fig.~\ref{fig:FJSIsOnlineBMIsNot}. In our FJS-type algorithm 
we can sometimes increment $n$ before reading the whole target timed word $w$ (``unnec.'' for ``unnecessary'' in Fig.~\ref{fig:FJSIsOnlineBMIsNot}); this is the case when we observe that no further transition is possible in the pattern automaton $\mathcal{A}$. (Additionally, thanks to the skip values $\Delta$ and $\beta$, sometimes we can increment $n$ by more than one). For real-world examples we can assume that matches tend to be much shorter than the whole log $w$; this means the ``unnec.'' parts are often big. 
\end{minipage}
 In the BM-type algorithm, in contrast, matching trials  start almost at the tail of $w$,\footnote{To be precise we can start without the last $m-1$ characters, where $m$ is the length of a shortest word accepted by $\mathcal{A}$. Usually $m$ is by magnitude smaller than $|w|$.}  and we have to wait until the arrival of the whole target word. This contrast is 
 witnessed in our experimental results, specifically on those for memory usage.


\section{Experiments}
\label{sec:experiments}


\newcommand{\OFFLINE}{\mathit{offline}}
\newcommand{\ONLINE}{\mathit{online}}
\newcommand{\BRUTEFORCE}{\mathit{bruteforce}}
\newcommand{\BMREGION}{\mathit{BMregion}}
\newcommand{\BMZONE}{\mathit{BMzone}}
\newcommand{\FJS}{\mathit{FJS}}
\newcommand{\IMPL}{\mathcal{I}}
\newcommand\set[1]{\{{#1}\}}

\begin{wrapfigure}[4]{r}{0pt}
\scalebox{.85}{\begin{tabular}{c||c|c|c||c}
&brute-force&BM  
 &FJS
 & Montre 
\\\hline\hline
  offline&from~\cite{DBLP:conf/formats/WagaAH16}&from~\cite{DBLP:conf/formats/WagaAH16}&new&from~\cite{DBLP:journals/corr/Ulus16}\\\hline
  online&from~\cite{DBLP:conf/formats/WagaAH16}&---&new&from~\cite{DBLP:journals/corr/Ulus16}\\
\end{tabular}
}
\end{wrapfigure}
We implemented our FJS-type algorithm for timed pattern matching---its online and offline variations difference between which will be elaborated later. We compared its performance with that of: brute-force algorithms (online and offline); the BM-type algorithm~\cite{DBLP:conf/formats/WagaAH16}; and the tool \emph{Montre}~\cite{DBLP:journals/corr/Ulus16} for timed pattern matching.

The BM- and FJS-type algorithms employ zone-based abstraction; it is implemented using \emph{difference bound matrices}, following~\cite{DBLP:conf/avmfss/Dill89}. Zone construction and calculation of skip values are done in the preprocessing stage, where the most expensive is the emptiness checking for $\beta(s)$ (see~(\ref{eq:skipValuesForTimedFJS})). 
We optimized this part, memorizing parts of zone automata and reusing them in computing  $\beta(s)$ for  different $s$. As a result the preprocessing stage takes a fraction of a second for each of our benchmark problems. See Appendix~\ref{appendix:preprocessing} for details.

For brute-force and FJS, the algorithms are the same in their
online and offline implementations.  In the online implementations, a
target timed word is read lazily and a memory cell is deallocated as
soon as we realize it is no longer needed.  In the offline
implementations, the whole target timed word is read and stored in memory in the beginning, and the memory cells are not deallocated until
the end. The tool Montre employs different algorithms in its online and offline usage modes. See~\cite{DBLP:journals/corr/Ulus16} for details.

 In our current implementations, we hardcode a pattern timed automaton
in the code. Developing a parser for user-defined timed automata should
not be hard.

The benchmark problems we used are in Fig.~\ref{fig:case1_pattern}--\ref{fig:case6_pattern}  (the pattern automata $\mathcal{A}$ and the set $W$ of target words). They are from automotive scenarios except for the first two.

\begin{figure}[tbp]
\setlength\abovecaptionskip{0pt}
 \begin{minipage}{0.48\textwidth}
\scalebox{0.78}{\begin{tikzpicture}
\node at (2.5,3.3) {\textbf{BM}};
 \draw (0,3) rectangle (5,2.5);
 \node at (2.5,2.75) {$w$};

 \foreach \x / \y in {3.5/2.2,4.0/1.9,4.5/1.6}
 {
 \draw [thick,|-|](2.5,\y)--(\x,\y);
 \node at (3.0,\y) {$\blacktriangleright$};
 }
 \node at (3.5,1.3) {$\vdots$}; 
 \draw [thick,|-|](2.5,0.9)--(5,0.9);
 \node at (3.0,0.9) {$\blacktriangleright$};
 \node at (1.2,1.5) {$\Biggg[$};
 \node at (-0.3,1.5) {$n = |w| - m + 2$};
 \node at (5.3,1.2) {$\Biggr]$};
 \node at (5.8,1.2) {unnec.};
 \foreach \x / \y in {3.0/0.4,3.5/0.1,4.0/-0.2}
 {
 \draw [thick,|-|](2.0,\y)--(\x,\y);
 \node at (2.5,\y) {$\blacktriangleright$};
 }
 \node at (3.0,-0.5) {$\vdots$};
 \draw [thick,|-|](2.0,-0.9)--(5.0,-0.9); 
 \node at (2.5,-0.9) {$\blacktriangleright$};
 \node at (1.2,-0.3) {$\Biggg[$};
 \node at (-0.2,-0.3) {$n = |w| - m + 1$};
 \node at (5.3,-0.6) {$\Biggr]$};
 \node at (5.8,-0.6) {unnec.};

 \node at (2.5,-1.3) {${\Large \vdots}$};
 \foreach \x / \y in {1.0/-1.9,1.5/-2.2,2.0/-2.5}
 {
 \draw [thick,|-|](0,\y)--(\x,\y);
 \node at (0.5,\y) {$\blacktriangleright$};
 }
 \node at (0.75,-2.8) {$\vdots$};
 \draw [thick,|-|](0.0,-3.2)--(5.0,-3.2);
 \node at (0.5,-3.2) {$\blacktriangleright$};
 \node at (-0.5,-2.5) {$\Biggg[$};
 \node at (-1.0,-2.5) {$n = 1$};
 \node at (5.3,-2.9) {$\Biggr]$};
 \node at (5.8,-2.9) {unnec.};
\end{tikzpicture}}
\end{minipage}
\hfill
 \begin{minipage}{0.48\textwidth}
\scalebox{0.78}{\begin{tikzpicture}
\node at (2.5,3.3) {\textbf{FJS}};
 \draw (0,3) rectangle (5,2.5);
 \node at (2.5,2.75) {$w$};

 \foreach \x / \y in {1.0/2.2,1.5/1.9,2.0/1.6}
 {
 \draw [thick,|-|](0,\y)--(\x,\y);
 \node at (0.5,\y) {$\blacktriangleright$};
 }
 \node at (1.0,1.3) {$\vdots$};
 \draw [thick,|-|](0,0.9)--(5,0.9);
 \node at (0.5,0.9) {$\blacktriangleright$};
 \node at (-0.5,1.5) {$\Biggg[$};
 \node at (-1.0,1.5) {$n = 1$};
 \node at (5.3,1.2) {$\Biggr]$};
 \node at (5.8,1.2) {unnec.};
 \foreach \x / \y in {1.5/0.4,2.0/0.1,2.5/-0.2}
 {
 \draw [thick,|-|](0.5,\y)--(\x,\y);
 \node at (1.0,\y) {$\blacktriangleright$};
 }
 \node at (1.25,-0.5) {$\vdots$};
 \draw [thick,|-|](0.5,-0.9)--(5,-0.9);
 \node at (1.0,-0.9) {$\blacktriangleright$};
 \node at (-0.5,-0.3) {$\Biggg[$};
 \node at (-1.0,-0.3) {$n = 2$};
 \node at (5.3,-0.6) {$\Biggr]$};
 \node at (5.8,-0.6) {unnec.};

 \node at (2.5,-1.3) {${\Large \vdots}$};

 \foreach \x / \y in {3.5/-1.9,4.0/-2.2,4.5/-2.5}
 {
 \draw [thick,|-|](2.5,\y)--(\x,\y);
 \node at (3.0,\y) {$\blacktriangleright$};
 }
 \node at (3.5,-2.8) {$\vdots$};
 \draw [thick,|-|](2.5,-3.2)--(5.0,-3.2);
 \node at (3.0,-3.2) {$\blacktriangleright$};
 \node at (1.5,-2.5) {$\Biggg[$};
 \node at (0,-2.5) {$n = |w| - m + 2$};
 \node at (5.3,-2.9) {$\Biggr]$};
 \node at (5.8,-2.9) {unnec.};
\end{tikzpicture}}
\end{minipage}
 \caption{How matching trials proceed: our previous BM-type algorithm (on the left) and our current FJS-type algorithm (on the right). }
 \label{fig:FJSIsOnlineBMIsNot}
\vspace{1.5em}
\begin{minipage}[t]{0.47\textwidth}
  \centering
   \scalebox{0.7}{
 \begin{tikzpicture}[shorten >=1pt,node distance=2cm,on grid,auto]
 \node[state,initial] (s_0) {$s_0$};
 \node[state] (s_1) [right of=s_0] {$s_1$};
 \node[state] (s_2) [right of=s_1] {$s_2$};
 \node[state,accepting] (s_3) [right of=s_2]{$s_3$};

 \path[->] 
  (s_0) edge [above] node {$\text{a},\mathbf{true}$} (s_1)
  (s_1) edge [above] node {$\text{b},\mathbf{true}$} (s_2)
  (s_2) edge [above] node {$\$,\mathbf{true}$} (s_3);
 \end{tikzpicture}}
 \caption{\textsc{Simple} from~\cite{DBLP:conf/formats/WagaAH16}.
 The set $W$ consists of alternations of \text{a} and \text{b} whose length is from 20 to 1,024,000. Timing is random.
}
 \label{fig:case1_pattern}
  \vspace{1em}
  \centering
   \scalebox{0.9}{
\begin{math}
 \bigl\langle\bigl(\,\bigl(\langle\mathrm{p}\cdot\mathrm{\neg p}\rangle_{(0,10]}\bigr)^*\land\bigl(\langle\mathrm{q}\cdot\mathrm{\neg
    q}\rangle_{(0,10]}\bigr)^*\,\bigr)\cdot \$\bigr\rangle_{(0,80]}
\end{math}}
 \caption{\textsc{Large Constraints} from~\cite{DBLP:conf/formats/WagaAH16}. The pattern $\mathcal{A}$ is a translation of the above timed regular expression (5 states and 9 transitions). The  set $W$ consists of 
 superpositions of the  alternations $p,\neg p,p,\neg p,\dots$ and $q,\neg q,q,\neg q,\dots$ whose timing follows a certain exponential distribution. The length of words in $W$ is from 1,934 to
   31,935. The pattern $\mathcal{A}$ is in Fig.~\ref{fig:case2_automaton}
}
 \label{fig:case2_pattern}
  \vspace{1em}
  \centering
 \scalebox{0.57}{
 \begin{tikzpicture}[shorten >=1pt,node distance=2.5cm,on grid,auto] 
    \node[state,initial] (s_0)  {$s_0$}; 
    \node[state,node distance=2cm] (s_1) [right=of s_0] {$s_1$}; 
    \node[state] (s_2) [right=of s_1] {$s_2$};
    \node[state] (s_3) [right=of s_2] {$s_3$};
    \node[state,node distance=2cm] (s_4) [below=of s_3] {$s_4$}; 
    \node[state,node distance=2.5cm] (s_5) [left=of s_4] {$s_5$}; 
    \node[state,node distance=2cm] (s_6) [left=of s_5] {$s_6$};
    \node[state,accepting,node distance=2cm] (s_7) [left=of s_6] {$s_7$};
    \path[->] 
    (s_0) edge  [above] node {\begin{tabular}{c}
                               $\textrm{low},\mathbf{true}$\\
                               $/x := 0$
                              \end{tabular}} (s_1)
    (s_1) edge  [above] node {\begin{tabular}{c}
                               $\textrm{high},$\\
                               $0 < x < 1$
                              \end{tabular}} (s_2)
    (s_2) edge  [above] node {\begin{tabular}{c}
                               $\textrm{high},$\\
                               $0 < x < 1$
                              \end{tabular}} (s_3)
    (s_3) edge  [right] node {\begin{tabular}{c}
                               $\textrm{high},$\\
                               $0 < x < 1$
                              \end{tabular}} (s_4)
    (s_4) edge  [above] node {\begin{tabular}{c}
                               $\textrm{high},$\\
                               $0 < x < 1$
                              \end{tabular}} (s_5)
    (s_5) edge  [above] node {\begin{tabular}{c}
                               $\textrm{high},$\\
                               $1 < x$
                              \end{tabular}} (s_6)
    (s_5) edge  [above, loop above] node {$\textrm{high},\mathbf{true}$} (s_5)
    (s_6) edge  [above] node {$\$,\mathbf{true}$} (s_7);
 \end{tikzpicture}}

 \caption{\textsc{Torque}, an automotive example
from~\cite{DBLP:conf/formats/WagaAH16}.
It monitors for
 five or more consecutive occurrences
of $\textrm{high}$ in one second.
 The target words in $W$ (length 242,808--4,873,207) are generated by 
the model
\texttt{sldemo\_enginewc.slx} in the Simulink Demo
palette~\cite{SimulinkGuide} with random input.
}
 \label{fig:case3_pattern}
\end{minipage}
\quad
\begin{minipage}[t]{0.5\textwidth}
  \centering
  \scalebox{0.57}{
  \begin{tikzpicture}[shorten >=1pt,node distance=2.5cm,on grid,auto] 
    \node[state,initial] (s_0)  {}; 
    \node[state,node distance=3.0cm] (s_1) [right=of s_0] {nml};
    \node[state,node distance=2.3cm] (s_2) [right=of s_1] {unstl};
    \node[state,accepting,node distance=2.7cm] (s_3) [right=of s_2] {\cmark};
    \path[->] 
    (s_0) edge  [above] node {$\textrm{normal} /x := 0$} (s_1)
    (s_1) edge  [above] node {$\textrm{unsettled}
                                $} (s_2)
    (s_2) edge  [above] node {$\$, x>100$} (s_3);
  \end{tikzpicture}}
  \caption{\textsc{Settling}. The set $W$ (length 472--47,200,000) is generated by 
  the  Simulink powertrain model  in~\cite{DBLP:conf/hybrid/JinDKUB14}.  
The pattern
 (Requirement (32) in~\cite{DBLP:conf/hybrid/JinDKUB14}) is for an event in which the system remains unsettled  for 100 seconds after moving to the normal mode.
}
  \label{fig:case4_pattern}
  \vspace{1em}
 \centering
  \scalebox{0.6}{
  \begin{tikzpicture}[shorten >=1pt,node distance=2.5cm,on grid,auto] 
    \node[state,initial] (s_0)  {}; 
    \node[state,node distance=2.5cm] (s_1) [right=of s_0] {$\text{g}_1$}; 
    \node[state,node distance=2.5cm] (s_2) [right=of s_1] {$\text{g}_2$};
    \node[state,accepting,node distance=1.5cm] (s_3) [right=of s_2] {\cmark};
    \path[->] 
    (s_0) edge  [above] node {$\text{g}_1 /x := 0$} (s_1)
    (s_1) edge  [above] node {$\text{g}_2, x < 2$} (s_2)
    (s_2) edge  [above] node {$\$$} (s_3);
  \end{tikzpicture}}
  \caption{\textsc{Gear}. The set $W$ (length 307--1,011,427) is generated by the automatic transmission system model in~\cite{DBLP:conf/cpsweek/HoxhaAF14}. The pattern, from $\phi^{\mathit{AT}}_5$ in~\cite{DBLP:conf/cpsweek/HoxhaAF14}, is for an event in which gear shift occurs too quickly (from the 1st to 2nd).
}
  \label{fig:case5_pattern}
  \vspace{1em}
  \centering
  \scalebox{0.45}{
  \begin{tikzpicture}[on grid,auto,node distance=2.5cm]
   \node[state, initial] (s_000) {$?$};

   \node[state] (s_100)[above right=of s_000] {$\text{g}_1$};
   \node[state] (s_001)[below right=of s_000] {$?$};

   \node[state] (s_200)[right=of s_100] {$\text{g}_2$};
   \node[state] (s_101)[right=of s_001] {$\text{g}_1$};

   \node[state] (s_300)[right=of s_200] {$\text{g}_3$};
   \node[state] (s_201)[right=of s_101] {$\text{g}_2$};

   \node[state] (s_400)[right=of s_300] {$\text{g}_4$};
   \node[state] (s_301)[right=of s_201] {$\text{g}_3$};

   \node[state] (s_401)[right=of s_301] {$\text{g}_4$};
   \node[state,accepting] (f)[right=of s_400] {\cmark};

   \path[->]
   (s_000) edge  [above left] node {$\text{g}_1, \mathbf{true}$} (s_100)
   (s_100) edge  [above] node {$\text{g}_2, \mathbf{true}$} (s_200)
   (s_200) edge  [above] node {$\text{g}_3, \mathbf{true}$} (s_300)
   (s_300) edge  [above] node {\begin{tabular}{c}
                                $\text{g}_4, x \leq 10$\\
                                $/x := 0$
                               \end{tabular}} (s_400)

   (s_100) edge  [below left] node {$\mathrm{rpmHigh}, \mathbf{true}$} (s_101)
   (s_200) edge  [below left] node {$\mathrm{rpmHigh}, \mathbf{true}$} (s_201)
   (s_300) edge  [below left] node {$\mathrm{rpmHigh}, \mathbf{true}$} (s_301)
   (s_400) edge  [below left] node {$\mathrm{rpmHigh}, \mathbf{true}$} (s_401)

   (s_001) edge  [above] node {$\text{g}_1, \mathbf{true}$} (s_101)
   (s_101) edge  [above] node {$\text{g}_2, \mathbf{true}$} (s_201)
   (s_201) edge  [above] node {$\text{g}_3, \mathbf{true}$} (s_301)
   (s_301) edge  [below] node {\begin{tabular}{c}
                                $\text{g}_4, x \leq 10$\\
                                $/x := 0$
                               \end{tabular}} (s_401)




   (s_000) edge  [below left] node {$\mathrm{rpmHigh}, \mathbf{true}$} (s_001)
   %

   (s_401) [bend right=0] edge [right] node {$\$, x > 1$} (f);
  \end{tikzpicture}}
  \caption{\textsc{Accel}.
The set $W$ (length 25,002--17,280,002) is generated by the same automatic transmission system model as in \textsc{Gear}. The pattern
is from
 $\phi^{\mathit{AT}}_8$ 
in~\cite{DBLP:conf/cpsweek/HoxhaAF14}: although the gear shifts from 1st to 4th and RPM is high enough somewhere in its course, the vehicle velocity is not high enough (i.e.\ the character veloHigh is absent). 
 }
  \label{fig:case6_pattern}
 \end{minipage}
\end{figure}

\subsection{Comparison with the Brute Force and BM-Type Algorithms}
\label{subsec:comparisonWithBFAndBM}

\begin{figure}[tbp]
\setlength\abovecaptionskip{0pt}
\begin{minipage}{.48\textwidth}
  \centering
  \scalebox{0.4}{
  \begin{tikzpicture}[gnuplot]
\path (0.000,0.000) rectangle (12.500,8.750);
\gpcolor{color=gp lt color border}
\gpsetlinetype{gp lt border}
\gpsetdashtype{gp dt solid}
\gpsetlinewidth{1.00}
\draw[gp path] (1.320,0.985)--(1.500,0.985);
\draw[gp path] (11.947,0.985)--(11.767,0.985);
\node[gp node right] at (1.136,0.985) {$0$};
\draw[gp path] (1.320,2.218)--(1.500,2.218);
\draw[gp path] (11.947,2.218)--(11.767,2.218);
\node[gp node right] at (1.136,2.218) {$20$};
\draw[gp path] (1.320,3.450)--(1.500,3.450);
\draw[gp path] (11.947,3.450)--(11.767,3.450);
\node[gp node right] at (1.136,3.450) {$40$};
\draw[gp path] (1.320,4.683)--(1.500,4.683);
\draw[gp path] (11.947,4.683)--(11.767,4.683);
\node[gp node right] at (1.136,4.683) {$60$};
\draw[gp path] (1.320,5.916)--(1.500,5.916);
\draw[gp path] (11.947,5.916)--(11.767,5.916);
\node[gp node right] at (1.136,5.916) {$80$};
\draw[gp path] (1.320,7.148)--(1.500,7.148);
\draw[gp path] (11.947,7.148)--(11.767,7.148);
\node[gp node right] at (1.136,7.148) {$100$};
\draw[gp path] (1.320,8.381)--(1.500,8.381);
\draw[gp path] (11.947,8.381)--(11.767,8.381);
\node[gp node right] at (1.136,8.381) {$120$};
\draw[gp path] (1.320,0.985)--(1.320,1.165);
\draw[gp path] (1.320,8.381)--(1.320,8.201);
\node[gp node center] at (1.320,0.677) {$0$};
\draw[gp path] (3.091,0.985)--(3.091,1.165);
\draw[gp path] (3.091,8.381)--(3.091,8.201);
\node[gp node center] at (3.091,0.677) {$20$};
\draw[gp path] (4.862,0.985)--(4.862,1.165);
\draw[gp path] (4.862,8.381)--(4.862,8.201);
\node[gp node center] at (4.862,0.677) {$40$};
\draw[gp path] (6.634,0.985)--(6.634,1.165);
\draw[gp path] (6.634,8.381)--(6.634,8.201);
\node[gp node center] at (6.634,0.677) {$60$};
\draw[gp path] (8.405,0.985)--(8.405,1.165);
\draw[gp path] (8.405,8.381)--(8.405,8.201);
\node[gp node center] at (8.405,0.677) {$80$};
\draw[gp path] (10.176,0.985)--(10.176,1.165);
\draw[gp path] (10.176,8.381)--(10.176,8.201);
\node[gp node center] at (10.176,0.677) {$100$};
\draw[gp path] (11.947,0.985)--(11.947,1.165);
\draw[gp path] (11.947,8.381)--(11.947,8.201);
\node[gp node center] at (11.947,0.677) {$120$};
\draw[gp path] (1.320,8.381)--(1.320,0.985)--(11.947,0.985)--(11.947,8.381)--cycle;
\node[gp node center,rotate=-270] at (0.246,4.683) {Execution Time [ms]};
\node[gp node center] at (6.633,0.215) {Number of Events [$\times 10000$]};
\node[gp node right] at (10.479,8.047) {brute-force};
\gpcolor{rgb color={0.580,0.000,0.827}}
\draw[gp path] (10.663,8.047)--(11.579,8.047);
\draw[gp path] (1.320,0.986)--(1.321,0.987)--(1.321,0.988)--(1.323,0.988)--(1.326,0.990)%
  --(1.329,0.992)--(1.331,0.995)--(1.338,1.001)--(1.343,1.008)--(1.347,1.048)--(1.355,1.023)%
  --(1.364,1.031)--(1.365,1.032)--(1.373,1.040)--(1.382,1.062)--(1.391,1.042)--(1.400,1.108)%
  --(1.409,1.122)--(1.411,1.074)--(1.462,1.113)--(1.603,1.207)--(1.887,1.442)--(2.206,1.705)%
  --(2.454,1.797)--(3.091,2.346)--(3.587,2.632)--(3.977,3.440)--(4.862,3.852)--(5.748,4.277)%
  --(5.854,4.392)--(6.634,5.659)--(7.519,6.114)--(8.405,6.628)--(9.290,7.012)--(10.176,7.657)%
  --(10.388,7.644);
\gpsetpointsize{4.00}
\gppoint{gp mark 1}{(1.320,0.986)}
\gppoint{gp mark 1}{(1.320,0.986)}
\gppoint{gp mark 1}{(1.321,0.987)}
\gppoint{gp mark 1}{(1.321,0.988)}
\gppoint{gp mark 1}{(1.323,0.988)}
\gppoint{gp mark 1}{(1.326,0.990)}
\gppoint{gp mark 1}{(1.329,0.992)}
\gppoint{gp mark 1}{(1.331,0.995)}
\gppoint{gp mark 1}{(1.338,1.001)}
\gppoint{gp mark 1}{(1.343,1.008)}
\gppoint{gp mark 1}{(1.347,1.048)}
\gppoint{gp mark 1}{(1.355,1.023)}
\gppoint{gp mark 1}{(1.364,1.031)}
\gppoint{gp mark 1}{(1.365,1.032)}
\gppoint{gp mark 1}{(1.373,1.040)}
\gppoint{gp mark 1}{(1.382,1.062)}
\gppoint{gp mark 1}{(1.391,1.042)}
\gppoint{gp mark 1}{(1.400,1.108)}
\gppoint{gp mark 1}{(1.409,1.122)}
\gppoint{gp mark 1}{(1.411,1.074)}
\gppoint{gp mark 1}{(1.462,1.113)}
\gppoint{gp mark 1}{(1.603,1.207)}
\gppoint{gp mark 1}{(1.887,1.442)}
\gppoint{gp mark 1}{(2.206,1.705)}
\gppoint{gp mark 1}{(2.454,1.797)}
\gppoint{gp mark 1}{(3.091,2.346)}
\gppoint{gp mark 1}{(3.587,2.632)}
\gppoint{gp mark 1}{(3.977,3.440)}
\gppoint{gp mark 1}{(4.862,3.852)}
\gppoint{gp mark 1}{(5.748,4.277)}
\gppoint{gp mark 1}{(5.854,4.392)}
\gppoint{gp mark 1}{(6.634,5.659)}
\gppoint{gp mark 1}{(7.519,6.114)}
\gppoint{gp mark 1}{(8.405,6.628)}
\gppoint{gp mark 1}{(9.290,7.012)}
\gppoint{gp mark 1}{(10.176,7.657)}
\gppoint{gp mark 1}{(10.388,7.644)}
\gppoint{gp mark 1}{(11.121,8.047)}
\gpcolor{color=gp lt color border}
\node[gp node right] at (10.479,7.739) {BM};
\gpcolor{rgb color={0.000,0.620,0.451}}
\draw[gp path] (10.663,7.739)--(11.579,7.739);
\draw[gp path] (1.320,0.986)--(1.321,0.986)--(1.321,0.987)--(1.323,0.988)--(1.326,0.990)%
  --(1.329,0.992)--(1.331,0.995)--(1.338,1.003)--(1.343,1.007)--(1.347,1.070)--(1.355,1.015)%
  --(1.364,1.037)--(1.365,1.023)--(1.373,1.026)--(1.382,1.038)--(1.391,1.048)--(1.400,1.115)%
  --(1.409,1.070)--(1.411,1.101)--(1.462,1.127)--(1.603,1.211)--(1.887,1.350)--(2.206,1.608)%
  --(2.454,1.680)--(3.091,2.249)--(3.587,2.471)--(3.977,3.059)--(4.862,3.444)--(5.748,3.820)%
  --(5.854,3.861)--(6.634,5.160)--(7.519,5.603)--(8.405,5.909)--(9.290,6.328)--(10.176,6.626)%
  --(10.388,6.842);
\gppoint{gp mark 2}{(1.320,0.986)}
\gppoint{gp mark 2}{(1.320,0.986)}
\gppoint{gp mark 2}{(1.321,0.986)}
\gppoint{gp mark 2}{(1.321,0.987)}
\gppoint{gp mark 2}{(1.323,0.988)}
\gppoint{gp mark 2}{(1.326,0.990)}
\gppoint{gp mark 2}{(1.329,0.992)}
\gppoint{gp mark 2}{(1.331,0.995)}
\gppoint{gp mark 2}{(1.338,1.003)}
\gppoint{gp mark 2}{(1.343,1.007)}
\gppoint{gp mark 2}{(1.347,1.070)}
\gppoint{gp mark 2}{(1.355,1.015)}
\gppoint{gp mark 2}{(1.364,1.037)}
\gppoint{gp mark 2}{(1.365,1.023)}
\gppoint{gp mark 2}{(1.373,1.026)}
\gppoint{gp mark 2}{(1.382,1.038)}
\gppoint{gp mark 2}{(1.391,1.048)}
\gppoint{gp mark 2}{(1.400,1.115)}
\gppoint{gp mark 2}{(1.409,1.070)}
\gppoint{gp mark 2}{(1.411,1.101)}
\gppoint{gp mark 2}{(1.462,1.127)}
\gppoint{gp mark 2}{(1.603,1.211)}
\gppoint{gp mark 2}{(1.887,1.350)}
\gppoint{gp mark 2}{(2.206,1.608)}
\gppoint{gp mark 2}{(2.454,1.680)}
\gppoint{gp mark 2}{(3.091,2.249)}
\gppoint{gp mark 2}{(3.587,2.471)}
\gppoint{gp mark 2}{(3.977,3.059)}
\gppoint{gp mark 2}{(4.862,3.444)}
\gppoint{gp mark 2}{(5.748,3.820)}
\gppoint{gp mark 2}{(5.854,3.861)}
\gppoint{gp mark 2}{(6.634,5.160)}
\gppoint{gp mark 2}{(7.519,5.603)}
\gppoint{gp mark 2}{(8.405,5.909)}
\gppoint{gp mark 2}{(9.290,6.328)}
\gppoint{gp mark 2}{(10.176,6.626)}
\gppoint{gp mark 2}{(10.388,6.842)}
\gppoint{gp mark 2}{(11.121,7.739)}
\gpcolor{color=gp lt color border}
\node[gp node right] at (10.479,7.431) {FJS};
\gpcolor{rgb color={1.000,0.000,0.000}}
\draw[gp path] (10.663,7.431)--(11.579,7.431);
\draw[gp path] (1.320,0.986)--(1.321,0.987)--(1.323,0.989)--(1.326,0.992)--(1.329,0.992)%
  --(1.331,0.995)--(1.338,1.000)--(1.343,1.012)--(1.347,1.012)--(1.355,1.014)--(1.364,1.103)%
  --(1.365,1.037)--(1.373,1.035)--(1.382,1.043)--(1.391,1.061)--(1.400,1.058)--(1.409,1.064)%
  --(1.411,1.071)--(1.462,1.098)--(1.603,1.217)--(1.887,1.422)--(2.206,1.687)--(2.454,1.736)%
  --(3.091,2.305)--(3.587,2.532)--(3.977,3.207)--(4.862,3.593)--(5.748,4.185)--(5.854,4.105)%
  --(6.634,5.443)--(7.519,5.828)--(8.405,6.245)--(9.290,6.699)--(10.176,7.213)--(10.388,7.259);
\gppoint{gp mark 3}{(1.320,0.986)}
\gppoint{gp mark 3}{(1.320,0.986)}
\gppoint{gp mark 3}{(1.321,0.987)}
\gppoint{gp mark 3}{(1.321,0.987)}
\gppoint{gp mark 3}{(1.323,0.989)}
\gppoint{gp mark 3}{(1.326,0.992)}
\gppoint{gp mark 3}{(1.329,0.992)}
\gppoint{gp mark 3}{(1.331,0.995)}
\gppoint{gp mark 3}{(1.338,1.000)}
\gppoint{gp mark 3}{(1.343,1.012)}
\gppoint{gp mark 3}{(1.347,1.012)}
\gppoint{gp mark 3}{(1.355,1.014)}
\gppoint{gp mark 3}{(1.364,1.103)}
\gppoint{gp mark 3}{(1.365,1.037)}
\gppoint{gp mark 3}{(1.373,1.035)}
\gppoint{gp mark 3}{(1.382,1.043)}
\gppoint{gp mark 3}{(1.391,1.061)}
\gppoint{gp mark 3}{(1.400,1.058)}
\gppoint{gp mark 3}{(1.409,1.064)}
\gppoint{gp mark 3}{(1.411,1.071)}
\gppoint{gp mark 3}{(1.462,1.098)}
\gppoint{gp mark 3}{(1.603,1.217)}
\gppoint{gp mark 3}{(1.887,1.422)}
\gppoint{gp mark 3}{(2.206,1.687)}
\gppoint{gp mark 3}{(2.454,1.736)}
\gppoint{gp mark 3}{(3.091,2.305)}
\gppoint{gp mark 3}{(3.587,2.532)}
\gppoint{gp mark 3}{(3.977,3.207)}
\gppoint{gp mark 3}{(4.862,3.593)}
\gppoint{gp mark 3}{(5.748,4.185)}
\gppoint{gp mark 3}{(5.854,4.105)}
\gppoint{gp mark 3}{(6.634,5.443)}
\gppoint{gp mark 3}{(7.519,5.828)}
\gppoint{gp mark 3}{(8.405,6.245)}
\gppoint{gp mark 3}{(9.290,6.699)}
\gppoint{gp mark 3}{(10.176,7.213)}
\gppoint{gp mark 3}{(10.388,7.259)}
\gppoint{gp mark 3}{(11.121,7.431)}
\gpcolor{color=gp lt color border}
\draw[gp path] (1.320,8.381)--(1.320,0.985)--(11.947,0.985)--(11.947,8.381)--cycle;
\gpdefrectangularnode{gp plot 1}{\pgfpoint{1.320cm}{0.985cm}}{\pgfpoint{11.947cm}{8.381cm}}
\end{tikzpicture}
  }
 \caption{\textsc{Simple}: exec.\ time}
 \label{fig:case1_exec_time}
\end{minipage}
\hfill
\begin{minipage}{.48\textwidth}
  \centering
 \scalebox{0.4}{
 \begin{tikzpicture}[gnuplot]
\path (0.000,0.000) rectangle (12.500,8.750);
\gpcolor{color=gp lt color border}
\gpsetlinetype{gp lt border}
\gpsetdashtype{gp dt solid}
\gpsetlinewidth{1.00}
\draw[gp path] (1.136,0.985)--(1.316,0.985);
\draw[gp path] (11.947,0.985)--(11.767,0.985);
\node[gp node right] at (0.952,0.985) {$0$};
\draw[gp path] (1.136,2.042)--(1.316,2.042);
\draw[gp path] (11.947,2.042)--(11.767,2.042);
\node[gp node right] at (0.952,2.042) {$2$};
\draw[gp path] (1.136,3.098)--(1.316,3.098);
\draw[gp path] (11.947,3.098)--(11.767,3.098);
\node[gp node right] at (0.952,3.098) {$4$};
\draw[gp path] (1.136,4.155)--(1.316,4.155);
\draw[gp path] (11.947,4.155)--(11.767,4.155);
\node[gp node right] at (0.952,4.155) {$6$};
\draw[gp path] (1.136,5.211)--(1.316,5.211);
\draw[gp path] (11.947,5.211)--(11.767,5.211);
\node[gp node right] at (0.952,5.211) {$8$};
\draw[gp path] (1.136,6.268)--(1.316,6.268);
\draw[gp path] (11.947,6.268)--(11.767,6.268);
\node[gp node right] at (0.952,6.268) {$10$};
\draw[gp path] (1.136,7.324)--(1.316,7.324);
\draw[gp path] (11.947,7.324)--(11.767,7.324);
\node[gp node right] at (0.952,7.324) {$12$};
\draw[gp path] (1.136,8.381)--(1.316,8.381);
\draw[gp path] (11.947,8.381)--(11.767,8.381);
\node[gp node right] at (0.952,8.381) {$14$};
\draw[gp path] (1.136,0.985)--(1.136,1.165);
\draw[gp path] (1.136,8.381)--(1.136,8.201);
\node[gp node center] at (1.136,0.677) {$0$};
\draw[gp path] (2.680,0.985)--(2.680,1.165);
\draw[gp path] (2.680,8.381)--(2.680,8.201);
\node[gp node center] at (2.680,0.677) {$5$};
\draw[gp path] (4.225,0.985)--(4.225,1.165);
\draw[gp path] (4.225,8.381)--(4.225,8.201);
\node[gp node center] at (4.225,0.677) {$10$};
\draw[gp path] (5.769,0.985)--(5.769,1.165);
\draw[gp path] (5.769,8.381)--(5.769,8.201);
\node[gp node center] at (5.769,0.677) {$15$};
\draw[gp path] (7.314,0.985)--(7.314,1.165);
\draw[gp path] (7.314,8.381)--(7.314,8.201);
\node[gp node center] at (7.314,0.677) {$20$};
\draw[gp path] (8.858,0.985)--(8.858,1.165);
\draw[gp path] (8.858,8.381)--(8.858,8.201);
\node[gp node center] at (8.858,0.677) {$25$};
\draw[gp path] (10.403,0.985)--(10.403,1.165);
\draw[gp path] (10.403,8.381)--(10.403,8.201);
\node[gp node center] at (10.403,0.677) {$30$};
\draw[gp path] (11.947,0.985)--(11.947,1.165);
\draw[gp path] (11.947,8.381)--(11.947,8.201);
\node[gp node center] at (11.947,0.677) {$35$};
\draw[gp path] (1.136,8.381)--(1.136,0.985)--(11.947,0.985)--(11.947,8.381)--cycle;
\node[gp node center,rotate=-270] at (0.246,4.683) {Execution Time [ms]};
\node[gp node center] at (6.541,0.215) {Number of Events [$\times 1000$]};
\node[gp node right] at (10.479,8.047) {brute-force};
\gpcolor{rgb color={0.580,0.000,0.827}}
\draw[gp path] (10.663,8.047)--(11.579,8.047);
\draw[gp path] (1.733,1.902)--(2.368,3.631)--(2.374,1.764)--(3.023,2.170)--(3.591,3.304)%
  --(3.600,3.462)--(3.650,2.318)--(4.891,2.949)--(6.037,5.307)--(6.058,5.674)--(8.545,5.755)%
  --(11.000,7.019);
\gpsetpointsize{4.00}
\gppoint{gp mark 1}{(1.733,1.902)}
\gppoint{gp mark 1}{(2.368,3.631)}
\gppoint{gp mark 1}{(2.374,1.764)}
\gppoint{gp mark 1}{(3.023,2.170)}
\gppoint{gp mark 1}{(3.591,3.304)}
\gppoint{gp mark 1}{(3.600,3.462)}
\gppoint{gp mark 1}{(3.650,2.318)}
\gppoint{gp mark 1}{(4.891,2.949)}
\gppoint{gp mark 1}{(6.037,5.307)}
\gppoint{gp mark 1}{(6.058,5.674)}
\gppoint{gp mark 1}{(8.545,5.755)}
\gppoint{gp mark 1}{(11.000,7.019)}
\gppoint{gp mark 1}{(11.121,8.047)}
\gpcolor{color=gp lt color border}
\node[gp node right] at (10.479,7.739) {BM};
\gpcolor{rgb color={0.000,0.620,0.451}}
\draw[gp path] (10.663,7.739)--(11.579,7.739);
\draw[gp path] (1.733,1.511)--(2.368,2.296)--(2.374,1.851)--(3.023,2.159)--(3.591,2.656)%
  --(3.600,2.796)--(3.650,2.246)--(4.891,2.566)--(6.037,3.426)--(6.058,4.611)--(8.545,4.754)%
  --(11.000,7.786);
\gppoint{gp mark 2}{(1.733,1.511)}
\gppoint{gp mark 2}{(2.368,2.296)}
\gppoint{gp mark 2}{(2.374,1.851)}
\gppoint{gp mark 2}{(3.023,2.159)}
\gppoint{gp mark 2}{(3.591,2.656)}
\gppoint{gp mark 2}{(3.600,2.796)}
\gppoint{gp mark 2}{(3.650,2.246)}
\gppoint{gp mark 2}{(4.891,2.566)}
\gppoint{gp mark 2}{(6.037,3.426)}
\gppoint{gp mark 2}{(6.058,4.611)}
\gppoint{gp mark 2}{(8.545,4.754)}
\gppoint{gp mark 2}{(11.000,7.786)}
\gppoint{gp mark 2}{(11.121,7.739)}
\gpcolor{color=gp lt color border}
\node[gp node right] at (10.479,7.431) {FJS};
\gpcolor{rgb color={1.000,0.000,0.000}}
\draw[gp path] (10.663,7.431)--(11.579,7.431);
\draw[gp path] (1.733,1.340)--(2.368,1.511)--(2.374,1.472)--(3.023,2.237)--(3.591,2.148)%
  --(3.600,2.274)--(3.650,2.980)--(4.891,3.668)--(6.037,4.033)--(6.058,3.481)--(8.545,7.557)%
  --(11.000,7.825);
\gppoint{gp mark 3}{(1.733,1.340)}
\gppoint{gp mark 3}{(2.368,1.511)}
\gppoint{gp mark 3}{(2.374,1.472)}
\gppoint{gp mark 3}{(3.023,2.237)}
\gppoint{gp mark 3}{(3.591,2.148)}
\gppoint{gp mark 3}{(3.600,2.274)}
\gppoint{gp mark 3}{(3.650,2.980)}
\gppoint{gp mark 3}{(4.891,3.668)}
\gppoint{gp mark 3}{(6.037,4.033)}
\gppoint{gp mark 3}{(6.058,3.481)}
\gppoint{gp mark 3}{(8.545,7.557)}
\gppoint{gp mark 3}{(11.000,7.825)}
\gppoint{gp mark 3}{(11.121,7.431)}
\gpcolor{color=gp lt color border}
\draw[gp path] (1.136,8.381)--(1.136,0.985)--(11.947,0.985)--(11.947,8.381)--cycle;
\gpdefrectangularnode{gp plot 1}{\pgfpoint{1.136cm}{0.985cm}}{\pgfpoint{11.947cm}{8.381cm}}
\end{tikzpicture}
  }
 \caption{\textsc{Large Constraints}: exec.\ time}
 \label{fig:case2_exec_time}
\end{minipage}
\begin{minipage}{0.48\textwidth}
  \centering
  \scalebox{0.4}{
  \begin{tikzpicture}[gnuplot]
\path (0.000,0.000) rectangle (12.500,8.750);
\gpcolor{color=gp lt color border}
\gpsetlinetype{gp lt border}
\gpsetdashtype{gp dt solid}
\gpsetlinewidth{1.00}
\draw[gp path] (1.320,0.985)--(1.500,0.985);
\draw[gp path] (11.947,0.985)--(11.767,0.985);
\node[gp node right] at (1.136,0.985) {$0$};
\draw[gp path] (1.320,2.464)--(1.500,2.464);
\draw[gp path] (11.947,2.464)--(11.767,2.464);
\node[gp node right] at (1.136,2.464) {$50$};
\draw[gp path] (1.320,3.943)--(1.500,3.943);
\draw[gp path] (11.947,3.943)--(11.767,3.943);
\node[gp node right] at (1.136,3.943) {$100$};
\draw[gp path] (1.320,5.423)--(1.500,5.423);
\draw[gp path] (11.947,5.423)--(11.767,5.423);
\node[gp node right] at (1.136,5.423) {$150$};
\draw[gp path] (1.320,6.902)--(1.500,6.902);
\draw[gp path] (11.947,6.902)--(11.767,6.902);
\node[gp node right] at (1.136,6.902) {$200$};
\draw[gp path] (1.320,8.381)--(1.500,8.381);
\draw[gp path] (11.947,8.381)--(11.767,8.381);
\node[gp node right] at (1.136,8.381) {$250$};
\draw[gp path] (1.320,0.985)--(1.320,1.165);
\draw[gp path] (1.320,8.381)--(1.320,8.201);
\node[gp node center] at (1.320,0.677) {$0$};
\draw[gp path] (2.383,0.985)--(2.383,1.165);
\draw[gp path] (2.383,8.381)--(2.383,8.201);
\node[gp node center] at (2.383,0.677) {$50$};
\draw[gp path] (3.445,0.985)--(3.445,1.165);
\draw[gp path] (3.445,8.381)--(3.445,8.201);
\node[gp node center] at (3.445,0.677) {$100$};
\draw[gp path] (4.508,0.985)--(4.508,1.165);
\draw[gp path] (4.508,8.381)--(4.508,8.201);
\node[gp node center] at (4.508,0.677) {$150$};
\draw[gp path] (5.571,0.985)--(5.571,1.165);
\draw[gp path] (5.571,8.381)--(5.571,8.201);
\node[gp node center] at (5.571,0.677) {$200$};
\draw[gp path] (6.634,0.985)--(6.634,1.165);
\draw[gp path] (6.634,8.381)--(6.634,8.201);
\node[gp node center] at (6.634,0.677) {$250$};
\draw[gp path] (7.696,0.985)--(7.696,1.165);
\draw[gp path] (7.696,8.381)--(7.696,8.201);
\node[gp node center] at (7.696,0.677) {$300$};
\draw[gp path] (8.759,0.985)--(8.759,1.165);
\draw[gp path] (8.759,8.381)--(8.759,8.201);
\node[gp node center] at (8.759,0.677) {$350$};
\draw[gp path] (9.822,0.985)--(9.822,1.165);
\draw[gp path] (9.822,8.381)--(9.822,8.201);
\node[gp node center] at (9.822,0.677) {$400$};
\draw[gp path] (10.884,0.985)--(10.884,1.165);
\draw[gp path] (10.884,8.381)--(10.884,8.201);
\node[gp node center] at (10.884,0.677) {$450$};
\draw[gp path] (11.947,0.985)--(11.947,1.165);
\draw[gp path] (11.947,8.381)--(11.947,8.201);
\node[gp node center] at (11.947,0.677) {$500$};
\draw[gp path] (1.320,8.381)--(1.320,0.985)--(11.947,0.985)--(11.947,8.381)--cycle;
\node[gp node center,rotate=-270] at (0.246,4.683) {Execution Time [ms]};
\node[gp node center] at (6.633,0.215) {Number of Events [$\times 10000$]};
\node[gp node right] at (10.479,8.047) {brute-force};
\gpcolor{rgb color={0.580,0.000,0.827}}
\draw[gp path] (10.663,8.047)--(11.579,8.047);
\draw[gp path] (1.836,1.313)--(2.355,1.646)--(2.876,1.969)--(3.915,2.622)--(5.210,3.470)%
  --(6.500,4.277)--(7.794,5.070)--(9.088,5.900)--(10.385,6.844)--(11.678,7.713);
\gpsetpointsize{4.00}
\gppoint{gp mark 1}{(1.836,1.313)}
\gppoint{gp mark 1}{(2.355,1.646)}
\gppoint{gp mark 1}{(2.876,1.969)}
\gppoint{gp mark 1}{(3.915,2.622)}
\gppoint{gp mark 1}{(5.210,3.470)}
\gppoint{gp mark 1}{(6.500,4.277)}
\gppoint{gp mark 1}{(7.794,5.070)}
\gppoint{gp mark 1}{(9.088,5.900)}
\gppoint{gp mark 1}{(10.385,6.844)}
\gppoint{gp mark 1}{(11.678,7.713)}
\gppoint{gp mark 1}{(11.121,8.047)}
\gpcolor{color=gp lt color border}
\node[gp node right] at (10.479,7.739) {BM};
\gpcolor{rgb color={0.000,0.620,0.451}}
\draw[gp path] (10.663,7.739)--(11.579,7.739);
\draw[gp path] (1.836,1.149)--(2.355,1.316)--(2.876,1.487)--(3.915,1.826)--(5.210,2.218)%
  --(6.500,2.608)--(7.794,3.052)--(9.088,3.473)--(10.385,3.863)--(11.678,4.293);
\gppoint{gp mark 2}{(1.836,1.149)}
\gppoint{gp mark 2}{(2.355,1.316)}
\gppoint{gp mark 2}{(2.876,1.487)}
\gppoint{gp mark 2}{(3.915,1.826)}
\gppoint{gp mark 2}{(5.210,2.218)}
\gppoint{gp mark 2}{(6.500,2.608)}
\gppoint{gp mark 2}{(7.794,3.052)}
\gppoint{gp mark 2}{(9.088,3.473)}
\gppoint{gp mark 2}{(10.385,3.863)}
\gppoint{gp mark 2}{(11.678,4.293)}
\gppoint{gp mark 2}{(11.121,7.739)}
\gpcolor{color=gp lt color border}
\node[gp node right] at (10.479,7.431) {FJS};
\gpcolor{rgb color={1.000,0.000,0.000}}
\draw[gp path] (10.663,7.431)--(11.579,7.431);
\draw[gp path] (1.836,1.121)--(2.355,1.256)--(2.876,1.396)--(3.915,1.684)--(5.210,1.993)%
  --(6.500,2.320)--(7.794,2.685)--(9.088,2.970)--(10.385,3.322)--(11.678,3.706);
\gppoint{gp mark 3}{(1.836,1.121)}
\gppoint{gp mark 3}{(2.355,1.256)}
\gppoint{gp mark 3}{(2.876,1.396)}
\gppoint{gp mark 3}{(3.915,1.684)}
\gppoint{gp mark 3}{(5.210,1.993)}
\gppoint{gp mark 3}{(6.500,2.320)}
\gppoint{gp mark 3}{(7.794,2.685)}
\gppoint{gp mark 3}{(9.088,2.970)}
\gppoint{gp mark 3}{(10.385,3.322)}
\gppoint{gp mark 3}{(11.678,3.706)}
\gppoint{gp mark 3}{(11.121,7.431)}
\gpcolor{color=gp lt color border}
\draw[gp path] (1.320,8.381)--(1.320,0.985)--(11.947,0.985)--(11.947,8.381)--cycle;
\gpdefrectangularnode{gp plot 1}{\pgfpoint{1.320cm}{0.985cm}}{\pgfpoint{11.947cm}{8.381cm}}
\end{tikzpicture}
  }
 \caption{\textsc{Torque}: exec.\ time}
 \label{fig:case3_exec_time}
 \end{minipage}
 \hfill
 \begin{minipage}{0.48\textwidth}
  \centering
  \scalebox{0.4}{
  \begin{tikzpicture}[gnuplot]
\path (0.000,0.000) rectangle (12.500,8.750);
\gpcolor{color=gp lt color border}
\gpsetlinetype{gp lt border}
\gpsetdashtype{gp dt solid}
\gpsetlinewidth{1.00}
\draw[gp path] (1.504,0.985)--(1.684,0.985);
\draw[gp path] (11.947,0.985)--(11.767,0.985);
\node[gp node right] at (1.320,0.985) {$0$};
\draw[gp path] (1.504,2.042)--(1.684,2.042);
\draw[gp path] (11.947,2.042)--(11.767,2.042);
\node[gp node right] at (1.320,2.042) {$200$};
\draw[gp path] (1.504,3.098)--(1.684,3.098);
\draw[gp path] (11.947,3.098)--(11.767,3.098);
\node[gp node right] at (1.320,3.098) {$400$};
\draw[gp path] (1.504,4.155)--(1.684,4.155);
\draw[gp path] (11.947,4.155)--(11.767,4.155);
\node[gp node right] at (1.320,4.155) {$600$};
\draw[gp path] (1.504,5.211)--(1.684,5.211);
\draw[gp path] (11.947,5.211)--(11.767,5.211);
\node[gp node right] at (1.320,5.211) {$800$};
\draw[gp path] (1.504,6.268)--(1.684,6.268);
\draw[gp path] (11.947,6.268)--(11.767,6.268);
\node[gp node right] at (1.320,6.268) {$1000$};
\draw[gp path] (1.504,7.324)--(1.684,7.324);
\draw[gp path] (11.947,7.324)--(11.767,7.324);
\node[gp node right] at (1.320,7.324) {$1200$};
\draw[gp path] (1.504,8.381)--(1.684,8.381);
\draw[gp path] (11.947,8.381)--(11.767,8.381);
\node[gp node right] at (1.320,8.381) {$1400$};
\draw[gp path] (1.504,0.985)--(1.504,1.165);
\draw[gp path] (1.504,8.381)--(1.504,8.201);
\node[gp node center] at (1.504,0.677) {$0$};
\draw[gp path] (2.548,0.985)--(2.548,1.165);
\draw[gp path] (2.548,8.381)--(2.548,8.201);
\node[gp node center] at (2.548,0.677) {$500$};
\draw[gp path] (3.593,0.985)--(3.593,1.165);
\draw[gp path] (3.593,8.381)--(3.593,8.201);
\node[gp node center] at (3.593,0.677) {$1000$};
\draw[gp path] (4.637,0.985)--(4.637,1.165);
\draw[gp path] (4.637,8.381)--(4.637,8.201);
\node[gp node center] at (4.637,0.677) {$1500$};
\draw[gp path] (5.681,0.985)--(5.681,1.165);
\draw[gp path] (5.681,8.381)--(5.681,8.201);
\node[gp node center] at (5.681,0.677) {$2000$};
\draw[gp path] (6.726,0.985)--(6.726,1.165);
\draw[gp path] (6.726,8.381)--(6.726,8.201);
\node[gp node center] at (6.726,0.677) {$2500$};
\draw[gp path] (7.770,0.985)--(7.770,1.165);
\draw[gp path] (7.770,8.381)--(7.770,8.201);
\node[gp node center] at (7.770,0.677) {$3000$};
\draw[gp path] (8.814,0.985)--(8.814,1.165);
\draw[gp path] (8.814,8.381)--(8.814,8.201);
\node[gp node center] at (8.814,0.677) {$3500$};
\draw[gp path] (9.858,0.985)--(9.858,1.165);
\draw[gp path] (9.858,8.381)--(9.858,8.201);
\node[gp node center] at (9.858,0.677) {$4000$};
\draw[gp path] (10.903,0.985)--(10.903,1.165);
\draw[gp path] (10.903,8.381)--(10.903,8.201);
\node[gp node center] at (10.903,0.677) {$4500$};
\draw[gp path] (11.947,0.985)--(11.947,1.165);
\draw[gp path] (11.947,8.381)--(11.947,8.201);
\node[gp node center] at (11.947,0.677) {$5000$};
\draw[gp path] (1.504,8.381)--(1.504,0.985)--(11.947,0.985)--(11.947,8.381)--cycle;
\node[gp node center,rotate=-270] at (0.246,4.683) {Execution Time [ms]};
\node[gp node center] at (6.725,0.215) {Number of Events [$\times 10000$]};
\node[gp node right] at (10.479,8.047) {brute-force};
\gpcolor{rgb color={0.580,0.000,0.827}}
\draw[gp path] (10.663,8.047)--(11.579,8.047);
\draw[gp path] (1.504,0.985)--(1.514,0.992)--(1.602,1.046)--(2.490,1.694)--(3.476,2.414)%
  --(4.461,3.107)--(5.447,3.805)--(6.433,4.377)--(7.419,5.218)--(8.405,5.831)--(9.391,6.539)%
  --(10.376,7.044)--(11.362,7.911);
\gpsetpointsize{4.00}
\gppoint{gp mark 1}{(1.504,0.985)}
\gppoint{gp mark 1}{(1.514,0.992)}
\gppoint{gp mark 1}{(1.602,1.046)}
\gppoint{gp mark 1}{(2.490,1.694)}
\gppoint{gp mark 1}{(3.476,2.414)}
\gppoint{gp mark 1}{(4.461,3.107)}
\gppoint{gp mark 1}{(5.447,3.805)}
\gppoint{gp mark 1}{(6.433,4.377)}
\gppoint{gp mark 1}{(7.419,5.218)}
\gppoint{gp mark 1}{(8.405,5.831)}
\gppoint{gp mark 1}{(9.391,6.539)}
\gppoint{gp mark 1}{(10.376,7.044)}
\gppoint{gp mark 1}{(11.362,7.911)}
\gppoint{gp mark 1}{(11.121,8.047)}
\gpcolor{color=gp lt color border}
\node[gp node right] at (10.479,7.739) {BM};
\gpcolor{rgb color={0.000,0.620,0.451}}
\draw[gp path] (10.663,7.739)--(11.579,7.739);
\draw[gp path] (1.504,0.985)--(1.514,0.992)--(1.602,1.047)--(2.490,1.699)--(3.476,2.418)%
  --(4.461,3.125)--(5.447,3.794)--(6.433,4.509)--(7.419,5.232)--(8.405,6.332)--(9.391,6.582)%
  --(10.376,7.211)--(11.362,7.890);
\gppoint{gp mark 2}{(1.504,0.985)}
\gppoint{gp mark 2}{(1.514,0.992)}
\gppoint{gp mark 2}{(1.602,1.047)}
\gppoint{gp mark 2}{(2.490,1.699)}
\gppoint{gp mark 2}{(3.476,2.418)}
\gppoint{gp mark 2}{(4.461,3.125)}
\gppoint{gp mark 2}{(5.447,3.794)}
\gppoint{gp mark 2}{(6.433,4.509)}
\gppoint{gp mark 2}{(7.419,5.232)}
\gppoint{gp mark 2}{(8.405,6.332)}
\gppoint{gp mark 2}{(9.391,6.582)}
\gppoint{gp mark 2}{(10.376,7.211)}
\gppoint{gp mark 2}{(11.362,7.890)}
\gppoint{gp mark 2}{(11.121,7.739)}
\gpcolor{color=gp lt color border}
\node[gp node right] at (10.479,7.431) {FJS};
\gpcolor{rgb color={1.000,0.000,0.000}}
\draw[gp path] (10.663,7.431)--(11.579,7.431);
\draw[gp path] (1.504,0.985)--(1.514,0.996)--(1.602,1.011)--(2.490,1.342)--(3.476,1.716)%
  --(4.461,2.083)--(5.447,2.390)--(6.433,2.685)--(7.419,3.119)--(8.405,3.398)--(9.391,3.700)%
  --(10.376,3.968)--(11.362,4.328);
\gppoint{gp mark 3}{(1.504,0.985)}
\gppoint{gp mark 3}{(1.514,0.996)}
\gppoint{gp mark 3}{(1.602,1.011)}
\gppoint{gp mark 3}{(2.490,1.342)}
\gppoint{gp mark 3}{(3.476,1.716)}
\gppoint{gp mark 3}{(4.461,2.083)}
\gppoint{gp mark 3}{(5.447,2.390)}
\gppoint{gp mark 3}{(6.433,2.685)}
\gppoint{gp mark 3}{(7.419,3.119)}
\gppoint{gp mark 3}{(8.405,3.398)}
\gppoint{gp mark 3}{(9.391,3.700)}
\gppoint{gp mark 3}{(10.376,3.968)}
\gppoint{gp mark 3}{(11.362,4.328)}
\gppoint{gp mark 3}{(11.121,7.431)}
\gpcolor{color=gp lt color border}
\draw[gp path] (1.504,8.381)--(1.504,0.985)--(11.947,0.985)--(11.947,8.381)--cycle;
\gpdefrectangularnode{gp plot 1}{\pgfpoint{1.504cm}{0.985cm}}{\pgfpoint{11.947cm}{8.381cm}}
\end{tikzpicture}
  }
  \caption{\textsc{Settling}: exec.\ time}
  \label{fig:case4_exec_time}
\end{minipage}
\begin{minipage}{0.48\textwidth}
  \centering
 \scalebox{0.4}{
  \begin{tikzpicture}[gnuplot]
\path (0.000,0.000) rectangle (12.500,8.750);
\gpcolor{color=gp lt color border}
\gpsetlinetype{gp lt border}
\gpsetdashtype{gp dt solid}
\gpsetlinewidth{1.00}
\draw[gp path] (1.136,0.985)--(1.316,0.985);
\draw[gp path] (11.947,0.985)--(11.767,0.985);
\node[gp node right] at (0.952,0.985) {$0$};
\draw[gp path] (1.136,2.042)--(1.316,2.042);
\draw[gp path] (11.947,2.042)--(11.767,2.042);
\node[gp node right] at (0.952,2.042) {$10$};
\draw[gp path] (1.136,3.098)--(1.316,3.098);
\draw[gp path] (11.947,3.098)--(11.767,3.098);
\node[gp node right] at (0.952,3.098) {$20$};
\draw[gp path] (1.136,4.155)--(1.316,4.155);
\draw[gp path] (11.947,4.155)--(11.767,4.155);
\node[gp node right] at (0.952,4.155) {$30$};
\draw[gp path] (1.136,5.211)--(1.316,5.211);
\draw[gp path] (11.947,5.211)--(11.767,5.211);
\node[gp node right] at (0.952,5.211) {$40$};
\draw[gp path] (1.136,6.268)--(1.316,6.268);
\draw[gp path] (11.947,6.268)--(11.767,6.268);
\node[gp node right] at (0.952,6.268) {$50$};
\draw[gp path] (1.136,7.324)--(1.316,7.324);
\draw[gp path] (11.947,7.324)--(11.767,7.324);
\node[gp node right] at (0.952,7.324) {$60$};
\draw[gp path] (1.136,8.381)--(1.316,8.381);
\draw[gp path] (11.947,8.381)--(11.767,8.381);
\node[gp node right] at (0.952,8.381) {$70$};
\draw[gp path] (1.136,0.985)--(1.136,1.165);
\draw[gp path] (1.136,8.381)--(1.136,8.201);
\node[gp node center] at (1.136,0.677) {$0$};
\draw[gp path] (2.938,0.985)--(2.938,1.165);
\draw[gp path] (2.938,8.381)--(2.938,8.201);
\node[gp node center] at (2.938,0.677) {$20$};
\draw[gp path] (4.740,0.985)--(4.740,1.165);
\draw[gp path] (4.740,8.381)--(4.740,8.201);
\node[gp node center] at (4.740,0.677) {$40$};
\draw[gp path] (6.542,0.985)--(6.542,1.165);
\draw[gp path] (6.542,8.381)--(6.542,8.201);
\node[gp node center] at (6.542,0.677) {$60$};
\draw[gp path] (8.343,0.985)--(8.343,1.165);
\draw[gp path] (8.343,8.381)--(8.343,8.201);
\node[gp node center] at (8.343,0.677) {$80$};
\draw[gp path] (10.145,0.985)--(10.145,1.165);
\draw[gp path] (10.145,8.381)--(10.145,8.201);
\node[gp node center] at (10.145,0.677) {$100$};
\draw[gp path] (11.947,0.985)--(11.947,1.165);
\draw[gp path] (11.947,8.381)--(11.947,8.201);
\node[gp node center] at (11.947,0.677) {$120$};
\draw[gp path] (1.136,8.381)--(1.136,0.985)--(11.947,0.985)--(11.947,8.381)--cycle;
\node[gp node center,rotate=-270] at (0.246,4.683) {Execution Time [ms]};
\node[gp node center] at (6.541,0.215) {Number of Events [$\times 10000$]};
\node[gp node right] at (10.479,8.047) {brute-force};
\gpcolor{rgb color={0.580,0.000,0.827}}
\draw[gp path] (10.663,8.047)--(11.579,8.047);
\draw[gp path] (1.139,0.988)--(2.285,1.856)--(3.440,2.758)--(4.588,3.869)--(5.719,4.479)%
  --(6.834,5.979)--(7.969,6.550)--(9.196,7.474)--(10.248,7.853);
\gpsetpointsize{4.00}
\gppoint{gp mark 1}{(1.139,0.988)}
\gppoint{gp mark 1}{(2.285,1.856)}
\gppoint{gp mark 1}{(3.440,2.758)}
\gppoint{gp mark 1}{(4.588,3.869)}
\gppoint{gp mark 1}{(5.719,4.479)}
\gppoint{gp mark 1}{(6.834,5.979)}
\gppoint{gp mark 1}{(7.969,6.550)}
\gppoint{gp mark 1}{(9.196,7.474)}
\gppoint{gp mark 1}{(10.248,7.853)}
\gppoint{gp mark 1}{(11.121,8.047)}
\gpcolor{color=gp lt color border}
\node[gp node right] at (10.479,7.739) {BM};
\gpcolor{rgb color={0.000,0.620,0.451}}
\draw[gp path] (10.663,7.739)--(11.579,7.739);
\draw[gp path] (1.139,0.988)--(2.285,1.783)--(3.440,2.507)--(4.588,3.573)--(5.719,4.113)%
  --(6.834,5.473)--(7.969,6.089)--(9.196,6.851)--(10.248,7.265);
\gppoint{gp mark 2}{(1.139,0.988)}
\gppoint{gp mark 2}{(2.285,1.783)}
\gppoint{gp mark 2}{(3.440,2.507)}
\gppoint{gp mark 2}{(4.588,3.573)}
\gppoint{gp mark 2}{(5.719,4.113)}
\gppoint{gp mark 2}{(6.834,5.473)}
\gppoint{gp mark 2}{(7.969,6.089)}
\gppoint{gp mark 2}{(9.196,6.851)}
\gppoint{gp mark 2}{(10.248,7.265)}
\gppoint{gp mark 2}{(11.121,7.739)}
\gpcolor{color=gp lt color border}
\node[gp node right] at (10.479,7.431) {FJS};
\gpcolor{rgb color={1.000,0.000,0.000}}
\draw[gp path] (10.663,7.431)--(11.579,7.431);
\draw[gp path] (1.139,0.988)--(2.285,1.714)--(3.440,2.564)--(4.588,3.427)--(5.719,3.997)%
  --(6.834,5.406)--(7.969,5.781)--(9.196,6.515)--(10.248,7.186);
\gppoint{gp mark 3}{(1.139,0.988)}
\gppoint{gp mark 3}{(2.285,1.714)}
\gppoint{gp mark 3}{(3.440,2.564)}
\gppoint{gp mark 3}{(4.588,3.427)}
\gppoint{gp mark 3}{(5.719,3.997)}
\gppoint{gp mark 3}{(6.834,5.406)}
\gppoint{gp mark 3}{(7.969,5.781)}
\gppoint{gp mark 3}{(9.196,6.515)}
\gppoint{gp mark 3}{(10.248,7.186)}
\gppoint{gp mark 3}{(11.121,7.431)}
\gpcolor{color=gp lt color border}
\draw[gp path] (1.136,8.381)--(1.136,0.985)--(11.947,0.985)--(11.947,8.381)--cycle;
\gpdefrectangularnode{gp plot 1}{\pgfpoint{1.136cm}{0.985cm}}{\pgfpoint{11.947cm}{8.381cm}}
\end{tikzpicture}
 }
  \caption{\textsc{Gear}: exec.\ time}
  \label{fig:case5_exec_time}
 \end{minipage}
 \hfill
 \begin{minipage}{0.48\textwidth}
  \centering
  \scalebox{0.4}{
  \begin{tikzpicture}[gnuplot]
\path (0.000,0.000) rectangle (12.500,8.750);
\gpcolor{color=gp lt color border}
\gpsetlinetype{gp lt border}
\gpsetdashtype{gp dt solid}
\gpsetlinewidth{1.00}
\draw[gp path] (1.136,0.985)--(1.316,0.985);
\draw[gp path] (11.947,0.985)--(11.767,0.985);
\node[gp node right] at (0.952,0.985) {$0$};
\draw[gp path] (1.136,2.042)--(1.316,2.042);
\draw[gp path] (11.947,2.042)--(11.767,2.042);
\node[gp node right] at (0.952,2.042) {$10$};
\draw[gp path] (1.136,3.098)--(1.316,3.098);
\draw[gp path] (11.947,3.098)--(11.767,3.098);
\node[gp node right] at (0.952,3.098) {$20$};
\draw[gp path] (1.136,4.155)--(1.316,4.155);
\draw[gp path] (11.947,4.155)--(11.767,4.155);
\node[gp node right] at (0.952,4.155) {$30$};
\draw[gp path] (1.136,5.211)--(1.316,5.211);
\draw[gp path] (11.947,5.211)--(11.767,5.211);
\node[gp node right] at (0.952,5.211) {$40$};
\draw[gp path] (1.136,6.268)--(1.316,6.268);
\draw[gp path] (11.947,6.268)--(11.767,6.268);
\node[gp node right] at (0.952,6.268) {$50$};
\draw[gp path] (1.136,7.324)--(1.316,7.324);
\draw[gp path] (11.947,7.324)--(11.767,7.324);
\node[gp node right] at (0.952,7.324) {$60$};
\draw[gp path] (1.136,8.381)--(1.316,8.381);
\draw[gp path] (11.947,8.381)--(11.767,8.381);
\node[gp node right] at (0.952,8.381) {$70$};
\draw[gp path] (1.136,0.985)--(1.136,1.165);
\draw[gp path] (1.136,8.381)--(1.136,8.201);
\node[gp node center] at (1.136,0.677) {$0$};
\draw[gp path] (2.337,0.985)--(2.337,1.165);
\draw[gp path] (2.337,8.381)--(2.337,8.201);
\node[gp node center] at (2.337,0.677) {$20$};
\draw[gp path] (3.538,0.985)--(3.538,1.165);
\draw[gp path] (3.538,8.381)--(3.538,8.201);
\node[gp node center] at (3.538,0.677) {$40$};
\draw[gp path] (4.740,0.985)--(4.740,1.165);
\draw[gp path] (4.740,8.381)--(4.740,8.201);
\node[gp node center] at (4.740,0.677) {$60$};
\draw[gp path] (5.941,0.985)--(5.941,1.165);
\draw[gp path] (5.941,8.381)--(5.941,8.201);
\node[gp node center] at (5.941,0.677) {$80$};
\draw[gp path] (7.142,0.985)--(7.142,1.165);
\draw[gp path] (7.142,8.381)--(7.142,8.201);
\node[gp node center] at (7.142,0.677) {$100$};
\draw[gp path] (8.343,0.985)--(8.343,1.165);
\draw[gp path] (8.343,8.381)--(8.343,8.201);
\node[gp node center] at (8.343,0.677) {$120$};
\draw[gp path] (9.545,0.985)--(9.545,1.165);
\draw[gp path] (9.545,8.381)--(9.545,8.201);
\node[gp node center] at (9.545,0.677) {$140$};
\draw[gp path] (10.746,0.985)--(10.746,1.165);
\draw[gp path] (10.746,8.381)--(10.746,8.201);
\node[gp node center] at (10.746,0.677) {$160$};
\draw[gp path] (11.947,0.985)--(11.947,1.165);
\draw[gp path] (11.947,8.381)--(11.947,8.201);
\node[gp node center] at (11.947,0.677) {$180$};
\draw[gp path] (1.136,8.381)--(1.136,0.985)--(11.947,0.985)--(11.947,8.381)--cycle;
\node[gp node center,rotate=-270] at (0.246,4.683) {Execution Time [ms]};
\node[gp node center] at (6.541,0.215) {Number of Events [$\times 10000$]};
\node[gp node right] at (10.479,8.047) {brute-force};
\gpcolor{rgb color={0.580,0.000,0.827}}
\draw[gp path] (10.663,8.047)--(11.579,8.047);
\draw[gp path] (1.140,0.989)--(2.422,1.894)--(3.709,2.815)--(4.997,3.862)--(6.268,4.866)%
  --(7.543,5.650)--(8.822,6.446)--(10.139,7.400)--(11.386,8.264);
\gpsetpointsize{4.00}
\gppoint{gp mark 1}{(1.140,0.989)}
\gppoint{gp mark 1}{(2.422,1.894)}
\gppoint{gp mark 1}{(3.709,2.815)}
\gppoint{gp mark 1}{(4.997,3.862)}
\gppoint{gp mark 1}{(6.268,4.866)}
\gppoint{gp mark 1}{(7.543,5.650)}
\gppoint{gp mark 1}{(8.822,6.446)}
\gppoint{gp mark 1}{(10.139,7.400)}
\gppoint{gp mark 1}{(11.386,8.264)}
\gppoint{gp mark 1}{(11.121,8.047)}
\gpcolor{color=gp lt color border}
\node[gp node right] at (10.479,7.739) {BM};
\gpcolor{rgb color={0.000,0.620,0.451}}
\draw[gp path] (10.663,7.739)--(11.579,7.739);
\draw[gp path] (1.140,0.989)--(2.422,1.828)--(3.709,2.710)--(4.997,3.406)--(6.268,4.177)%
  --(7.543,5.031)--(8.822,5.803)--(10.139,6.581)--(11.386,7.827);
\gppoint{gp mark 2}{(1.140,0.989)}
\gppoint{gp mark 2}{(2.422,1.828)}
\gppoint{gp mark 2}{(3.709,2.710)}
\gppoint{gp mark 2}{(4.997,3.406)}
\gppoint{gp mark 2}{(6.268,4.177)}
\gppoint{gp mark 2}{(7.543,5.031)}
\gppoint{gp mark 2}{(8.822,5.803)}
\gppoint{gp mark 2}{(10.139,6.581)}
\gppoint{gp mark 2}{(11.386,7.827)}
\gppoint{gp mark 2}{(11.121,7.739)}
\gpcolor{color=gp lt color border}
\node[gp node right] at (10.479,7.431) {FJS};
\gpcolor{rgb color={1.000,0.000,0.000}}
\draw[gp path] (10.663,7.431)--(11.579,7.431);
\draw[gp path] (1.140,0.988)--(2.422,1.352)--(3.709,1.788)--(4.997,2.337)--(6.268,2.470)%
  --(7.543,2.983)--(8.822,3.162)--(10.139,3.648)--(11.386,3.913);
\gppoint{gp mark 3}{(1.140,0.988)}
\gppoint{gp mark 3}{(2.422,1.352)}
\gppoint{gp mark 3}{(3.709,1.788)}
\gppoint{gp mark 3}{(4.997,2.337)}
\gppoint{gp mark 3}{(6.268,2.470)}
\gppoint{gp mark 3}{(7.543,2.983)}
\gppoint{gp mark 3}{(8.822,3.162)}
\gppoint{gp mark 3}{(10.139,3.648)}
\gppoint{gp mark 3}{(11.386,3.913)}
\gppoint{gp mark 3}{(11.121,7.431)}
\gpcolor{color=gp lt color border}
\draw[gp path] (1.136,8.381)--(1.136,0.985)--(11.947,0.985)--(11.947,8.381)--cycle;
\gpdefrectangularnode{gp plot 1}{\pgfpoint{1.136cm}{0.985cm}}{\pgfpoint{11.947cm}{8.381cm}}
\end{tikzpicture}
  }
  \caption{\textsc{Accel}: exec.\ time}
  \label{fig:case6_exec_time}
 \end{minipage}

\vspace{1em}
 
 \renewcommand{\figurename}{Table}
 \setlength\abovecaptionskip{0pt}
 \begin{minipage}[t]{0.24\textwidth}
 \tblcaption{\textsc{Simple} (sec.)}
 \label{table:case1_montre_exec_time}
  \centering
  \scalebox{0.58}{
   \pgfplotstabletypeset[
   sci,
   sci zerofill,
   multicolumn names, 
   display columns/0/.style={
   column name=$|w|$, 
   fixed,fixed zerofill,precision=0,
   },
   columns={[index]0,[index]2,[index]3,[index]4},
   display columns/1/.style={
   fixed,fixed zerofill,precision=2,
   column name=\begin{tabular}{c}FJS\\ (online)\end{tabular}
 },
   display columns/2/.style={
   fixed,fixed zerofill,precision=2,
   column name=\begin{tabular}{c}Montre\\(offline)\end{tabular}},
   display columns/3/.style={
   fixed,fixed zerofill,precision=2,
   column name=\begin{tabular}{c}Montre\\(online)\end{tabular}},
   fixed,fixed zerofill,precision=2,
   every head row/.style={
   before row={\toprule}, 
   after row={\midrule} 
   },
   empty cells with={Timeout},
   every last row/.style={after row=\bottomrule}, 
   ]{elapsed-time-0-0419193201.table}}
 \end{minipage}
 \hfill
 \begin{minipage}[t]{0.28\textwidth}
 \tblcaption{\textsc{Settling} (sec.)}
 \label{table:case4_montre_exec_time}
  \centering
  \scalebox{0.58}{
   \begin{filecontents*}{elapsed-time-14-0505151239.dat}
300 0 0 0.01 0.01
30000 0.01 0.01 0.01 0.01
300000 0.12 0.11 0.01 0.01
3000000 1.23 1.11 3.85 299.85
6000000 2.46 2.23 7.74 600.66
9000000 3.7 3.34 11.66 893.88
12000000 4.94 4.46 15.65 1188.02
15000000 6.18 5.58 19.75 1475.89
18000000 7.44 6.72 24.48 1788.18
21000000 8.7 9.27 27.8 1800
24000000 9.92 8.96 31.78 1800
27000000 11.18 10.09 37.1 1800
30000000 12.41 11.21 41.1 1800
\end{filecontents*}
\pgfplotstabletypeset[
   sci,
   sci zerofill,
   multicolumn names, 
   columns={[index]0,[index]2,[index]3,[index]4},
   display columns/0/.style={
   column name=$|w|$, 
   fixed,fixed zerofill,precision=0,
   },
   display columns/1/.style={
   fixed,fixed zerofill,precision=2,
   column name=\begin{tabular}{c}FJS\\ (online)\end{tabular}
 },
   display columns/2/.style={
   fixed,fixed zerofill,precision=2,
   column name=\begin{tabular}{c}Montre\\(offline)\end{tabular}},
   display columns/3/.style={
   fixed,fixed zerofill,precision=2,
   string replace={180}{},
   column name=\begin{tabular}{c}Montre\\(online)\end{tabular}},
   fixed,fixed zerofill,precision=2,
   every head row/.style={
   before row={\toprule}, 
   after row={\midrule} 
   },
   string replace={1800}{},
   empty cells with={Timeout},
   every last row/.style={after row=\bottomrule}, 
   ]{elapsed-time-14-0505151239.dat}}
 \end{minipage}
 \hfill
 \begin{minipage}[t]{0.23\textwidth}
 \tblcaption{\textsc{Gear} (sec.)}
 \label{table:case5_montre_exec_time}
  \centering
  \scalebox{0.58}{
   \begin{filecontents*}{elapsed-time-12-0419193201.dat}
1000 0 0 0.01 0.04
86400 0.04 0.04 0.15 11.63
172800 0.09 0.08 0.29 23.48
259200 0.13 0.13 0.42 37.51
345600 0.17 0.17 0.54 47.2
432000 0.22 0.21 0.67 57.99
518400 0.26 0.25 0.85 69.76
604800 0.31 0.3 0.96 87.59
691200 0.35 0.34 1.09 90.36
\end{filecontents*}
\pgfplotstabletypeset[
   sci,
   sci zerofill,
   multicolumn names, 
   columns={[index]0,[index]2,[index]3,[index]4},
   display columns/0/.style={
   column name=$|w|$, 
   fixed,fixed zerofill,precision=0,
   },
   display columns/1/.style={
   fixed,fixed zerofill,precision=2,
   column name=\begin{tabular}{c}FJS\\ (online)\end{tabular}
 },
   display columns/2/.style={
   fixed,fixed zerofill,precision=2,
   column name=\begin{tabular}{c}Montre\\(offline)\end{tabular}},
   display columns/3/.style={
   fixed,fixed zerofill,precision=2,
   column name=\begin{tabular}{c}Montre\\(online)\end{tabular}},
   fixed,fixed zerofill,precision=2,
   every head row/.style={
   before row={\toprule}, 
   after row={\midrule} 
   },
   empty cells with={Timeout},
   every last row/.style={after row=\bottomrule}, 
   ]{elapsed-time-12-0419193201.dat}} 
 \tblcaption{\textsc{Accel} (sec.)}
  \label{table:case6_montre_exec_time}
  \centering
  \scalebox{0.58}{
   \begin{filecontents*}{elapsed-time-10-0419193201.dat}
1000 0 0 0.01 69.05
86400 0.07 0.06 0.63 1800
172800 0.15 0.13 1.25 1800
259200 0.23 0.2 1.88 1800
345600 0.3 0.26 2.5 1800
432000 0.38 0.33 3.12 1800
518400 0.46 0.4 3.75 1800
604800 0.53 0.46 4.38 1800
691200 0.61 0.53 4.99 1800
\end{filecontents*}
\pgfplotstabletypeset[
   sci,
   sci zerofill,
   multicolumn names, 
   columns={[index]0,[index]2,[index]3,[index]4},
   display columns/0/.style={
   column name=$|w|$, 
   fixed,fixed zerofill,precision=0,
   },
   display columns/1/.style={
   fixed,fixed zerofill,precision=2,
   column name=\begin{tabular}{c}FJS\\ (online)\end{tabular}
 },
   display columns/2/.style={
   fixed,fixed zerofill,precision=2,
   column name=\begin{tabular}{c}Montre\\(offline)\end{tabular}},
   display columns/3/.style={
   fixed,fixed zerofill,precision=2,
   column name=\begin{tabular}{c}Montre\\(online)\end{tabular}},
   fixed,fixed zerofill,precision=2,
   every head row/.style={
   before row={\toprule}, 
   after row={\midrule} 
   },
   string replace={1800}{},
   empty cells with={Timeout},
   every last row/.style={after row=\bottomrule}, 
   ]{elapsed-time-10-0419193201.dat}}
 \end{minipage}
 \hfill
  \begin{minipage}[t]{0.23\textwidth}
   \tblcaption{Memory consumption of FJS (online) and BM}
   \scalebox{0.7}{
   \begin{filecontents*}{max-ram-14-0506001548.table.dat}
300 4390912 1.16016 1.16406
30000 4390912 2.60938 1.16406
300000 4423680 15.5469 1.16406
3000000 4423680 145.211 1.16406
6000000 4423680 289.246 1.16406
9000000 4390912 433.309 1.16406
12000000 4440064 577.32 1.19141
15000000 4390912 721.367 1.17969
18000000 4423680 865.418 1.19141
21000000 4423680 1009.46 1.16406
24000000 4423680 1153.5 1.16406
27000000 4423680 1297.57 1.16406
30000000 4390912 1441.61 1.16406
\end{filecontents*}
\pgfplotstabletypeset[
   sci,
   sci zerofill,
   multicolumn names, 
   columns={[index]0,[index]2,[index]3},
   display columns/0/.style={
   column name=$|w|$, 
   fixed,fixed zerofill,precision=0,
   },
   display columns/1/.style={
   fixed,fixed zerofill,precision=2,
   column name=BM (MB),
   },
   display columns/2/.style={
   fixed,fixed zerofill,precision=2,
   column name=FJS (MB),
   },
   fixed,fixed zerofill,precision=2,
   every head row/.style={
   before row={\toprule}, 
   after row={\midrule} 
   },
   every last row/.style={after row=\bottomrule}, 
   ]{max-ram-14-0506001548.table.dat}}
   \label{fig:ram_usage}
  \end{minipage}

\end{figure}

We implemented the brute-force, BM, FJS algorithms in
 C++ \cite{TimedFJSSample} and
we compiled them by clang-800.0.42.1.  
All the experiments are done on MacBook Pro
Early 2013 with 2.6 GHz Intel Core i5 processor and 8 GB 1600MHz DDR3
RAM.

\vspace{.5em}
\noindent
\begin{minipage}{\textwidth}
\paragraph{Speed (i.e.\ Permissible Density in Online Usage)}
In~Fig.~\ref{fig:case1_exec_time}--\ref{fig:case6_exec_time} 
are the comparison of the offline implementations of the brute-force, BM
and FJS algorithms, respectively (average of  five runs). Preprocessing time is excluded (it is
anyway negligible, see Appendix~\ref{appendix:preprocessing}). We also
exclude time of loading the input timed word in memory; this is because in many deployment scenarios like embedded ones,   I/O is pipelined by, for
  example, DMA.
\end{minipage}

 The pattern automata
for the benchmarks \textsc{Torque}, \textsc{Setting}, and \textsc{Gear}
 look
similar to each other. However their input timed words---generated by
a suitable Simulink model for each benchmark---exhibit different
characteristics, such as how often the characters in the pattern
automaton occur in the input timed words. Accordingly the performance of
the timed pattern matching algorithms varies, as we see in
Fig.~\ref{fig:case3_exec_time}--\ref{fig:case5_exec_time}.

  We observe that our FJS algorithm generally outperforms the BM and
  brute-force ones. 
  For \textsc{Settling} and  \textsc{Accel} the performance gap is roughly twice, and it possibly makes a big practical difference e.g.\ when a data set is huge and the monitoring task takes hours. For \textsc{Large
  Constraints} it seems to depend on specific words which algorithm
  performs better. The advantage in performance is as we expected, given
  that the FJS algorithm combines 
  the KMP-type skipping (that works well roughly when the BM-type one does)
  and the Quick Search-type skipping (that complements KMP). 
  After all, it is encouraging to observe that our FJS algorithm performs better in the automotive examples, where our motivation is drawn. 


In every benchmark except for \textsc{Large Constraint}, the execution
time grows roughly linearly on the length of the input word. This is a pleasant property for monitoring algorithms for which an input word can be very long. 


These results for \emph{offline} implementations also support our claim of FJS's superiority in \emph{online} usage scenarios. In online usage we must process an input word faster than the speed with which the word arrives; otherwise the word eventually floods memory. Thus running twice as fast means that our algorithm can handle  twice as dense input---or that we can use cheaper hardware to conduct the same monitoring task. Note that the difference between our 
online and offline implementations is only in the memory management and I/O. Thus their speed should be similar.

\noindent\textit{Memory Usage}\quad
In Table~\ref{fig:ram_usage} is the memory consumption of our \emph{online} FJS implementation and that of BM, for the \textsc{Settling} benchmark (the tendency is the same for the other benchmarks). The absolute values are not very important because they include our program and dynamically linked libraries; what matters is the tendency that memory consumption is almost constant for online FJS while it increases for BM. Constant memory consumption is an important property for monitoring algorithms, especially in online usage. The results here also concurs with our theoretical observation at the end of~\S{}\ref{sec:timedFJS} (see Fig.~\ref{fig:FJSIsOnlineBMIsNot}). 

\subsection{Comparison with Montre}
Here we compare with Montre, a recent tool for (both online and offline) timed pattern matching~\cite{DBLP:journals/corr/Ulus16}. Montre's online and offline algorithms differ from each other; both of them are quite different from our FJS algorithm, too. Montre's emphasis is on the algebraic structure of timed regular expressions and compositional reasoning thereby, while our algorithm features automata-theoretic views on the problem. 

Since we had difficulty running Montre in the same environment as in~\S{}\ref{subsec:comparisonWithBFAndBM}, we instead used 
GCC 4.9.3
as a compiler, and conducted experiments 
on an Amazon EC2 c4.xlarge instance (April 2017, 4 vCPUs and 7.5 GB RAM) that runs Ubuntu 14.04
LTS (64 bit).  The timeout is set to thirty minutes.

In 
Tables~\ref{table:case1_montre_exec_time}--\ref{table:case6_montre_exec_time} are the results. Here we use the benchmarks
\textsc{Simple},
\textsc{Settling}, \textsc{Gear}, and \textsc{Accel},
for which the translation between timed words (our input) and signals (Montre's input) makes sense. 
Our
(online) FJS implementation is about 3 to 8 times faster than offline
Montre and about 250 times faster than online Montre.  
The big performance advantage over \emph{online} Montre can be attributed to various reasons, including:
1)
online Montre needs to frequently compute  derivatives of TREs; 2) online Montre is comparable to our  brute-force algorithm in that there is no skipping involved; and 3) Montre is implemented in a
functional language (Pure~\cite{PureLang}) that is in general slower.
The reason for the advantage over \emph{offline} Montre is yet to be seen: given that the algorithms are very different, the advantage may well be solely attributed to implementation details. We claim however that good online performance of our FJS algorithm is a big advantage for monitoring applications.

\auxproof{
\subsection{Benchmarks}

We used the following six microbenchmarks in the experiments:
\textsc{Simple}, \textsc{Large Constraints}, \textsc{Torque},
\textsc{Settling}, \textsc{Gear}, and \textsc{Accel}.  Each
benchmark consists of a single pattern timed automaton $\mathcal{A}$
and a set of timed words $W$.  Each implementation solves the timed
pattern matching problem for each $w \in W$.  We explain each
benchmark in the following.

\paragraph{\textsc{Simple}}

Fig.~\ref{fig:case1_pattern}. shows the pattern timed automaton
$\mathcal{A}$ of benchmark \textsc{Simple} taken from the Case 1
in~\cite{DBLP:conf/formats/WagaAH16}.  The set $W$ consists of 37
timed words, each of which has 20 to 1,024,000 events.  The events of
each timed word is the alternation of \texttt{a} and \texttt{b}.  The
duration of each event is set randomly.


\paragraph{ \textsc{Large Constraints}}

The pattern $\mathcal{A}$ is a translation of timed regular expression
$\bigl\langle\bigl(\,\bigl(\langle\mathrm{p}\cdot\mathrm{\neg p}\rangle_{(0,10]}\bigr)^*\land\bigl(\langle\mathrm{q}\cdot\mathrm{\neg
    q}\rangle_{(0,10]}\bigr)^*\,\bigr)\cdot \$\bigr\rangle_{(0,80]}$.  The
   pattern timed automaton $\mathcal{A}$ consists of five states and
   nine transitions.  The length of each timed word in $W$ is 1,934 to
   31,935.  The events of a timed word in $W$ is an interleaving of
   $p,\neg p,p,\neg p,\dots$ and $q,\neg q,q,\neg q,\dots$.  The
   duration of each switching follows certain exponential
   distribution.


\paragraph{\textsc{Torque}}

The benchmark \textsc{Torque}, taken
from~\cite{DBLP:conf/formats/WagaAH16}, is the scenario of monitoring
the output torque of an automotive powertrain.  The set $W$ consists
of ten timed words.  The length of each word is 242,808 to 4,873,207.
In order to construct $W$, we conducted simulation of the model
\texttt{sldemo\_enginewc.slx} in the Simulink Demo
palette~\cite{SimulinkGuide}\footnote{This model takes an input
  representing the desired rpm.  We fed a random signal following the
  Gaussian distribution $N(2000,10^{6})$.}; the discretized output
representing the engine torque is translated to a sequence of
$\textrm{high}$ and $\textrm{low}$ with the threshold $40 \mathrm{N
  \cdot m}$.

Fig.~\ref{fig:case3_pattern} shows the timed pattern automaton of this
benchmark.  This pattern detects a five or more consecutive occurrence
of $\textrm{high}$ in one second.




\paragraph{\textsc{Settling}}

We used a Simulink model of an automotive
powertrain~\cite{DBLP:conf/hybrid/JinDKUB14} to construct $W$ of
benchmark \textsc{Settling}.  This model outputs two signals: the mode
of the powertrain (``normal'' and ``power'') and an A/F value error
$\mu$.  The simulation result is translated to a sequence of
$\set{\textrm{normal},\textrm{power},|\mu| \ge 0.02, |\mu| < 0.02}$ as
follows: (1) The letter $a \in \set{\textrm{normal},\textrm{power}}$
is recorded at time $\tau$ if a mode change to $a$ happens at this
time and (2) the letter $a \in \set{|\mu| \ge 0.02, |\mu| < 0.02}$ is
recorded at time $\tau$ if the value of $|\mu|$ crosses 0.02 at time
$\tau$.  For example, the event $(\textrm{normal}, 0.5)$ is in a word
if a mode change from $\textrm{power}$ to $\textrm{normal}$ happens at
0.5 sec.  The event $(|\mu| \ge 0.02, 0.6)$ is in a word if the value
of $|\mu|$ crosses $0.02$ from below.  The set $W$ consists of twelve
timed words.  The length of each word is 472 to 47,200,000 events.

The pattern timed automaton $\mathcal{A}$, shown in
Figure~\ref{fig:case4_pattern}, detects an unstable A/F value error:
This pattern matches a word $w$ if $|\mu|$ is not less than $0.02$ for
$100$ milliseconds in $w$.  This pattern models the violation of the
requirement (32) in~\cite{DBLP:conf/hybrid/JinDKUB14}.


\paragraph{\textsc{Gear}}

Benchmark \textsc{Gear} is inspired by the scenario of monitoring gear
change of an automatic transmission system.  We constructed the set
$W$ as follows.  We conducted simulation of the model of an automatic
transmission system~\cite{DBLP:conf/cpsweek/HoxhaAF14}.  We used
S-TaLiRo~\cite{DBLP:conf/tacas/AnnpureddyLFS11} to generate an input
to this model; it generates a gear change signal that is fed to the
model.  A gear is chosen from $\set{\text{g}_1,\text{g}_2,\text{g}_3,\text{g}_4}$.  The generated
gear change is recorded in a timed word.  The set $W$ consists of 9
timed words; the length of each word is 307 to 1,011,427.

The pattern timed automaton $\mathcal{A}$, shown in
Fig.~\ref{fig:case5_pattern}, detects the violation of the following
condition: If the gear is changed to 1, it should not be changed to 2
within 2 seconds.  This condition is related to the requirement
$\phi^{\mathit{AT}}_5$ proposed in~\cite{DBLP:conf/cpsweek/HoxhaAF14}.

\paragraph{\textsc{Accel}}
\label{subsec:case6}

The $W$ of benchmark $\textsc{Accel}$ is also constructed from the
Simulink model of the automated transmission
system~\cite{DBLP:conf/cpsweek/HoxhaAF14}.  For this benchmark, the
(discretized) value of three state variables are recorded in $W$:
engine RPM (discretized to ``high'' and ``low'' with a certain
threshold), velocity (discretized to ``high'' and ``low'' with a
certain threshold), and 4 gear positions.  We used
S-TaLiRo~\cite{DBLP:conf/tacas/AnnpureddyLFS11} to generate a input
sequence of gear change.  The set $W$ consists of 9 timed words.  The
length of each word is 25,002 to 17,280,002.

The pattern timed automaton $\mathcal{A}$ of this benchmark is shown
in Fig.~\ref{fig:case6_pattern}.  This pattern matches a part of a
timed word that violates the following condition: If a gear changes
from 1 to 2, 3, and 4 in this order in 10 seconds and engine RPM
becomes large during this gear change, then the velocity of the car
must be sufficiently large in one second.  This condition models the
requirement $\phi^{\mathit{AT}}_8$ proposed
in~\cite{DBLP:conf/cpsweek/HoxhaAF14}.

}

\auxproof{
\subsection{Memory usage of the online implementations}
\begin{figure}[t]
  \scalebox{0.50}{
  \input{max-ram-0-0406183616.tikz}
 }
 \scalebox{0.50}{
 \input{max-ram-14-0506001548.tikz}
 }
 \caption{RAM usage of the online implementation of our FJS-type
   algorithm for timed pattern matching.  We separated the results
   into two figures because the order of the length of timed words
   used in \textsc{Settling} is larger than the other.}
 \label{fig:ram_usage}
\end{figure}

Figure~\ref{fig:ram_usage} shows the maximum memory consumption of our
timed FJS algorithm for each benchmark.  The numbers are measured
using GNU time.  We remark that the results in
Table~\ref{fig:ram_usage} contain the size of our programs and
dynamically linked libraries; these numbers highly depend on the
environment.

Theoretically, the memory consumption of our timed FJS algorithm is
not constant with respect to the length of the input timed word.
However, Figure~\ref{fig:ram_usage} suggests that the memory
consumption can be constant to the length of the timed word in
practice.  This implies that our algorithm is suitable for the use
case in which we need to monitor the input for long time without
running into an out-of-memory error.
}

\auxproof{
\subsection{Comparison with Montre}

Montre~\cite{DBLP:journals/corr/Ulus16} is a tool for online/offline
timed pattern matching.  Both offline and online modes of Montre take a
TRE as input but its online mode lazily constructs a labeled transition
system (LTS) on memory whereas its offline mode inductively assign the
truth value of each subformula as a signal.  Notice that our current
hardcodes a timed automaton that corresponds to the TRE.

We compared the execution time of 
$\IMPL_{\ONLINE,\FJS}$ and online/offline Montre.  This experiment is
conducted on an Amazon EC2 c4.xlarge instance.  The OS is Ubuntu 14.04
LTS (64 bit).  The following benchmarks are used: \textsc{Simple},
\textsc{Settling}, \textsc{Gear}, and \textsc{Accel}; the pattern of
the other benchmarks are not expressible in Montre's state-based
semantics.  The timeout is set to thirty minutes.


\begin{table}[t]
\centering
\caption{Online algorithms of timed pattern matching}\label{table:online_algorithms}
\begin{tabular}{c|c|c}
&Brute-force& Skipping \\\hline
  \begin{tabular}{c}
   LTS is constructed\\
   on the fly\\
  \end{tabular}
 &Montre~\cite{DBLP:journals/corr/Ulus16}&
         ---\\\hline
  LTS is given& Brute-force
     in~\cite{DBLP:conf/formats/WagaAH16}&FJS-type in this paper\\
\end{tabular}
\end{table}

Tables~\ref{table:case1_montre_exec_time},
\ref{table:case4_montre_exec_time},
\ref{table:case5_montre_exec_time}, and
\ref{table:case6_montre_exec_time} show the result.  
 Our FJS algorithm (its online implementation) is 
 about 3 to 8 times faster than the offline-mode
Montre and about 250 times faster than the online-mode Montre.  Our
implementation outperforms the online mode of Montre because (1) the
online-mode Montre frequently computes a derivative of a TRE to
augment an LTS; (2) Montre's brute-force algorithm requires more
matching trials than our algorithm; (3) Montre is implemented in a
functional language Pure\cite{PureLang}, which is in general slower
and may use procedures that scale poorly;
and (4) our implementations hardcode patterns, whereas Montre takes
them as input and parses them.
}

\section{Conclusions and Future Work}
We continued~\cite{DBLP:conf/formats/WagaAH16} and presented an algorithm for timed pattern matching. Based on the FJS algorithm~\cite{DBLP:journals/jda/FranekJS07} it exhibits better online properties, as witnessed in our experiments. 
As future work we wish to implement an interface of our experimental implementation and distribute as a tool.
We also wish to try the algorithm in actual embedded hardware, like~\cite{kane2015runtime}.

\vspace{.3em}
\noindent\textit{Acknowledgments}\quad
Thanks are due to
Sean Sedwards
 for  useful discussions and comments.
The authors are supported by
JSPS Grant-in-Aid
15KT0012. 
M.W.\ and I.H.\ are supported by JST ERATO HASUO Metamathematics for Systems Design Project (No.\
JPMJER1603),  and JSPS Grant-in-Aid No.\ 15K11984. 
K.S.\ is supported by JST PRESTO (No.\ JPMJPR15E5) and JSPS Grant-in-Aid No.\ 70633692.

\newpage
\appendix

\section{Detailed Pseudocode of Our FJS-type Algorithm for Timed Pattern Matching}
\label{appendix:detailFJS}

\begin{algorithm}[tbp]
  \caption{Detail of our FJS-type algorithm for timed pattern matching}
  \label{alg:timedFJSDetail}
  \begin{algorithmic}[1]
   \Require A timed word $w = (\overline{a},\overline{\tau})$, and a timed
   automaton $\mathcal{A} = (\Sigma,S,S_0,C,E,F)$.
   \Ensure $\bigcup Z$ is the match set $\mathcal{M} (w,\mathcal{A})$ in
   Def.~\ref{def:TimedPatternMatching}.
   \State $n \gets 1$
   \Comment{$n$ is the position in $\str$ of the head of the current matching trial}
   \State $\CurrConf \gets \emptyset;\; Z\gets\emptyset$
   \While{$n \leq |w| - m + 2$}
   \While{$\forall\, \overline{r} \in L'.\, \overline{a}_{n+m-2}\neq a' \,\text{where}\, \overline{r}_{m-2} \xrightarrow{a'} \overline{r}_{m-1}$}
   \Comment{Try to match the tail of $L'$}
   \State $n \gets n + \Delta({\overline{a}_{n+m-1}})$
   \Comment{ Quick Search-type skipping}
   \If{$n > |w| - m + 2$}
   \Return
   \EndIf
   \EndWhile

   \State $\CurrConf \gets \{(s,\rhoEmpty,[\tau_{n-1},\tau_n)) \mid s \in S_0\}$
   \For{$n' \in \{n,n+1,\cdots,|w|\}$}   \Comment{ We try matching in the same way as~\cite{DBLP:conf/formats/WagaAH16}}
   \State{$\NextConf \gets \emptyset$}

   \For{$(s,\rho,T) \in \CurrConf$}
   \For{$(s,s',a_n,\lambda,\delta) \in E$} 
   \State $T' \gets \{t_0 \in T \mid \eval(\rho,\tau_n,t_0) \models \delta\}$
   \If{$T' \neq \emptyset$}
   \State $\rho' \gets \rho$

   \For{$x \in \lambda$}
   \State $\rho' \gets \reset(\rho',x,\tau_n)$
   \EndFor

   \State $\NextConf \gets \NextConf \cup (s',\rho',T')$

   \For{$s_f \in F, (s',s_f,\$,\lambda',\delta') \in E$}
   \State $T'' \gets (\tau_{n'},\tau_{n'+1}]$
   \State $Z \gets Z \cup \solConstr (T',T'',\rho',\delta')$
   \EndFor
   \EndIf

   \EndFor
   \EndFor
   \If{$\NextConf = \emptyset$}
   \Break
   \EndIf
   \State $\CurrConf \gets \NextConf$
   \EndFor
   
   \For {$k \in \{n+1,\cdots,n + \max \{ \beta(s) \mid (s,\rho,T) \in \CurrConf\} - 1\}$}
   \State \Comment{Matching trial stacks at the states $\{ s \mid (s,\rho,T) \in \CurrConf\}$}
   \For {$s \in S_0, s_f \in F, (s, s_f, \$, \rho, \delta) \in E$}
   \State $Z \gets Z \cup \solConstr([\tau_{k-1},\tau_k),(\tau_{k-1},\tau_k],\rho,\delta)$
   \EndFor
   \EndFor
   \State $n \gets n + \max \{ \beta(s) \mid (s,\rho,T) \in \CurrConf\}$
   \Comment{ KMP-type skipping}
   \EndWhile
  \end{algorithmic}
 \end{algorithm}

\begin{mydefinition}[$\eval, \reset,  \solConstr$]
 Let a pattern timed automaton be $\mathcal{A} = (\Sigma,S,S_0,C,E,F)$.
 For a partial function $\rho: C \rightharpoonup \Rp$ and $t,t_0 \in \Rp$,
 the clock interpretation 
 $\eval (\rho,t,t_0)\colon C\to \mathbb{R}_{\ge 0}$
 is $\eval (\rho,t,t_0)(x)=t-\rho(x)$ (if $\rho(x)$ is defined) and 
 $\eval (\rho,t,t_0)(x)=t-t_{0}$ (otherwise).
 For a partial function $\rho: C \rightharpoonup \Rp$, $t_{r}\in\Rp$ and $x \in C$, $\reset(\rho,x,t_r)\colon C\rightharpoonup\mathbb{R}_{>0}$
 is the following partial function
 : $\reset(\rho,x,t_r)(x)=t_{r}$; and $\reset(\rho,x,t_r)(y)=\rho(y)$
 for each $y\in C \setminus\{x\}$.
 (The latter is Kleene's equality between partial functions, to be precise.)
 For intervals $T,T'\subseteq\Rp$, a partial function 
 $\rho\colon C\rightharpoonup\Rnn$, and a clock constraint
 $\delta \in \Phi (C)$ (\S{}\ref{subsec:timedPatternMatching}), we define
 $\solConstr(T,T',\rho,\delta)=\bigl\{\,(t,t')\,\bigl|\bigr.\,t\in T, t'\in T',
 \eval(\rho,t',t)\models \delta\,\bigr\}$.
\end{mydefinition}

The detail of our FJS-type algorithm for timed pattern matching is in Algorithm~\ref{alg:timedFJSDetail}.

\section{Our  FJS-Type Timed Pattern Matching Problem, Illustrated}
\label{appendix:TimedFJSIllustrated}
Let up look at the example in Fig.~\ref{fig:input_timed_pattern_matching}. The zone automaton $\SG(\mathcal{A})$ is in Fig.~\ref{fig:zoneAutomaton};  the execution of our algorithm is illustrated in Fig.~\ref{fig:tpm_configs}.

The first configuration in Fig.~\ref{fig:tpm_configs} means we are after possible matches that start at $t \in [0,0.5)$. With $m=4$ (the length of the shortest accepted word),
we try matching of the third target character $\overline{a}_{m-1}= \overline{a}_{3} = \text{b}$ with the tail of every length-3 prefix of $L(\mathcal{A})$ using the zone automaton $\SG(\mathcal{A})$ in Fig.~\ref{fig:zoneAutomaton}.
The trial fails and we invoke Quick Search-type skipping $\Delta(\overline{a}_{4} = \text{b})$.
Since $\overline{a}_{4} = \text{b}$ does not appear in any transition of $\SG(\mathcal{A})$, we can skip four events and reach the second configuration where we look for potential matches that start at $t \in [1.7,2.8)$.

We again try matching form the tail. This time it succeeds because $\overline{a}_7 = \text{a}$ appears in the third character of a word accepted by $\SG(\mathcal{A})$.
Then we move to Line~\ref{line:leftToRightMatchingInTimedFJS} of Algorithm~\ref{alg:timedFJSDetail} where we 
try matching from left to right. After the trial stacks at $s_2 \in S$,
we invoke the KMP-type skipping. 

The KMP-type skip value $\beta(s_2)$ is computed as shown in Fig.~\ref{fig:betaforFJSTimedPatternMatching}. Here it is much more intricate how to decide {\color{dgreen}{\cmark}} or {\color{red}{\xmark}}, i.e.\ if the ``prefix'' on the top ($L'_{s_{2}}$) matches the shifts of $L'$ below. Previously for string or (untimed) pattern matching we just compared characters (Fig.~\ref{fig:tableForBeta} \&~\ref{fig:tableForPatternBeta}); here the question is if there exists a timed word that causes both a transition in the prefix (on the top) and the corresponding transition in a shift (below). For this purpose we employ the zone automaton 
$\SG(\mathcal{A} \times \mathcal{A})$ of the product timed automaton $\mathcal{A} \times \mathcal{A}$. For example, the shift by one position does not match ({\color{red}{\xmark}}) in Fig.~\ref{fig:betaforFJSTimedPatternMatching} because there is no transition 
\scalebox{0.55}{
  \tikz[baseline=(s_0.base),auto,node distance=0.5cm]{
  \regionstate{s_0}{$(s_0,s_1)$}{$x=x'=0$};
  \regionstate[right=of s_0]{s_1}{$(s_1,s_2)$}{$x=x'=0$};
  \path[->](s_0) edge node {$a$} (s_1);
   }} 
in $\SG(\mathcal{A} \times \mathcal{A})$. 


In the fourth configuration, we try matching from $t \in [3.7,4.9)$.
We again try matching from the tail; it succeeds; we try matching from left to right; and we find a matching $\{(t,t') \mid t \in [3.7,3.9), t' \in (6.0,\infty)\}$.

\begin{figure}[tbp]
 \centering
 \scalebox{0.7}{
\begin{tabular}{c c c}
 &
 \begin{tabular}{|c|}
  \hline
  \begin{tikzpicture} 
  \draw [thick, -stealth](-0.5,0)--(6.5,0) node [anchor=north]{$t$};
  \draw (0,0.1) -- (0,-0.1) node [anchor=north]{$0$};

 \draw (0.5,0.1) node[anchor=south]{$\text{a}$} -- (0.5,-0.1) node[anchor=north]{$0.5$};
 \draw (0.9,0.1) node[anchor=south]{$\text{a}$} -- (0.9,-0.1) node[anchor=north]{$0.9$};
 \draw (1.3,0.1) node[anchor=south]{$\text{b}$} -- (1.3,-0.1) node[anchor=north]{$1.3$};
 \draw (1.7,0.1) node[anchor=south]{$\text{b}$} -- (1.7,-0.1) node[anchor=north]{$1.7$};
 \draw (2.8,0.1) node[anchor=south]{$\text{a}$} -- (2.8,-0.1) node[anchor=north]{$2.8$};
 \draw (3.7,0.1) node[anchor=south]{$\text{a}$} -- (3.7,-0.1) node[anchor=north]{$3.7$};
 \draw (5.3,0.1) node[anchor=south]{$\text{a}$} -- (5.3,-0.1) node[anchor=north]{$5.3$};
 \draw (4.9,0.1) node[anchor=south]{$\text{a}$} -- (4.9,-0.1) node[anchor=north]{$4.9$};
 \draw (6.0,0.1) node[anchor=south]{$\text{a}$} -- (6.0,-0.1) node[anchor=north]{$6.0$};

 \draw (0.1,-0.5) rectangle (1.6,-1.0);
 \node at (0.85,-0.75) {$\mathcal{W}(L')$};
 \end{tikzpicture}
  \\\hline
 \end{tabular}
 $\stackrel{\Delta(\textrm{b})=4}{\Longrightarrow}$&
 \begin{tabular}{|c|}
  \hline
 \begin{tikzpicture} 
  \draw [thick, -stealth](-0.5,0)--(6.5,0) node [anchor=north]{$t$};
  \draw (0,0.1) -- (0,-0.1) node [anchor=north]{$0$};

 \draw (0.5,0.1) node[anchor=south]{$\text{a}$} -- (0.5,-0.1) node[anchor=north]{$0.5$};
 \draw (0.9,0.1) node[anchor=south]{$\text{a}$} -- (0.9,-0.1) node[anchor=north]{$0.9$};
 \draw (1.3,0.1) node[anchor=south]{$\text{b}$} -- (1.3,-0.1) node[anchor=north]{$1.3$};
 \draw (1.7,0.1) node[anchor=south]{$\text{b}$} -- (1.7,-0.1) node[anchor=north]{$1.7$};
 \draw (2.8,0.1) node[anchor=south]{$\text{a}$} -- (2.8,-0.1) node[anchor=north]{$2.8$};
 \draw (3.7,0.1) node[anchor=south]{$\text{a}$} -- (3.7,-0.1) node[anchor=north]{$3.7$};
 \draw (5.3,0.1) node[anchor=south]{$\text{a}$} -- (5.3,-0.1) node[anchor=north]{$5.3$};
 \draw (4.9,0.1) node[anchor=south]{$\text{a}$} -- (4.9,-0.1) node[anchor=north]{$4.9$};
 \draw (6.0,0.1) node[anchor=south]{$\text{a}$} -- (6.0,-0.1) node[anchor=north]{$6.0$};

 \draw (1.8,-0.5) rectangle (5.2,-1.0);
 \node at (3.5,-0.75) {$\mathcal{W}(L')$};
 \end{tikzpicture}\\\hline
  \end{tabular}
 \\
 $\Longrightarrow$
&
\begin{tabular}{|c|}
 \hline
  \begin{tikzpicture} 
  \draw [thick, -stealth](-0.5,0)--(6.5,0) node [anchor=north]{$t$};
  \draw (0,0.1) -- (0,-0.1) node [anchor=north]{$0$};

 \draw (0.5,0.1) node[anchor=south]{$\text{a}$} -- (0.5,-0.1) node[anchor=north]{$0.5$};
 \draw (0.9,0.1) node[anchor=south]{$\text{a}$} -- (0.9,-0.1) node[anchor=north]{$0.9$};
 \draw (1.3,0.1) node[anchor=south]{$\text{b}$} -- (1.3,-0.1) node[anchor=north]{$1.3$};
 \draw (1.7,0.1) node[anchor=south]{$\text{b}$} -- (1.7,-0.1) node[anchor=north]{$1.7$};
 \draw (2.8,0.1) node[anchor=south]{$\text{a}$} -- (2.8,-0.1) node[anchor=north]{$2.8$};
 \draw (3.7,0.1) node[anchor=south]{$\text{a}$} -- (3.7,-0.1) node[anchor=north]{$3.7$};
 \draw (5.3,0.1) node[anchor=south]{$\text{a}$} -- (5.3,-0.1) node[anchor=north]{$5.3$};
 \draw (4.9,0.1) node[anchor=south]{$\text{a}$} -- (4.9,-0.1) node[anchor=north]{$4.9$};
 \draw (6.0,0.1) node[anchor=south]{$\text{a}$} -- (6.0,-0.1) node[anchor=north]{$6.0$};

 \node at (2.8,-0.8) {$\{s_1\}$};
 \node at (3.7,-0.8) {$\{s_2\}$};
 \node at (4.5,-0.8) {\color{red}{\xmark}};
\end{tikzpicture}
\\\hline
\end{tabular}
 $\stackrel{\beta(s_2)=2}{\Longrightarrow}$&
\begin{tabular}{|c|}
 \hline
  \begin{tikzpicture} 
  \draw [thick, -stealth](-0.5,0)--(6.5,0) node [anchor=north]{$t$};
  \draw (0,0.1) -- (0,-0.1) node [anchor=north]{$0$};

 \draw (0.5,0.1) node[anchor=south]{$\text{a}$} -- (0.5,-0.1) node[anchor=north]{$0.5$};
 \draw (0.9,0.1) node[anchor=south]{$\text{a}$} -- (0.9,-0.1) node[anchor=north]{$0.9$};
 \draw (1.3,0.1) node[anchor=south]{$\text{b}$} -- (1.3,-0.1) node[anchor=north]{$1.3$};
 \draw (1.7,0.1) node[anchor=south]{$\text{b}$} -- (1.7,-0.1) node[anchor=north]{$1.7$};
 \draw (2.8,0.1) node[anchor=south]{$\text{a}$} -- (2.8,-0.1) node[anchor=north]{$2.8$};
 \draw (3.7,0.1) node[anchor=south]{$\text{a}$} -- (3.7,-0.1) node[anchor=north]{$3.7$};
 \draw (5.3,0.1) node[anchor=south]{$\text{a}$} -- (5.3,-0.1) node[anchor=north]{$5.3$};
 \draw (4.9,0.1) node[anchor=south]{$\text{a}$} -- (4.9,-0.1) node[anchor=north]{$4.9$};
 \draw (6.0,0.1) node[anchor=south]{$\text{a}$} -- (6.0,-0.1) node[anchor=north]{$6.0$};

 \draw (3.8,-0.5) rectangle (6.5,-1.0);
 \node at (5.15,-0.75) {$\mathcal{W}(L')$};
\end{tikzpicture}
\\\hline
\end{tabular}
 \\
 $\Longrightarrow$
&
\begin{tabular}{|c|}
 \hline
  \begin{tikzpicture} 
  \draw [thick, -stealth](-0.5,0)--(7.0,0) node [anchor=north]{$t$};
  \draw (0,0.1) -- (0,-0.1) node [anchor=north]{$0$};

 \draw (0.5,0.1) node[anchor=south]{$\text{a}$} -- (0.5,-0.1) node[anchor=north]{$0.5$};
 \draw (0.9,0.1) node[anchor=south]{$\text{a}$} -- (0.9,-0.1) node[anchor=north]{$0.9$};
 \draw (1.3,0.1) node[anchor=south]{$\text{b}$} -- (1.3,-0.1) node[anchor=north]{$1.3$};
 \draw (1.7,0.1) node[anchor=south]{$\text{b}$} -- (1.7,-0.1) node[anchor=north]{$1.7$};
 \draw (2.8,0.1) node[anchor=south]{$\text{a}$} -- (2.8,-0.1) node[anchor=north]{$2.8$};
 \draw (3.7,0.1) node[anchor=south]{$\text{a}$} -- (3.7,-0.1) node[anchor=north]{$3.7$};
 \draw (5.3,0.1) node[anchor=south]{$\text{a}$} -- (5.3,-0.1) node[anchor=north]{$5.3$};
 \draw (4.9,0.1) node[anchor=south]{$\text{a}$} -- (4.9,-0.1) node[anchor=north]{$4.9$};
 \draw (6.0,0.1) node[anchor=south]{$\text{a}$} -- (6.0,-0.1) node[anchor=north]{$6.0$};

 \node at (4.9,-0.8) {$\{s_1\}$};
 \node at (5.5,-0.8) {$\{s_2\}$};
 \node at (6.0,-0.8) {$\{s_3\}$};
 \node at (6.6,-0.8) {$\{s_4\}$};
 \node at (7.3,-0.8) {\color{dgreen}{\cmark}};
\end{tikzpicture}\\
\hline
\end{tabular}
\end{tabular}}
 \caption{Our FJS-type algorithm for pattern matching, for the example in Fig.~\ref{fig:input_timed_pattern_matching}}
 \label{fig:tpm_configs}
\end{figure}
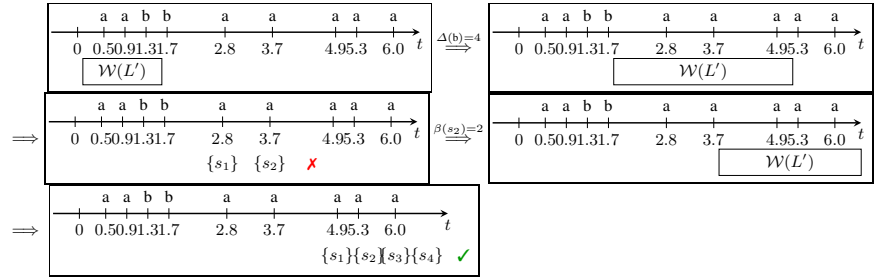
\begin{figure}[tbp]
\begin{minipage}{0.5\textwidth}
\[
 \SG(\mathcal{A}) =
  \scalebox{0.55}{
  \tikz[baseline=(s_0.base),auto,node distance=0.5cm]{
   \regionstate{s_0}{$s_0$}{$x=0$};
   \regionstate[right=of s_0]{s_1}{$s_1$}{$x=0$};
   \regionstate[right=of s_1]{s_2}{$s_2$}{$x=0$};
   \regionstate[right=of s_2]{s_3}{$s_3$}{$x>1$};
   \regionstate[right=of s_3,accepting]{s_4}{$s_4$}{$x>1$};
   \path[->](s_0) edge  node {$\text{a}$} (s_1);
   \path[->](s_1) edge  node {$\text{a}$} (s_2);
   \path[->](s_2) edge  node {$\text{a}$} (s_3);
   \path[->](s_3) edge  node {$\$$} (s_4);
   }}
\]
 \caption{The zone automaton $\SG(\mathcal{A})$ for $\mathcal{A}$ in Fig.~\ref{fig:input_timed_pattern_matching}}
 \label{fig:zoneAutomaton}
\end{minipage}
\hfill
\begin{minipage}{0.4\textwidth}
 \centering
\scalebox{0.8}{\begin{tabular}{c|l}
 &$ \left.
  \scalebox{0.70}{
  \tikz[baseline=(s_0.base),auto,node distance=0.5cm]{
   \regionstate{s_0}{$s_0$}{$x=0$};
   \regionstate[right=of s_0]{s_1}{$s_1$}{$x=0$};
   \regionstate[right=of s_1]{s_2}{$s_2$}{$x=0$};
   \path[->](s_0) edge  node {$\text{a}$} (s_1);
   \path[->](s_1) edge  node {$\text{a}$} (s_2);
   }}\right\} = L'_{s_2}$\\\hline
 \color{red}{\xmark}&
$\left.\quad
  \scalebox{0.70}{
  \tikz[baseline=(s_0.base),auto,node distance=0.5cm]{
   \node(any) at (3,0) {*};
   \regionstate[node distance=0.8cm,right=of any]{s_0}{$s_0$}{$x=0$};
   \regionstate[right=of s_0]{s_1}{$s_1$}{$x=0$};
   \regionstate[right=of s_1]{s_2}{$s_2$}{$x=0$};
   \regionstate[right=of s_2]{s_3}{$s_3$}{$x>1$};
   \path[->](s_0) edge  node {$\text{a}$} (s_1);
   \path[->](s_1) edge  node {$\text{a}$} (s_2);
   \path[->](s_2) edge  node {$\text{a}$} (s_3);
   }}\right\}  = L'$\\
 \color{dgreen}{\cmark}&
$\left.\quad
  \scalebox{0.70}{
  \tikz[baseline=(s_0.base),auto,node distance=0.5cm]{
   \node(any0) at (3,0) {*};
   \node(any) at (4.6,0) {*};
   \regionstate[node distance=0.85cm,right=of any]{s_0}{$s_0$}{$x=0$};
   \regionstate[right=of s_0]{s_1}{$s_1$}{$x=0$};
   \regionstate[right=of s_1]{s_2}{$s_2$}{$x=0$};
   \regionstate[right=of s_2]{s_3}{$s_3$}{$x>1$};
   \path[->](s_0) edge  node {$\text{a}$} (s_1);
   \path[->](s_1) edge  node {$\text{a}$} (s_2);
   \path[->](s_2) edge  node {$\text{a}$} (s_3);
   }}\right\}  = L'$\\
\end{tabular}}
 \caption{Table for $\beta(s_2)=2$}
 \label{fig:betaforFJSTimedPatternMatching}
\end{minipage}
\end{figure}

\section{Correctness of Our FJS-Type Timed Pattern Matching Algorithm}
\label{appendix:correctness}

\begin{mytheorem}[Correctness of $\Delta$ and $\beta$]\label{thm:correctness}
 Let $\Opt (n) = \min\{ i \in \Zp \mid \exists t \in [\tau_{n+i-1}, \tau_{n+i}), t' \in (t,\infty).\, (t,t') \in \mathcal{M} (w,\mathcal{A})\}$.
 For $n \in [1,|w|]$,
 we have both $\Opt (n) \geq \Delta (\overline{a}_{n+m-1})$ and
 $\Opt (n) \geq \max_{s \in S'}\beta (s)$
 where $\nu_0$ is the clock valuation assigning 0 for any $x \in C$,
 $n' = \max\{n' \in [1, |\str|] \mid \exists s_0 \in S_0,s \in S,
 \nu \in (\Rnn)^C.\, (s_0,\nu_0) \xrightarrow{w(n,n')} (s,\nu)\}$ and
 $S' = \{s \in S \mid \exists s_0 \in S_0,\nu \in (\Rnn)^C.\, (s_0,\nu_0) \xrightarrow{w(n,n')} (s,\nu)\}$.
\end{mytheorem}

\begin{proof}
 When $\Opt (i) > m$, both $\Opt (i) \geq \Delta (w(i+m-1))$ and
 $\Opt (i) \geq \max\{\beta (s) \mid(s,\rho,T) \in \Conf(i,j)\}$
 hold because for any 
 $a \in \Sigma$ and $s \in S$, we have
 $m + 1 \geq \Delta (a)$ and  $m + 1 \geq \beta (s)$.
 Assume $\Opt (i) \leq m$ in the following.
 Let $L_{-\$} (\mathcal{A})$ be $\{w (1,|w|-1) \mid w \in L (\mathcal{A})\}$.

The membership of a timed word segment leads the membership in the
approximated languages, as follows.
 \begin{align*}
  &\exists t \in [i + n - 1, i + n), t' \in (t,\infty).\, (t,t') \in
  \mathcal{M} (w,\mathcal{A})\\
  \iff& \exists t \in [i + n - 1, i + n), t' \in (t,\infty).\, w|_{(t,t')}
  \in L (\mathcal{A})\\
  \Rightarrow& \exists t \in [i + n - 1, i + n), t' \in (t,\infty), k
  \in [i + n - 1, |w|].\\ 
 &\qquad (w (i+n,k) - t) \circ (\$,t') \in L (\mathcal{A})\\
  \Rightarrow& \exists t \in [i + n - 1, i + n), k
  \in [i + n - 1, |w|].\,\\
  & \qquad (w (i+n,k) - t) \in L_{-\$}
  (\mathcal{A})\\
  \Rightarrow& \exists t \in [i + n - 1, i + n).\,
  (w (i+n,|w|) - t) \in L_{-\$}
  (\mathcal{A}) \cdot (\Sigma \times \Rp)^*\\
  \Rightarrow&
  (w (i,|w|) - \tau_i) \in
  (\Sigma \times \Rp)^{n} \cdot L_{-\$} (\mathcal{A}) \cdot
  (\Sigma \times \Rp)^*\\
  \Rightarrow&
  (w (i,|w|) - \tau_i) \in
  (\Sigma \times \Rp)^{n} \cdot \mathcal{W}(L') \cdot
  (\Sigma \times \Rp)^*\\
 \end{align*}
We have $\Opt (i) \geq \Delta (w (i + m-1))$ because of the follows.
 \begin{align*}
  &(w (i,|w|) - \tau_i) \in
  (\Sigma \times \Rp)^{n} \cdot \mathcal{W}(L') \cdot
  (\Sigma \times \Rp)^*\\
  \Rightarrow &(w (i,i+m-1) - \tau_i) \cdot (\Sigma \times \Rp)^* \cap\\
  & \qquad
  (\Sigma \times \Rp)^{n} \cdot \mathcal{W}(L') \cdot
  (\Sigma \times \Rp)^* \neq \emptyset\\
  \Rightarrow & (\Sigma \times \Rp)^m \cdot w (i+m-1) \cdot
  (\Sigma \times \Rp)^* \cap\\
  & \qquad
  (\Sigma \times \Rp)^{n} \cdot \mathcal{W}(L') \cdot
  (\Sigma \times \Rp)^* \neq \emptyset\\
 \end{align*}
Similarly, we have $\Opt (i) \geq \max\{\beta (s) \mid(s,\rho,T) \in
 \Conf(i,j)\}$ because of the follows.
 \begin{align*}
  &(w (i,|w|) - \tau_i) \in
  (\Sigma \times \Rp)^{n} \cdot \mathcal{W}(L') \cdot
  (\Sigma \times \Rp)^*\\
  \Rightarrow &\forall (s, \rho,T).\, 
  \mathcal{W}(L_s) \cdot (\Sigma \times \Rp)^* \cap
  (\Sigma \times \Rp)^{n} \cdot \mathcal{W}(L') \cdot
  (\Sigma \times \Rp)^*\\
 \end{align*}
 \qed
\end{proof}

\section{The pattern timed automaton in \textsc{Large Constraints}}

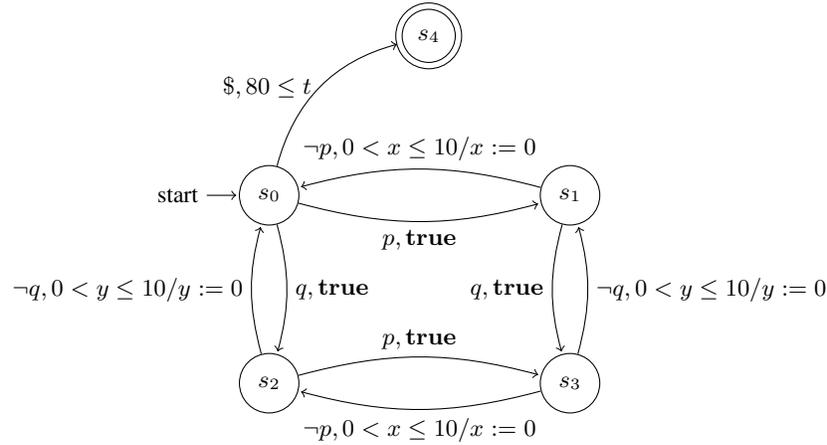
\begin{figure}[ht]
 \centering
 \begin{tikzpicture}[shorten >=1pt,node distance=3cm,on grid,auto] 
  \node[state,initial] (s_0)   {$s_0$}; 
  \node[state,node distance=4cm] (s_1) [right=of s_0] {$s_1$}; 
  \node[state,node distance=2.5cm] (s_2) [below=of s_0] {$s_2$};
  \node[state,node distance=2.5cm] (s_3) [below=of s_1] {$s_3$};
  \node[state,accepting] (s_4) [above right=of s_0] {$s_4$};
  \path[->] 
  (s_0) edge  [below,bend right=15,bend left=-15] node {$p,\mathbf{true}$} (s_1)
  (s_1) edge  [above,bend right=15,bend left=-15] node {$\neg p,0 < x \leq 10/x:=0$} (s_0)
  
  (s_0) edge  [right,bend right=-15,bend left=15] node {$q,\mathbf{true}$} (s_2)
  (s_2) edge  [left,bend right=-15,bend left=15] node {$\neg q,0 < y \leq 10/y:=0$} (s_0)

  (s_1) edge  [left,bend right=15,bend left=-15] node {$q,\mathbf{true}$} (s_3)
  (s_3) edge  [right,bend right=15,bend left=-15] node {$\neg q,0 < y \leq 10/y:=0$} (s_1)

  (s_2) edge  [above,bend right=-15,bend left=15] node {$p,\mathbf{true}$} (s_3)
  (s_3) edge  [below,bend right=-15,bend left=15] node {$\neg p,0 < x
  \leq 10/x:=0$} (s_2)
  (s_0) edge  [left,bend right=-30,bend left=30] node {$\$,80 \leq t$} (s_4);
 \end{tikzpicture}
 \caption{The pattern timed automaton in \textsc{Large Constraints}}
 \label{fig:case2_automaton}
\end{figure}

\section{Optimization of Preprocessing for Zone Abstraction and Skip Value Computation}
\label{appendix:preprocessing}
In Table~\ref{table:preprocessing} is how long our preprocessing takes for each of our benchmark problems. We see that our implementation is efficient in the preprocessing stage; this is largely due to our memorization technique in which we reuse parts of zone automata. 

For reference we also present results for \emph{region-based} abstraction~\cite{Alur1994}: though equivalent in terms of finiteness, zones give more efficient abstraction than regions. In our previous work~\cite{DBLP:conf/formats/WagaAH16} we used regions in place of zones, and that posed a bottleneck, as we can see in Table~\ref{table:preprocessing}.

\begin{table}[t]
\centering
\caption{The duration of preprocessing (ms). The timeout
is set to three minutes.}
\label{table:preprocessing}
\begin{tabular}{c||c|c|c}
&BM (region)&BM (zone) &FJS (zone)\\\hline
\textsc{Simple}&1.09e-03&2.53e-03&3.46e-03\\
\textsc{Large Constraints}&\color{red}{Timeout}&1.02e-03&1.05e-03\\
\textsc{Torque}&1.74e+00&4.25e-01&2.09e-01\\
\textsc{Settling}&1.69e+04&4.60e-03&5.40e-03\\
\textsc{Gear}&1.00e-03&4.47e-03&1.00e-03\\
\textsc{Accel}&\color{red}{Timeout}&2.95e-02&1.00e-03\\
\end{tabular}
\end{table}

\auxproof{
We measured the execution time of the preprocessing phase of
$\IMPL_{\ONLINE,\BMREGION}$, $\IMPL_{\ONLINE,\BMZONE}$, and
$\IMPL_{\ONLINE,\FJS}$ to evaluate the effectiveness of the zone-based
abstraction compared to the region-based abstraction; in the previous
work~\cite{DBLP:conf/formats/WagaAH16} that used the region-based
abstraction, the preprocessing phase was the bottleneck.  The timeout
is set to three minutes.  Table~\ref{table:preprocessing} shows the
results.

Even for the benchmarks that $\IMPL_{\ONLINE,\BMREGION}$ timeouts
(i.e., \textsc{Large Constraints}, \textsc{Settling}, and
\textsc{Accel}), $\IMPL_{\ONLINE,\BMZONE}$ and $\IMPL_{\ONLINE,\FJS}$
can handle the input (and runs reasonably fast).  The execution time
of the preprocessing phase of a region-based algorithms largely
depends on the maximum constant that appears as the guards of the
transition of the automaton, whereas that of zone-based algorithms
depends only on the structure of the automaton.
}


\auxproof{
\section{Notes for Authors}
\paragraph{Notations}
\begin{itemize}
 \item target string: $w$ in any problems
 \item pattern: $\pat$ in string matching, $\mathcal{A}$ or
       $L(\mathcal{A})$ in other problems
\end{itemize}

\paragraph{Suggestion}

\begin{itemize}
 \item Move the short overview of string matching algorithms in
       \S{}\ref{subsec:stringMatching} to the introduction.
 \item Remove or move \S{}\ref{subsec:zoneAutomata} to appendix and mention a paper in
       \S{}\ref{sec:timedFJS}.
       I do not think the detail of zone is important in this paper.
\end{itemize}
}
\end{document}